\documentclass[sigconf,screen,nonacm]{acmart}
\AtBeginDocument{%
  }

\setcopyright{acmlicensed}
\copyrightyear{2018}
\acmYear{2018}
\acmDOI{XXXXXXX.XXXXXXX}
\acmConference[Conference acronym 'XX]{}{June 03--05,
  2018}{Woodstock, NY}
\acmISBN{978-1-4503-XXXX-X/2018/06}

\usepackage{amsmath,amsthm,amsfonts,bbold,xspace,graphicx,float,mathtools}
\usepackage{braket,subcaption,xcolor,textcomp,pifont}
\usepackage{braket}
\usepackage[mathcal]{eucal} %Euler script
\usepackage{stmaryrd} %for double bracket
\usepackage{enumitem} %for specifying label in enum
\usepackage{hyperref}
\usepackage[nameinlink,capitalize]{cleveref}

\usepackage{tabularray}

\usepackage{tikz}
\usetikzlibrary{arrows, shapes.multipart}
\usepackage{algorithm}
\usepackage{algpseudocodex}
\usepackage{wrapfig}
\usepackage{centernot}

\makeatletter
\renewcommand{\ALG@beginalgorithmic}{\small}
\makeatother

\newfloat{algorithm}{t}{lop}

\AtEndPreamble{%
\crefname{equation}{}{}
\newtheorem{question}{Question}
\Crefname{question}{Question}{Questions}
\theoremstyle{acmdefinition}

\Crefname{condition}{Condition}{Conditions}
}

\newcommand{\reop}{\operatorname{Re}}
\newcommand{\imop}{\operatorname{Im}}
\newcommand{\midv}{\,\middle\vert\,}
\newcommand{\CI}{\mathrel{\perp\mspace{-10mu}\perp}}

\newcommand{\Obs}{\operatorname{Obs}}
\newcommand{\Sub}{\mathbf{Sub}}
\newcommand{\Obj}{\mathbf{Obj}}
\newcommand{\Objc}{\mathbf{Obj}_{\textup{c}}}
\newcommand{\Objq}{\mathbf{Obj}_{\textup{q}}}
\newcommand{\Rt}{\mathbf{Rt}}
\newcommand{\Rtc}{\mathbf{Rt}_{\textup{c}}}
\newcommand{\Rtq}{\mathbf{Rt}_{\textup{q}}}
\newcommand{\Atr}{\mathbf{Attr}}
\newcommand{\Rule}{\mathbf{Rule}}
\newcommand{\Req}{\mathbf{Req}}
\newcommand{\Auth}{\mathit{Auth}}
\newcommand{\PU}{\mathit{Post}}

\newcommand{\Int}{\mathbf{Int}}

\newcommand{\Macc}{M_{\textup{acc}}}
\newcommand{\Mcc}{M_{\textup{c}}}
\newcommand{\Mq}{M_{\textup{q}}}
\newcommand{\Me}{M_{\textup{e}}}

\newcommand{\SUBSYS}{\mathsf{SUBSYS}}
\newcommand{\GRP}{\mathsf{GRP}}
\newcommand{\ENT}{\mathsf{ENT}}

\newcommand{\Co}{\mathbb C}
\newcommand{\N}{\mathbb N}
\newcommand{\bbone}{\mathbb 1}
\newcommand{\Id}{\bbone}

\renewcommand{\Pr}{\mathop{\bf Pr\/}}                   
\newcommand{\E}{\mathop{\bf E\/}}
\newcommand{\Ex}{\mathop{\bf E\/}}

\DeclarePairedDelimiter\parens{\lparen}{\rparen}
\DeclarePairedDelimiter\abs{\lvert}{\rvert}
\DeclarePairedDelimiter\norm{\lVert}{\rVert}
\DeclarePairedDelimiter\floor{\lfloor}{\rfloor}

\DeclarePairedDelimiter\braces{\lbrace}{\rbrace}
\DeclarePairedDelimiter\bracks{\lbrack}{\rbrack}

\newcommand{\calA}{\mathcal{A}}

\newcommand{\calE}{\mathcal{E}}

\newcommand{\calH}{\mathcal{H}}

\newcommand{\calP}{\mathcal{P}}

\newcommand{\calS}{\mathcal{S}}

\newcommand{\calX}{\mathcal{X}}
\newcommand{\calY}{\mathcal{Y}}

\begin{document}

\title{Access Control Threatened by Quantum Entanglement}

\author{Zhicheng Zhang}
\orcid{0000-0002-7436-0426}
\affiliation{%
  \institution{University of Technology Sydney}
  \city{Sydney}
  \country{Australia}
}
\email{zhicheng.zhang@student.uts.edu.au}

\author{Mingsheng Ying}
\orcid{0000-0003-4847-702X}
\affiliation{%
  \institution{University of Technology Sydney}
  \city{Sydney}
  \country{Australia}
}
\email{mingsheng.ying@uts.edu.au}

\renewcommand{\shortauthors}{}

\begin{abstract}
    Access control is a cornerstone of computer security that prevents unauthorised access to resources.
	In this paper, we study access control in quantum computer systems.
	We present the first explicit scenario of a \textit{security breach} when a classically secure access control system 
	is straightforwardly adapted to the quantum setting.
	The breach is ultimately due to that quantum mechanics allows the phenomenon of entanglement
	and violates Mermin inequality,
	a multi-party variant of the celebrated Bell inequality.
	This reveals a \textit{threat from quantum entanglement} to access control
	if existing computer systems integrate with quantum computing.
	To \textit{protect} against such threat, we propose several new models of quantum access control, 
	and rigorously analyse their security, flexibility and efficiency.
\end{abstract}






% PUT PAGE NUMBER, REMEMBER TO DELETE
\settopmatter{printfolios=true}

\maketitle

\section{Introduction}
\label{sec:introduction}

A fundamental issue in computer security 
is how to control access to resources in computer systems.
Initially proposed with the seminal concept of explicitly managing \textit{rights} granted to a \textit{subject} to access an \textit{object},
the access matrix model~\cite{Lampson69,Lampson74,GD71,Denning71,HRU76} has served as the standard core model of access control.
Over time, according to different security requirements, 
it has evolved into various sophisticated access control models,
such as discretionary~\cite{Saltzer74,DRKJ85}, mandatory~\cite{BLP73,BLP76,Biba77,CW87,Denning76} and role-based access control~\cite{FK92,SCFY96,FSGKC01,BBF00}, along with their further extensions,
which are now widely deployed in modern computer systems.

On the other hand, the rapid emerging of quantum computing technology has raised increasing attention to 
the security in quantum computer systems.
For example, to protect user privacy against untrusted quantum computing servers,
numerous efforts have been devoted to delegated quantum computation (and further, blind quantum computation)~\cite{Childs05,ABE08,BFK09,DKL12,Morimae12,MF12,GMMR13,MPF13,MF13,MK13,Morimae14,MDK15,FK17,Fitz17}, as well as
quantum computer trusted execution environment~\cite{TXD+23,TXD+b23,TJ24,TDX+24} in the recent years.
Protecting security against hardware and side channel attacks in quantum computers has also attracted much attention~\cite{MDS22,XES23,XCMS23}.
The first attempt to study access control in quantum systems 
was made by~\cite{YFY13} through quantum information flow security.
In the context of quantum internet, specific control of entanglement accessibility was also studied~\cite{GI19}.

Still, a significant question --- whether 
the access control security will be threatened by
integrating quantum computing into existing computer systems --- remains open.
More precisely:
\begin{question}
	\label{qus:main}
		Suppose you are a user of a classical system, which earns your trust by providing a proof  
		that its access control mechanism can protect your private information from being leaked to other users.
		One day, you are notified that the system will be upgraded by integrating new quantum computing services
		and the access control remains unchanged.
		Should you still trust the security of the system?
\end{question}
This question is becoming increasingly crucial as IBM Quantum and other researchers are actively exploring quantum-centric supercomputing as the next generation of classical-quantum hybrids. This approach integrates traditional high-performance computers with quantum computing~\cite{Gambetta22a,Gambetta22b,AAB24,MCG24,PC24}. Definitely, we hope that the hybrid systems remain secure.

However, the answer to \Cref{qus:main} is probably \textit{no}.
The first aim of this paper is to show that a security breach can occur in this case, through an explicit scenario.
This highlights the necessity to develop new models of access control for quantum computer systems, which is the second aim of this paper.
%we reveal the threat from quantum entanglement to access control if existing access control systems integrate with quantum computing.
%Then, we propose several new models of quantum access control to protect against such threat.

\subsection{Contributions}
\label{sub:contributions}

More concretely, the contribution of this work is twofold:
\begin{itemize}
	\item 
		\textit{Reveal of threats from quantum entanglement to access control}
		if existing computer systems integrate with quantum computing (\Cref{sec:scenario_threat_from_quantum_entanglement}).

		For the first time, an explicit scenario of a security breach is presented 
		when a classical secure access control system is straightforwardly adapted to the quantum setting.
		The ultimate cause of this breach is quantum entanglement,
		a fundamental  phenomenon that distinguishes quantum mechanics from classical mechanics.
		A key tool in the proof of insecurity is Mermin inequality~\cite{Mermin90}, a multi-party variant of the celebrated Bell inequality~\cite{Bell64,CHSH69,FC72,ADR82,GHSZ90},
		which will be violated by entanglement even without direct communication.
		Since the entanglement is believed to be the source of quantum advantages~\cite{JL03} for many quantum algorithms~\cite{Shor94,Grover96,Lloyd96,HHL09},
		our scenario highlights the importance of developing models of quantum access control against threats from entanglement.
	\item
		\textit{Design of models of quantum access control}, including
		subsystem control, group control and entanglement control (\Cref{sec:protection_access_control_in_quantum_computing}).

		These models allow explicit control of multi-object quantum operations or entanglement.
		We rigorously analyse their (a) \textbf{security} against the threat in our scenario;
		(b) \textbf{flexibility} regarding the granularity of specifying the	access control;
		and (c) \textbf{efficiency} regarding the space and time complexity for implementation.
\end{itemize}

\section{Background}
\label{sec:background}

\subsection{Access Control}
\label{sub:access_control}

Let us start with the framework of access control considered in this paper,
which adopts ideas and concepts from the modern access (usage) control framework UCON~\cite{SP03,PS04,ZPSP05}.

An access control system involves the following components.
\begin{itemize}
	\item 
		A set $\Sub$ of \textit{subjects},
		a set $\Obj$ of \textit{objects},
		and a set $\Rt$ of \textit{rights}.

		A subject can access an object by exercising a right.
		Examples of subjects include users, processes and applications.
		Examples of objects include files, directories, registers, pages and segments.
		Examples of rights include \texttt{read}, \texttt{write} and \texttt{execute}.

		In this paper, we restrict our subjects to be users, objects to be (classical and quantum) registers,
		and rights to be abilities to perform certain operations on registers.
        Unless explicitly specified, classical registers are initialised to $0$,
        and quantum registers are initialised to $\ket{0}$.
	\item
		A set $\Atr$ of \textit{attributes}.

		An attribute is a (partial) function with domain $\Sub$, $\Obj$ or $\Sub\times \Obj$.
		Attributes can be used by the system to enforce access control rules.

		Standard attributes in the literature~\cite{PS04,HFK13} are only functions with domain $\Sub$ and $\Obj$,
		known as subject attributes and object attributes, respectively.
		Here, we slightly extend this notion for convenience of presentation.
	\item
		A set $\Rule$ of \textit{rules}.

		A rule describes how the system handles an access \textit{request} of the form $\parens*{s,o,r}\in \Sub\times \Obj\times \Rt$,
		which means that subject $s$ requests to exercise right $r$ on object $o$.
		
		In this paper, we focus on \textit{authorisation rule} that describes whether to grant or deny an access request,
		and \textit{post-update rule} that describes how to update attributes after authorising a request.
		They will be explained in detail later.
\end{itemize}
We use a $5$-tuple $\calA=\parens*{\Sub,\Obj,\Rt,\Atr,\Rule}$ to denote an access control system,
and $\Req=\Sub\times \Obj\times \Rt$ to denote the set of requests.
In context without ambiguity, we simply say system instead of access control system,
and request instead of access request.

%Mention properties of the reference monitor.

The most basic rule in access control is the authorisation rule.
Upon receiving a request, the system will determine whether to grant or deny the access
according to a function $\Auth$.

\begin{definition}[Authorisation]
	An authorisation rule is a function $\Auth:\Req\rightarrow \braces*{\mathit{true},\mathit{false}}$.
	%An action $(s,o,r)\in \Act$ is authorised by access matrix $M$ iff $M\parens*{s,o,r}=1$.
\end{definition}

A widely used attribute is the access matrix, initially proposed in the seminal paper~\cite{Lampson74}
and later refined by~\cite{GD71,Denning71,HRU76}.

\begin{definition}[Access matrix]
	\label{def:access-matrix}
	An access matrix is a function $\Macc:\Sub\times \Obj\rightarrow \calP\parens*{\Rt}$.
\end{definition}

The simplest authorisation rule based on the access matrix is defining
$\Auth\parens*{s,o,r}\equiv r\in \Macc\bracks*{s,o}$.
In practice, the access matrix is usually sparse 
and can be implemented via access control lists (ACL),
capability lists~\cite{Lampson74}, or other data structures to reduce the space and time complexity~\cite{SS94}.
However, for illustration, we still focus on the access matrix.

%A example of access control system using the access matrix for authorisation is as follows.

Another rule we will use (in particular, in \Cref{sub:control_of_entanglement}) is the post-update rule.
After a request is authorised, 
and before the next request is handled,
several post-update operations can be performed on attributes
according to the partial function $\PU$, for future authorisation decisions.

\begin{definition}[Post-update]
        \label{def:post-update}
	A post-update rule is a partial function $\PU$ such that
	for request $\parens*{s,o,r}\in \Req$ and attribute $f\in \Atr$,
	$\PU\parens*{s,o,r}\parens*{f}=f'$ for some $f'$ of the same function type as $f$.
\end{definition}

Intuitively, after authorising an request $\parens*{s,o,r}$,
rule $\PU$ updates $f$ to $f'$. 
For example, for an attribute $f:\Obj\rightarrow \braces*{0,1}$,
a possible post-update rule can be $\PU\parens*{s,o,r}\parens*{f}\equiv f'_u$,
where $f'_u\bracks*{u}=1-f\bracks*{u}$ and $f'_u\bracks*{o}=f\bracks*{o}$ for $o\neq u$.
This $\PU$ means whenever a request $\parens*{s,o,r}$ is authorised,
$f\bracks*{u}$ is updated to be $1-f\bracks*{u}$ and other $f\bracks*{o}$ remain unchanged.

%\begin{definition}[Action and authorisation]
%	Given the access matrix $M$,
%	let $\mathit{allow}:\mathbf{Act}\rightarrow \braces*{0,1}$ be an authorisation function
%	that decides whether an action is allowed.
%	At this point, we define 
%	\begin{equation*}
%		\mathit{allow}\parens*{s,c,X}=
%		\begin{cases}
%			1, & c\in M\bracks*{s,X},\\
%			0, & c\notin M\bracks*{s,X}.
%		\end{cases}
%	\end{equation*}
%\end{definition}

\subsection{Execution Model}
\label{sub:execution_model}

Next we describe the execution of an access control system.
Since there are multiple subjects in the system,
the execution is intrinsically concurrent.
%For our purpose, we assume the system takes one atomic action at each time unit.
For our purpose, we assume the requests in the system are atomic.
During an execution, the system receives a sequence of requests from subjects and enforces the access control rules accordingly.
To describe the non-deterministic ordering of requests made by different subjects,
we use the notion of a scheduler (like in e.g., \cite{Puterman14,Rabin80}).

\begin{definition}[Scheduler]
	A scheduler of the system is a function $S:\bigcup_{k=0}^{\infty} \Req^k \to \Sub$.
\end{definition}

Intuitively, given any finite sequence of access requests, 
the scheduler $S$ determines the next subject $s$ to make a request.
Note that a scheduler can be an adversary:
if we want to prove a safety property that something bad (e.g., security breach) never happens in a system,
we need to consider it against all schedulers.

To describe valid sequences of requests under a scheduler, we introduce the notion of a history.

\begin{definition}[History]
	\label{def:history}
	Given a scheduler $S$ of the system,
	a history is a (finite or infinite) sequence of access requests %authorised by the access matrix;
	%specifically controlled by access matrix $M(t)$ is a sequence of actions
	$\alpha=\alpha(0),\alpha(1),\ldots $ such that
	for all $t\in \N$,
	if $\alpha(t)=\parens*{s,o,r}$
	then $s=S\parens*{\alpha\parens*{0},\ldots, \alpha\parens*{t-1}}$.
	Further, a history $\alpha$ is said to be \textit{authorised} if $\Auth\parens*{\alpha\parens*{t}}=\mathit{true}$ for all $t\in \N$.
\end{definition}

%It is easy to see that any prefix of an authorised history is also authorised.

The scheduler $S$ alone does not fully determines the history of an execution.
While it determines the next subject $s$ to make a request $\parens*{s,o,r}$, the object $o$
and the right $r$ in this request are determined by the behaviour of the subject $s$,
which is specified by a program $P_s$ (or any other computational model).
Let us collect all programs $P_s$ for $s\in \Sub$ and the initial state of objects into a program $P$.
Then, we can use $\parens*{S,P}$ to denote an execution of the system.

Each execution $\parens*{S,P}$ generates a history (or a probabilistic distribution over histories,
if $P$ is probabilistic).
The actual generation is determined by the specific programming language and
the explicit semantics of the program and requests.
%and the initial state of objects and attributes in the system.
For example, consider a system with $\Sub=\braces*{s}$ and $\Obj=\braces*{o}$.
Suppose program $P_s\equiv o:=o+1$, and $o$ is initialised to $0$.
In this case, if $\Rt=\braces*{\texttt{read},\texttt{write}}$,
then the history generated could be $\parens*{s,o,\texttt{read}},\parens*{s,o,\texttt{write}}$;
if $\Rt=\braces*{\texttt{inc}}$, where $\texttt{inc}$ means the ability to increment the value of the register by $1$,
then the history generated could be $\parens*{s,o,\texttt{inc}}$.

%In the analysis, often we will fix a scheduler $S$.
%If the subjects in the system are deterministic,
%then the system only yields a deterministic history.
%If the subjects are probabilistic,
%then the system can yield a probability distribution over histories.
%Given a scheduler $S$, how the programs of all subjects
%generate the history (or the distribution over histories) of requests

For simplicity, we do not bother formalising such generation,
because our focus is the access control system.
Nevertheless, we can define the equivalence between two systems 
with respect to authorised histories (see \Cref{def:history}) as follows.

\begin{definition}[Equivalent systems]
	\label{def:equivalent-sys}
	Two systems $\calA$ and $\calA'$ are said to be equivalent,
	denoted by $\calA\simeq \calA'$,
	if for any program $P$ of concern and any scheduler $S$:
	\begin{itemize}
		\item 
			$\parens*{S,P}$ can generate (valid) histories in both $\calA$ and $\calA'$; and
		\item
			The histories generated by $\parens*{S,P}$ in $\calA$ are authorised 
			iff the histories generated by $\parens*{S,P}$ in $\calA'$ are authorised.
	\end{itemize}
\end{definition}

%It is easy to see that the equivalence between two systems implies that 
%they have the same $\Sub$ and $\Obj$.

%In this paper, when we compare two systems,
%the set of programs of concern are independent of $\Rt,\Atr$ and $\Rule$.

An access control model is a family of access control systems.
An important metric to evaluate an access control model is its \textit{flexibility}.
While in general the flexibility cannot be characterised by a quantity,
we can compare the flexibility of two models by the following definition.

\begin{definition}[Flexibility]
	\label{def:flexibility}
	An access control model $\mathsf{M}$ is said to be less flexible than another $\mathsf{M}'$,
	denoted by $\mathsf{M}\leq \mathsf{M}'$,
	if for any system $\calA\in\mathsf{M}$,
	there exists a system $\calA'\in\mathsf{M}'$ such that $\calA\simeq \calA'$.
	Further, $\mathsf{M}$ is said to be strictly less flexible than $\mathsf{M}'$,
	denoted by $\mathsf{M}< \mathsf{M}'$,
	if $\mathsf{M}\leq \mathsf{M}'$ and $\mathsf{M}'\not\leq \mathsf{M}$.
\end{definition}

%An request $\alpha\parens*{t}$ at time $t$ takes the state of the system at time $t$
%to the state at time $t+1$.
%Also, as each subject has local memory,
%We do not bother formalising the state of the system and the effect of an action.
%In our example, these concepts will be explicit and clear.
%In particular, $\alpha\parens*{t}=(u,o)$ represents that at time $t\in \N$,
%the subject $u$ performs an atomic operation $o$.
%At any time $t\in \N$,
%the effect of $\alpha\parens*{t}$ is 
%changing the state 
%controlled by the access matrix $M(t')$,
%changing the state of the system from $\sigma\parens*{t}$ to $\sigma\parens*{t'}$.

%\begin{definition}[Authorised history]
%	Given an access control system,
%	a history $\alpha\parens*{t}$ is authorised if $\forall t\in \N, \Auth\parens*{\alpha\parens*{t}}=\mathit{true}$.
%\end{definition}

\subsection{Quantum Computing}
\label{sub:quantum_computing}

Now we briefly introduce quantum computing.
The readers are referred to \cite{NC10} for a more thorough introduction.

A qubit is the basic unit of information in quantum computing, 
compared to its classical counterpart bit.
The state of a qubit lives in the Hilbert space $\calH_{\mathbf{Bit}}=\Co^2$,
and can be represented by a complex vector $\alpha\ket{0}+\beta\ket{1}$
with $\abs*{\alpha}^2+\abs*{\beta}^2=1$,
a superposition of the computational basis states $\ket{0}$ and $\ket{1}$.
A quantum register consists of a set of qubits.
The state of a quantum register composed of $n$ qubits 
can be presented by $\sum_{x\in \braces*{0,1}^n}\alpha_x\ket{x}$ 
with $\sum_{x\in \braces*{0,1}^n}\abs*{\alpha_x}=1$,
and lives in the Hilbert space $\calH_{\mathbf{Bit}}^{\otimes n}$.
Quantum superposition leads to the phenomenon of \textit{quantum entanglement}:
state $\ket{\psi}$ is entangled iff it cannot be represented as a product $\ket{\psi_1}\otimes\ket{\psi_2}$.
For example, the simplest entangled state is an EPR state $\ket{+}_{AB}=\frac{1}{2}\parens*{\ket{0}_A\ket{0}_B+\ket{1}_A\ket{1}_B}$,
where we use the subscripts $A$ and $B$ to denote two qubits.

In quantum computing, there are two basic types of quantum operations.
The first is unitary gate.
After applying a unitary gate $U$ (with $UU^\dagger =U^\dagger U =\Id$), 
a quantum state $\ket{\psi}$ becomes $U\ket{\psi}$.
Typical one-qubit unitary gates include the three Pauli gates 
$X=\begin{bsmallmatrix}0&1\\ 1&0\end{bsmallmatrix},Y=\begin{bsmallmatrix}0&-i\\i&0\end{bsmallmatrix},
Z=\begin{bsmallmatrix}1&0\\0&-1\end{bsmallmatrix}$,
the Hadamard $H=\frac{1}{\sqrt{2}}\begin{bsmallmatrix}1&1\\1&-1\end{bsmallmatrix}$ gate, 
the $S=\begin{bsmallmatrix}1&0\\0&i\end{bsmallmatrix}$ gate and the $T=\begin{bsmallmatrix}1&0\\0&e^{-i\pi/4}\end{bsmallmatrix}$ gate.
Typical two-qubit unitary gates include the $\mathit{CNOT}=\ket{0}\!\bra{0}\otimes \Id+\ket{1}\!\bra{1}\otimes X$ gate.
$\mathit{SWAP}=\sum_{x,y}\ket{xy}\!\bra{yx}$ gate is also useful.

The second type of quantum operations is measurement.
A measurement can be specified by a set of Kraus operators $M=\braces*{M_m}_{m}$ with $\sum_{m}M_m^\dagger M_m=\Id$.
After applying the measurement $M$, a quantum state $\ket{\psi}$ becomes 
$M_m\ket{\psi}/\norm{M_m\ket{\psi}}$ and yields classical outcome $m$ with probability $\norm*{M_m\ket{\psi}}^2$.
Typical measurements include the computational basis measurement with $M_m=\ket{m}\!\bra{m}$.
A measurement is \textit{complete} if the range of $m$ is equal to the dimension of the state being measured.

\section{Scenario: Threat from Quantum Entanglement}
\label{sec:scenario_threat_from_quantum_entanglement}

In this section, to answer \Cref{qus:main} in the introduction,
we reveal a threat from quantum entanglement
by presenting an explicit scenario of a security breach when a classically secure access control system is straightforwardly adapted to the quantum setting. As computer security usually concerns the worst case, the threat shows the inadequacy of existing access control models for quantum computer security.
In \Cref{sub:problem_setting},
a classical access control system consisting of multiple users 
is specified using the notations in \Cref{sub:access_control}.
This system is proven to be secure in the classical case in \Cref{sub:security_protected_in_the_classical_case}.
Then, we prove it becomes no longer secure after it is straightforwardly adapted to the quantum case in \Cref{sub:security_breach_in_the_quantum_case}.

\subsection{Problem Setting}
\label{sub:problem_setting}

%The following problem setting formalises the question raised previously in \Cref{sec:introduction}.
%Consider a system of $n+2$ users $v,u,w_1,\ldots,w_n$.
%Suppose that you are the user $u$ who is concerned about the security of your private information when you use the system.
%%$u$ is a new user who is concerned about its security;
%$v$ is a system user whose behaviour is fixed,
%and $w_1,\ldots,w_n$ are other users whose behaviours are arbitrary.
%
%In the following, we describe a system in which the security is provable in the classical setting,
%but the threat from quantum entanglement can be explicitly seen.

	%\label{ex:main-example}
Let us consider a system $\calS=\parens*{\Sub,\Obj,\Rt,\Atr,\Rule}$
with
\begin{itemize}
    \item 
    $\Sub=\braces*{u,v,w_1,\ldots,w_n}$,
    \item
    $\Obj=\braces*{A,B,C_1,\ldots,C_n,\Macc}$,
    \item 
    $\Rt=\braces*{\texttt{read},\texttt{write},\texttt{flip},\texttt{all}}$,
    \item 
    $\Atr=\braces*{\Macc,L}$, and
    \item
    $\Rule=\braces*{\Auth}$.
\end{itemize}
Here, $L:\Sub \rightarrow \Int$ with $\Int$ being the set of (bounded) integers, and $\Auth\parens*{s,o,r}\equiv r\in \Macc\bracks*{s,o}$.

The ingredients of this system are explained as follows.
\begin{itemize}
    \item 
        In $\Sub$: $u,v,w_1,\ldots,w_n$ are all users.
    \item
        In $\Obj$: $A,B$ are bit registers and $C_1,\ldots, C_n$ are integer registers.
        Slightly abusing the notation, $\Macc$ represents an integer register\footnote{
        Here, using an integer register to store the whole matrix $\Macc$ is solely for simplifying the presentation of results in \Cref{sec:scenario_threat_from_quantum_entanglement}.
        In practice and later in \Cref{sec:protection_access_control_in_quantum_computing},
        we actually use multiple register (or memory locations) to store a matrix (that represents an attribute),
        where each register (or location) can store an entry of the matrix.
        } storing the access matrix $\Macc$.
        %All objects except the access matrix $M$ are initialised to $0$ at $t=0$.
    \item
        %$\mathbf{Rt}=\braces*{\texttt{read},\texttt{write},\texttt{flip}}$ consists of operations on registers.
        In $\Rt$: $\texttt{read}$ and $\texttt{write}$ correspond to standard read and write operations.
        Exercising $\texttt{flip}$ means changing every bit $0$ to $1$ and $1$ to $0$ in a register.
        The right $\texttt{all}$ means full access, allowing to perform any operations.
    \item
        The $\Atr$ consists of only two elements:
        (i) the access matrix $\Macc$ in \Cref{def:access-matrix}; and
        (ii) an attribute $L:\Sub\rightarrow \Int$.
        Here, for each user $s\in \Sub$,
        $L\bracks*{s}$ denotes the local memory of $s$,
        used to store temporary results for exercising rights $\texttt{read}$ and $\texttt{write}$.\footnote{
            In the classical literature, the local memory is often not explicitly stated as an attribute.
            In this paper, we include $L$ as an attribute for the following two reasons:
            $L$ is useful in the statement and analysis of system security (see \Cref{thm:classical-security});
            and whether $L$ is classical or quantum in a system with quantum objects needs to be explicitly specified 
            (see \Cref{sub:security_breach_in_the_quantum_case,sec:protection_access_control_in_quantum_computing}).
        }
        Only $s$ can access $L\bracks*{s}$.
        It should be noticed that $L$ is not in $\Obj$ and thus not guarded by the access control.
        %It is important to note the implicit assumption that every subject has a local memory 
        %(composed of registers):
        %when a \texttt{read} is executed, the result will be stored into the local memory;
        %when a \texttt{write} is executed, prior results stored in the local memory can be used.
        %These local memories are not listed in $\mathbf{Obj}$ for simplicity.
        %Each subject has full access to its local memory, 
        %on which operations are not controlled by the access matrix and hence assumed to take no time
        %(compared to actions in $\Req$).
        %There are some subtleties about local memories when we generalise the access matrix to the quantum setting.
\end{itemize}

The behaviour of $v$ is fixed and shown as a program in \Cref{fig:userv}.
We should notice that it is actually a probabilistic program,
as in Line $2$, $v$ samples from a random distribution.
Consequently, the security we prove to be protected in this system later in \Cref{sub:security_protected_in_the_classical_case}
is also probabilistic.
%where we use use $\texttt{*}=\Rt$ to denote full access to an object.
We explain what accesses are allowed when $\Macc=M_0,M_1,M_2$ in \Cref{fig:userv}, respectively:
\begin{itemize}
    \item 
        $\Macc=M_0$:
        user $u$ can write one bit of secret information into $A$.
        Other users $w_1,\ldots,w_n$ can access $C_1,\ldots,C_n$,
        through which they can communicate and devise some strategy in an attempt to learn the secret of $u$ later.
    \item
        $\Macc=M_1$:
        user $v$ can read the secret of $u$ from $A$ and access $B$, $C_1,\ldots,C_n$.
        For each $j\in [n]$, user $w_j$ can only access $C_j$ and flip $B$.
        These $w_j$ cannot communicate with each other,
        but they can exploit any pre-determined strategy.
    \item
        $\Macc=M_2$:
        for each $j\in [n]$,
        $w_j$ can access $C_j$ and read $B$.
        These $w_j$ still cannot communicate with each other.
\end{itemize}

\begin{figure}
    \centering
      \includegraphics[width=0.4\textwidth]{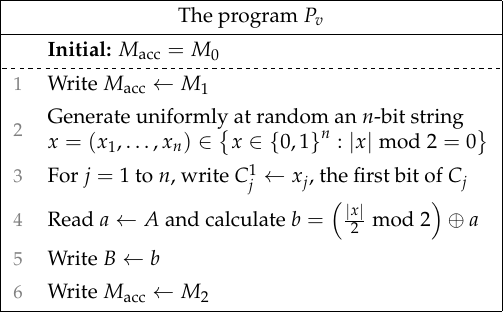}
    \caption{
        %The left column shows the fixed (informal) program of user $v$, and the right column records the time
        %(empty if unconcerned) of when $v$ make requests. 
        The program $P_v$ that describes the behaviour of user $v$.
        Here, matrices $M_0$, $M_1$ and $M_2$ are
    shown in~\Cref{fig:acc-matrix-1,fig:acc-matrix-2,fig:acc-matrix-3}, respectively.}
    \label{fig:userv}
\end{figure}

\begin{figure}[t]
    \centering
      \includegraphics[width=0.35\textwidth]{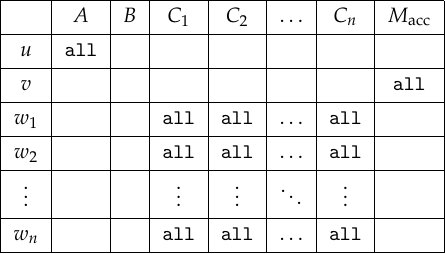}
    \caption{Matrix $M_0$. }
    \label{fig:acc-matrix-1}
\end{figure}

\begin{figure}[t]
    \centering
      \includegraphics[width=0.385\textwidth]{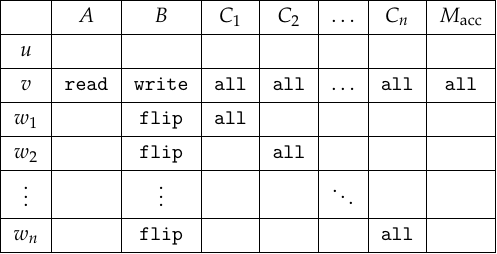}
    \caption{Matrix $M_1$. }
    \label{fig:acc-matrix-2}
\end{figure}

\begin{figure}[t]
    \centering
      \includegraphics[width=0.36\textwidth]{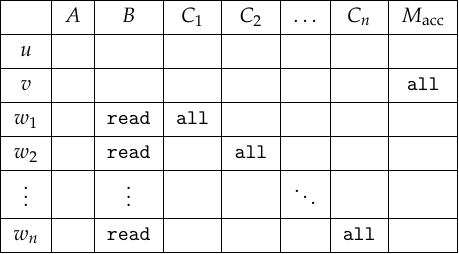}
    \caption{Matrix $M_2$. }
    \label{fig:acc-matrix-3}
\end{figure}

Finally, to correspond with \Cref{qus:main}, 
we can think of $u$ as the user concerned about the security,
$v$ as a system user with trusted and fixed behaviour,
and $w_1,\ldots,w_n$ as other users of the system.
%Finally, for the system in \Cref{ex:main-example},
Our security policy is to prevent the secret information of user $u$
from leaking to other users $w_1,\ldots,w_n$.

%In \Cref{sub:security_protected_in_the_classical_case},
%we will show such security is protected in the classical case (with probability close to $1$);
%and in \Cref{sub:security_breach_in_the_quantum_case},
%we will show there is a security breach when the system is directly lifted to the quantum case.

%The behaviours of all subjects are controlled by the (dynamic) access matrix.

\subsection{Security Protected in the Classical Case}
\label{sub:security_protected_in_the_classical_case}

If the whole system described in \Cref{sub:problem_setting} is classical, 
then we can rigorously prove that the amount of information from $u$ leaked to any other user $w_j$ is exponentially small in $n$.
This proof can assures user $u$ that $u$ can safely write private information into the system,
without (significantly) leaking it to other users $w_1,\ldots,w_n$.
%Therefore, the new user $u$ can trust the security of the system.
As a notation convention, for a register $X$, we use $X\parens*{t}$ to represent its value at time $t$.

\begin{theorem}[Security protected in the classical case]
	\label{thm:classical-security}
	Let $n\geq 5\in \N$.
	If all objects in the system described in \Cref{sub:problem_setting} are classical,
	then the secret information of user $u$ can only leak with negligible probability.
	That is, for any execution $\parens*{S,P}$ with $P_v$ described in \Cref{fig:userv},
	any time $t_u,t_w\in \N$ and any $j\in [n]$, the mutual information
	\begin{equation}
		\label{eq:sec-thm-main}
		I\parens*{A(t_u);\Obs\parens*{w_j,t_w}}\leq 2^{-(n-7)/2},
	\end{equation}
	where
	$\Obs\parens*{w_j,t}:=\braces*{o\in \mathbf{Obj}:\textup{\texttt{read}}\in \Macc\bracks*{w_j,o}\parens*{t}}\cup \braces*{L\bracks*{w_j}(t)}$
	is what $w_j$ can observe at time $t$.
\end{theorem}

Intuitively, even for a small system with approximately $100$ users,
any user $w_j$ can only learn about $10^{-14}$ bits of secret information from $u$,
an amount that is practically negligible.
The proof of \Cref{thm:classical-security} essentially relies the following variant of Mermin inequality~\cite{Mermin90}.

\begin{lemma}[A variant of Mermin inequality~\cite{Mermin90}]
	\label{lmm:Mermin}
	Let $n\in \N$ be a fixed natural number.
	Let $\calX_{b}:=\braces*{x\in \braces*{0,1}^n: \abs*{x}\bmod 2 =b}$, where $b\in \braces*{0,1}$.
	Let $\calY=\braces*{0,1}^n$.
	For any fixed $b\in \braces*{0,1}$,
	consider random variable $X=X_1,\ldots,X_n$ chosen uniformly at random from $\calX_{b}$,
	any random variable $Y=Y_1,\ldots,Y_n$ in $\calY$, and any random variable $\Lambda=\Lambda_1,\ldots,\Lambda_n$ independent of $X$ such that 
	\begin{equation*}
		\Pr\bracks*{Y=y\mid X=x,\Lambda=\lambda}=\prod_{j=1}^n \Pr\bracks*{Y_j=x_j\midv X_j=x_j,\Lambda_j=\lambda_j},
	\end{equation*}
	Then we have
	\begin{equation}
		\abs*{\E\bracks*{\parens*{-1}^{\abs*{X}/2+\abs*{Y}+b/2}}}\leq 2^{-n/2+1}.
		\label{eq:Mermin-ex}
	\end{equation}
\end{lemma}

The original Mermin inequality in~\cite{Mermin90} is the special case of $b=0$ in \Cref{lmm:Mermin}.
Mermin inequality extends the celebrated Bell inequality~\cite{Bell64,CHSH69,FC72,ADR82,GHSZ90} to the $n$-party case
and reveals the fundamental difference between classical and quantum mechanics.
%{\color{blue}the above variant only slightly differs from the original one.}

For readability, 
we only provide a proof sketch of \Cref{thm:classical-security} below.
The full proof is rather tedious (though complicated) and deferred to \Cref{sub:proof_of_classical-security}.

\begin{proof}[Proof sketch of \Cref{thm:classical-security}]
	Intuitively, within the system described in \Cref{sub:problem_setting},
	the ``best possible'' strategy for users $w_j$ to learn the secret information of $u$ is
	learning the value $\frac{\abs*{x}}{2}\bmod 2$ in \Cref{fig:userv}
	and then taking the $\oplus$ operation with $b$ in \Cref{fig:userv} to exactly recover $a$.
	However, the behaviours of all $w_j$ are constrained by the access matrix $\Macc$,
	and this strategy turns out to only work with negligible probability,
	essentially due to the variant of Mermin inequality in \Cref{lmm:Mermin}.
	%so $w_j$ need to collaborate (e.g., using pre-shared randomness) and exploit their restricted rights.

	Now we explain how to formalise the above intuition.
	Consider any execution $\parens*{S,P}$.
	%Specifically, since the program $P$ being considered is probabilistic, the value of any register can be seen as a random variable.
	By analysing how $\Macc$ constrains information flow,
	proving \Cref{eq:sec-thm-main} can be first reduced to proving the special case of $t_u=t_1$ and $t_w\geq t_2+1$,
	where $t_1$ and $t_2$ are time points after the write requests 
	in Lines 1 and 6 of $P_v$ (see \Cref{fig:userv}) are issued, respectively.
	Denote $C_j,L\bracks*{w_j}$ by $D_j$.
	Using the symmetry of $\Macc$ (with respect to different $w_j$),
	we can further reduce our goal to proving 
	\begin{equation}
		\label{eq:mutual-info-close}
		\frac{\Pr\bracks*{A\parens*{t_1}=a\mid B\parens*{t_w}=b,D_1\parens*{t_w}=d}}
		{\Pr\bracks*{A\parens*{t_1}=a}}\approx 1
	\end{equation}
	for any bit $a,b$ and integer $d$.
	Here, the degree of the approximation $\approx$ is related to the RHS of \Cref{eq:sec-thm-main}.

	The remaining analysis largely relies on the concept of conditional independence and techniques in probabilistic graphical models.
	We first identify several time points and random variables of concern.
	For example, for each $j\in [n]$, let $t_{v,j}$ be the time point after the write request in Line 3 of $P_v$ is issued.
	Then, $C_j^1\parens*{t_{v,j}}$ is equal to the value $x_j$ chosen by $v$ in $P_v$.
	Similarly, we can find another random variable $B\parens*{t_v}$ equal to the value $b$ written by $v$ in $P_v$,
	where $t_v$ corresponds to Line 5.
	Next, we can analyse relations between these random variables,
	based on the program $P_v$, matrices $M_0,M_1,M_2$ and temporal ordering of requests.
	These relations are visualised as a graph in \Cref{fig:graphical}, deferred to \Cref{sub:proof_of_classical-security}.
	From the graph, we can obtain conditional independence relations.
	They are used in a tedious but complicated analysis to break down~\Cref{eq:mutual-info-close},
	through decomposition of joint probability distributions, into terms closer to the form in Mermin inequality in \Cref{lmm:Mermin}.
	In particular, we need to use \Cref{lmm:Mermin} for the $\parens*{n-1}$-party case
	(instead of $n$-party, technically due to $\Obs\parens*{w_j,t_w}=B\parens*{t_w},D_j\parens*{t_w}$).
	Finally, we can obtain an upper bound $2^{-(n-7)/2}$ on the degree of approximation in \Cref{eq:mutual-info-close},
	and the conclusion follows.
\end{proof}

\subsection{Security Breach in the Quantum Case}
\label{sub:security_breach_in_the_quantum_case}

Now let us consider the case when registers $C_1,\ldots,C_n$ in the system described in \Cref{sub:problem_setting}
become quantum registers.
This could happen, as indicated in \Cref{qus:main},
when the system upgrades by integrating new quantum computing services.
We need to consider how to properly lift\footnote{In this paper, the terms ``adapt'' and ``lift'' will be used interchangeably.}
this system in \Cref{sub:problem_setting} to the quantum setting.
%Possible rights on a quantum register include applying unitary gates and measurements,
%and in general quantum operations.
We do not bother considering how to lift \texttt{read}, \texttt{write} and \texttt{flip} to the quantum case.
Instead, let us focus on how to lift the right \texttt{all}
(representing full access to a register),
as this suffices to reveal the key problem.
%First, the set $\mathbf{Rt}$ of rights needs to be generalised.
%Let $\Rt_{\textup{c}}=\braces*{\texttt{read},\texttt{write},\texttt{flip}}$ be the set of rights for classical registers
%as in \Cref{ex:main-example}.

At first glance, one might try to interpret a request $\parens*{s,X,\texttt{all}}$ in the quantum setting
as: user $s$ can perform any quantum operation $\calE$ on quantum register $X$.
However, this interpretation forbids any quantum entanglement between objects in $\Obj$.
Since entanglement is believed to be the source of quantum advantages (e.g., \cite{JL03})
for many quantum algorithms~\cite{Shor94,Grover96,Lloyd96,HHL09},
%the power of quantum computing is believed to be significantly rely on quantum entanglement,
such lifting of \texttt{all} is definitely an unsatisfactory choice.

The remaining natural lifting is to interpret:
\begin{itemize}
	\item 
		(LF) Request $\parens*{s,X,\texttt{all}}$ means
		user $s$ can perform any quantum operation $\calE$
		on the composite system of quantum register $X$ and the local memory $L\bracks*{s}$ of $s$.
\end{itemize}
This lifting (LF) implicitly assumes that the local memories of subjects also become quantum;
that is, $L:\Obj\rightarrow \calH_{\Int}$, where $\calH_{\Int}$ is the Hilbert space lifted from $\Int$.
In this case, quantum entanglement can be generated between quantum registers in $\Obj$.
For example, in the system described in \Cref{sub:problem_setting}, 
when $\Macc=M_0$,
user $w_1$ can generate an EPR state $\frac{1}{\sqrt{2}}\parens*{\ket{0}_{C_1}\ket{1}_{C_2}+\ket{1}_{C_1}\ket{1}_{C_2}}$ 
in $C_1$ and $C_2$ (technically, their first qubits),
by first performing a Hadamard $H$ gate on $C_1$,
followed by a $\mathit{CNOT}$ gate on $C_1$ and $L\bracks*{w_1}$,
and finally a $\mathit{SWAP}$ gate between $C_2$ and $L\bracks*{w_1}$.
However, this lifting also turns out to be an unsatisfactory choice,
because it can actually lead to a \textit{security breach}.
In particular, for the system described in \Cref{sub:problem_setting},
the security guaranteed by \Cref{thm:classical-security} will be broken in the quantum case,
as stated in the following theorem.

\begin{theorem}[Security breach in the quantum case]
	\label{thm:quantum-security-breach}
	If $C_1,\ldots,C_n$ in the system described in \Cref{sub:problem_setting} become quantum registers
	and we adopt the lifting (LF),
	%we use the second generalisation of $\Rt$,
	then the secret information of user $u$ can be leaked with certainty in the worst case.
	Specifically, there exists an execution $\parens*{S,P}$ with $P_v$ described in \Cref{fig:userv} such that the mutual information 
	\begin{equation*}
		I\parens*{A\parens*{t_1};\Obs\parens*{w_1,t_2}}=1,
	\end{equation*}
	where $t_1,t_2$ are time points after the write requests in Line $1$ and $6$ in \Cref{fig:userv} are issued, respectively.
	Here, $\Obs\parens*{\cdot,\cdot}$ is defined in \Cref{thm:classical-security}.
\end{theorem}

It is important to note that the security breach in \Cref{thm:quantum-security-breach} is \textit{not} due to additional communication channel created by entanglement,
as the access matrix $\Macc$ of the system does not change. 
Indeed, it is well-known that entanglement cannot enable information transmission between users without direct communication.
Instead, the insecurity proof relies on how entanglement violates Mermin inequality~\cite{Mermin90}. This also implies the threat we reveal has a quantum nature and is not restricted to the specific system considered here. 

\begin{proof}[Proof of \Cref{thm:quantum-security-breach}]
    Note that 
    \begin{equation*}
        \Obs\parens*{w_1,t_2}=B\parens*{t_2},C_j\parens*{t_2},L\bracks*{w_1}\parens*{t_2}.
    \end{equation*}
	It suffices to show there exists an execution $\parens*{S,P}$ in which all user $w_j$ can cooperate 
	such that $\Pr\bracks*{B\parens*{t_2}=A\parens*{t_1}}=1$.
	The program $P$ (in particular, $P_{w_j}$) we construct exactly follows
	the quantum strategy for Mermin $n$-player game~\cite{Mermin90,BBT05},
	which leads to a violation of Mermin inequality in the quantum setting.

	%In \Cref{fig:userw}, the behaviour of each $w_j$ is programmed as $P_{w_j}$.
	Let us first construct the program $P$.
	The program $P_{w_j}$ that describes the behaviour of each $w_j$ is shown in \Cref{fig:userw}.
	Note that when $\Macc= M_1$, %for $t\in [t_1,t_2-1]$, $\Macc(t)\bracks*{w_1,C_j}=\texttt{*}$ for $j\in [n]$,
	from the lifting (LF),
	Line $1$ in \Cref{fig:userw} can be executed by 
	(a) first swapping the content of $C_j$ for each $j\in [n]$ into the local memory $L\bracks*{w_1}$;
	(b) next preparing the state $\ket{\textup{GHZ}(n)}$ in the local memory $L\bracks*{w_1}$; and
	(c) swapping back the content of $L\bracks*{w_1}$ to $C_j$ for each $j\in [n]$, which moves the GHZ state to $C^2$.
	Without loss of generality, we set $P_u$ to consist of a single write $A\gets a$,
	where $a\in \braces*{0,1}$ is the secret information of $u$.

	\begin{figure}
		\centering
		  \includegraphics[width=0.47\textwidth]{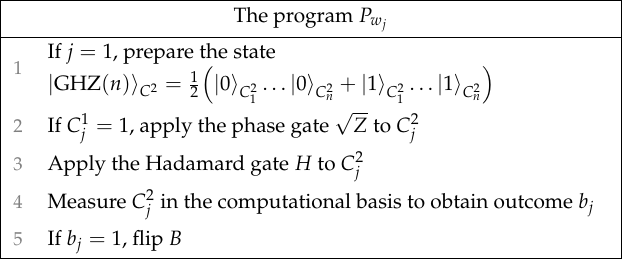}
		\caption{The program $P_{w_j}$ that describes the behaviour of each user $w_j$,
			in an attempt to learn the secret information of user $u$.
		Here, $C_j^k$ represents the $k^{\textup{th}}$ qubit of $C_j$,
		and $C^2=C_1^2, \ldots, C_{n}^2$.}
		\label{fig:userw}
	\end{figure}

	Next we construct the scheduler $S$.
	We take $t_1=2$ and $t_2=8n+5$.
	$S$ is defined such that 
	for $t\in \N$, $S\parens*{\alpha\parens*{0},\ldots,\alpha\parens*{t-1}}=s(t)$,
	where $s(t)$ is defined below.
	For each $s(t)$, we also describe its corresponding behaviour at time $t$.
	%We assume that $v$ and $w_j$ follow their corresponding programs in \Cref{fig:userv,fig:userw}, respectively,
	\begin{itemize}
		\item
			$s(0)=u$: $u$ writes one bit of secret information into $A$.
		\item
			$s(1)=v$: $v$ executes Line $1$ in \Cref{fig:userv} to modify $\Macc$.
		\item
			$s(2)=\ldots = s(2n+1)= w_1$:
			$w_1$ executes Line $1$ in \Cref{fig:userw}.
		\item
			$s(2n+2)=\ldots = s(3n+3) = v$: $v$ executes Lines $2$--$5$ in \Cref{fig:userv}.
		\item
			For $k=0$ to $4$, and $j\in [n]$, $s\parens*{(k+3)n+j+3}=w_j$:
			$w_j$ executes Lines $2$--$5$ in \Cref{fig:userw}.
		\item
			$s(8n+4)=v$: $v$ executes Line $6$ in \Cref{fig:userv} to modify $\Macc$.
		\item
			$s(8n+5)=w_1$: $w_1$ reads the value in $B$.
	\end{itemize}
	In the above, we implicitly fix how to generate requests from the program $P$
	(see also the remark about history generation after \Cref{def:history}).
	The time points above (e.g., $2n+1$, $3n+3$) are chosen regarding this specific generation.
	For example, $w_1$ can executes Line 2 in \Cref{fig:userw} 
	through two requests $\parens*{w_1,C_j^1,\texttt{read}},\parens*{w_1, C_j^2,\texttt{all}}$,
	at time $t=3n+4$ and $t=3n+5$.
	%What really matters is the temporal ordering between different requests.

	Now we verify that the execution $\parens*{S,P}$ constructed above yields $\Pr\bracks*{B(t_2)=A(t_1)}=1$.
	Note that in our system, only $C^2$ will be in quantum superposition.
	Actions on $C^1$ are actually classical, so $C^1$ can be still regarded as a classical random variable, for simplicity of presentation.
	Let $E:=\abs*{C^1(3n+4)}/2$ and 
	\begin{equation*}
		F:=\abs*{\braces*{t\in [3n+4,8n+3]:\alpha\parens*{t}=\parens*{w_j,B,\texttt{flip}},j\in [n]}}\bmod 2.
	\end{equation*}
	By the programs $P_v$ in \Cref{fig:userv} and $P_{w_j}$ in \Cref{fig:userw},
	we have $B\parens*{t_2}=E\oplus F\oplus A\parens*{t_1}$.

	Now it suffices to show that $\Pr\bracks*{E=F}=1$.
	For $b\in \braces*{0,1}$, define
	\begin{equation*}
		\ket{\psi_b}:=\frac{1}{2}H^{\otimes n}\parens*{\ket{0}^{\otimes n}+\parens*{-1}^{b}\ket{1}^{\otimes n}}
		=\frac{1}{2^{(n-1)/2}}\sum_{\abs*{y}\bmod 2=b} \ket{y}.
	\end{equation*}
	It is easy to see that the state of $C^2(6n+4)$ (before each $w_j$ executes Line $4$ in \Cref{fig:userw}) is $\ket{\psi_E}$.
	Thus, we can calculate
	\begin{equation*}
		\Pr\bracks*{F=b\midv E=b}=\sum_{\abs*{y}\bmod 2= b} \abs*{\braket{y|\psi_b}}^2=1.
	\end{equation*}
	The conclusion immediately follows.
\end{proof}

\section{Protection: Access Control in Quantum Computing}
\label{sec:protection_access_control_in_quantum_computing}

Through the explicit scenario in the last section,
we have seen that if the access control system is not properly adapted to the quantum setting,
the security can be threatened.
As indicated by the proofs of \Cref{thm:classical-security,thm:quantum-security-breach},
while the system described in \Cref{sub:problem_setting} is specific,
we have identified that the threat \textit{intrinsically} stems from quantum entanglement,
which is indispensable to quantum computing. 
In this section, we study how to handle such threat from entanglement.

%We first point out an issue with the previous two generalisations of access matrix in \Cref{sec:scenario_threat_from_quantum_entanglement}:
%when we change to a quantum system,
%we implicitly assume that the local memory of each subject also becomes a quantum memory.
%This turns out to be a bad decision,
%as local quantum memories allow uncontrollable (by the access matrix) multi-object operations,
%resulting in the security breach as shown in \Cref{thm:quantum-security-breach}.
%Hence, it will be better to restrict all local memories to be classical,
%and explicitly control multi-object operations by introducing virtual objects that are subsets of real objects.
%To this end, we extend the definition of access matrix.

In classical access control, usually an access request $\parens*{s,o,r}$ only involves a single object $o$,
which is sufficient in most practical scenarios.
However, quantum operations on multiple objects (registers) can generate entanglement between them even when they were initially in a separable state.
These quantum operations should be explicitly controlled
to protect the security of quantum systems.
For this purpose, we extend the set $\Obj$ to include every quantum subsystem 
consisting of multiple quantum objects as a virtual object, as suggested in \cite{YFY13}.
More precisely, suppose that $\Objc$ and $\Objq$ are the sets of real classical and quantum objects, respectively.
Then the set of objects in the system considered in this section is $\Obj=\Objc\cup \calP_+\parens*{\Objq}$,
where $\calP_+\parens*{\cdot}$ stands for the set of all non-empty subsets.

%Since we now explicitly deal with multi-object quantum operations,
Meanwhile, in this section, we restrict the local memories of subjects to be classical;
i.e., we only consider $L:\Sub\rightarrow \Int$ (instead of $L:\Sub\rightarrow \calH_{\Int}$).
As shown in \Cref{thm:quantum-security-breach},
allowing local memories to be quantum is likely to introduce uncontrollable quantum entanglement 
that may lead to security breach.
Note that avoiding implicit local quantum memory is equivalent to managing all quantum objects explicitly in the access control, and thus does not affect the computational power of the system being protected.

Consequently, in a quantum access control system, we have $\Rt=\Rtc\cup \Rtq$,
where $\Rtc$ and $\Rtq$ consist of abilities to perform operations on classical registers
and quantum subsystems, respectively.
Note that if $s\in \Sub$ performs a quantum measurement on quantum registers,
the classical outcomes produced will be stored into the local memory $L\bracks*{s}$.

We summarise these conventions in the following definition for clarity.

\begin{definition}[Core model of quantum access control]
	\label{def:core-model}
	The components in the core model of quantum access control are specified as follows.
	\begin{itemize}
		\item 
			$\Sub$ is a set of users.
			$\Obj=\Objc\cup \calP_+\parens*{\Objq}$,
			where $\Objc$ and $\Objq$ are sets of classical and quantum registers.
		\item
			The local memories $L:\Sub\rightarrow \Int$ of subjects are classical.
		\item
			The classical part of the access control is guarded by the access matrix
			$\Mcc:\Sub\times \Objc\rightarrow \calP\parens*{\Rtc}$.
	\end{itemize}
\end{definition}

All models of quantum access control to be studied in this section are refinements of the core model in~\Cref{def:core-model}.
To handle threats from quantum entanglement,
we introduce two types of models.
In \Cref{sub:control_of_quantum_operations},
we consider explicitly controlling quantum operations on subsystems of multiple quantum registers;
in \Cref{sub:control_of_entanglement},
we consider explicitly controlling the resource of quantum entanglement.

To evaluate and compare these models, 
we consider the following three metrics for an access control model, following~\cite{HFK06,HK12}:
\begin{enumerate}
	\item
		\textbf{Security}, in this paper, concerns whether the model can properly manage quantum entanglement between objects
        and therefore protect against threats from entanglement.
        In particular,
		if a model is secure, then the system described in \Cref{sub:problem_setting} can be lifted to such model 
        while retaining the security guarantee in \Cref{thm:classical-security}.
	\item 
		\textbf{Flexibility} is related to the granularity of specifying the access control,
		and thus how well the model can support the principle of least privilege~\cite{SS75}.
		In this paper, we compare the flexibility of different models by \Cref{def:flexibility}.
		%Intuitively, a model is more flexible than another if any system in the latter can be equivalently 
		%(with respect to \Cref{def:equivalent-sys}) specified in the former.
	\item
		\textbf{Efficiency} measures the space complexity for implementing the model,
		and the time complexity for handling an access request.
\end{enumerate}

All of the proposed models are secure, but their flexibility and efficiency vary.
In practice, the choice of which model to use depends
on the \textit{specific requirements} about the flexibility and efficiency.
One can also consider a hybrid of these models.
For visualisation, in \Cref{tab:compare-sec-eff}, 
we compare the security and efficiency of different models introduced in the following subsections,
and in \Cref{fig:compare-flex} we compare the flexibility.
%They will be explained in detail in the following sections.

%\newcommand{\tikzxmark}{%
%\tikz[scale=0.25] {
%    \draw[line width=1,line cap=round,red!90!black] (0,0) to [bend left=6] (1,1);
%    \draw[line width=1,line cap=round,red!90!black] (0.2,0.95) to [bend right=3] (0.8,0.05);
%}}
%\newcommand{\tikzcmark}{%
%\tikz[scale=0.25] {
%    \draw[line width=1,line cap=round,green!70!black] (0.25,0) to [bend left=10] (1,1);
%    \draw[line width=1.2,line cap=round,green!70!black] (0,0.35) to [bend right=1] (0.23,0);
%}}

\begin{figure}
	\centering
    \includegraphics[width=0.47\textwidth]{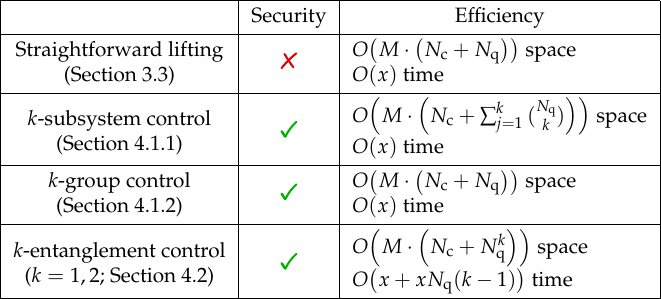}
	\caption{Comparison of Security and Efficiency of different quantum access control models in \Cref{sec:protection_access_control_in_quantum_computing}.
	Here, the security is against threats from entanglement revealed in \Cref{sec:scenario_threat_from_quantum_entanglement}.
	The efficiency is about the space complexity for the access control and the time complexity to handle an access request.
	We assume $\abs*{\Sub}=M$, $\abs*{\Objc}=N_{\textup{c}}$, $\abs*{\Objq}=N_{\textup{q}}$,
	and the request has length $x$.
	}
	\label{tab:compare-sec-eff}
\end{figure}

\begin{figure}
	\centering
        \includegraphics[width=0.4\textwidth]{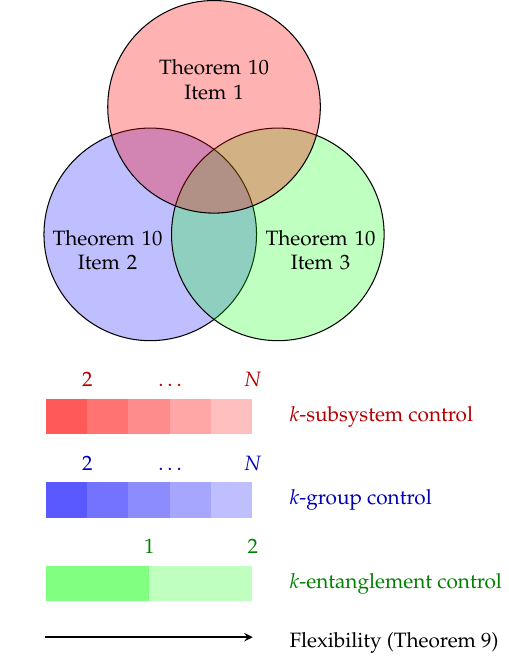}
	\caption{Comparison of flexibility of different quantum access control models.}
	\label{fig:compare-flex}
\end{figure}

%\begin{definition}[Access matrix for quantum subsystems]
%	Suppose that $\mathbf{Obj}=\mathbf{Obj}_c\cup \mathbf{Obj}_q$,
%	where $\mathbf{Obj}_c$ and $\mathbf{Obj}_q$ stand for the sets of classical and quantum objects, respectively.
%	An access matrix for quantum subsystems has the form $M:\mathbf{Sub}\times \mathbf{Obj}'\rightarrow \calP\parens*{\mathbf{Rt}}$,
%	where $\mathbf{Obj}'=\mathbf{Obj}_c\cup \mathbf{Obj}_q'$ and 
%	$\mathbf{Obj}_q'$ is determined by the specific security model
%	with elements falling in $\calP\parens*{\mathbf{Obj}_q}$.
%\end{definition}

\subsection{Control of Quantum Operations}
\label{sub:control_of_quantum_operations}

%Direct access control on quantum operations has been studied in \cite{YFY13}.
%Here, we provide a more rigorous investigation and discuss different variants.

%These generalisations have levels of granularity from the coarsest to the finest.

%\begin{example}[No multi-object operations]
%	\label{ex:no-multi}
%	$\Sub$: users,
%	$\Obj$: classical registers and quantum subsystems,
%	$\Rt=\Rt_c \cup \braces*{\texttt{op}_{\calE}:\calE\in \OP\parens*{\calH}}$,
%	$\Atr=\braces*{\Mcc,\Mq,L}$,
%	$\Pol=\braces*{\Auth}$,
%	where $\Mcc: \Sub\times \Obj_c\rightarrow \calP\parens*{\Rt}$,
%	$\Mq:\Sub\times \Obj_q\rightarrow \calP\parens*{\Rt}$,
%	$L:\Obj \rightarrow \Int$ and 
%	\begin{equation*}
%		\mathit{Auth}\parens*{s,o,r}\Leftarrow \parens*{o\in \Obj_c \wedge r\in \Mcc\bracks*{s,o}}\vee \parens*{o=\braces*{X}\wedge X\in \Obj_q\wedge r\in \Mq\bracks*{s,o}}.
%	\end{equation*}
%
%	The above system is just another way to express the first generalisation in \Cref{sub:analysis_security_broken_in_the_quantum_case}.
%	It forbids quantum operations on multiple objects and therefore the quantum entanglement between different objects.
%\end{example}

%It is possible in the example in Section 1, 
%at time $t\in [0,t_1)$, all users $w_j$ need to run the Shor's algorithm and create large entanglement.
%How to deal with this case?

\subsubsection{Subsystem Control}
\label{sub:subsystem_control}

Subsystem control has been initially studied in~\cite{YFY13}.
The original observation in~\cite{YFY13} is that having full access to a composite subsystem of quantum registers $A$ and $B$
is not the same as the combination of separate accesses to $A$ and to $B$.
Thus, they proposed to regard every quantum subsystem of multiple quantum registers as a virtual object, 
as mentioned at the beginning of \Cref{sec:protection_access_control_in_quantum_computing}.
In our terminology, they define the authorisation rule via an access matrix $M:\Sub\times \Obj\rightarrow \calP\parens*{\Rt}$,
where $\Obj=\Objc\cup \calP_+\parens*{\Objq}$ is as defined in the core model~\Cref{def:core-model}.
In the following, we slightly extend this idea to $k$-subsystem control,
which offers a better trade-off between flexibility and efficiency.

\begin{definition}[$k$-subsystem control]
	\label{def:subsys-control} 
	Suppose that $1\leq k\leq \abs*{\Obj_q}$.
	The $k$-subsystem control model, denoted by $\SUBSYS^k$, extends \Cref{def:core-model} by letting $\Atr=\braces*{\Mcc,\Mq,L}$,
	$\Rule=\braces*{\Auth}$,
	%where $\Mcc: \Sub\times \Objc\rightarrow \calP\parens*{\Rt}$,
$\Mq:\Sub\times \calP_{\leq k}\parens*{\Objq}\rightarrow \calP\parens*{\Rtq}$,
	and \begin{equation*}
		\begin{split}
			\Auth\parens*{s,o,r}\equiv \ & p_{\textup{c}}\wedge p_{\textup{q}},\ \text{where:}\\
			p_{\textup{c}}\equiv \ & o\in \Objc \rightarrow  r\in \Mcc\bracks*{s,o},\\
			p_{\textup{q}}\equiv \ &o\in \calP_+\parens*{\Objq}\rightarrow \abs*{o}\leq k\wedge r\in \Mq\bracks*{s,o}.
		\end{split}
	\end{equation*}
	Here, $\calP_{\leq k}\parens*{\cdot}$ denotes the set of non-empty subsets of cardinality $\leq k$.
\end{definition}

Intuitively, in the authorisation rule, $p_{\textup{c}}$ says that if $o$ is a classical register, 
then we check if $r\in\Mcc\bracks*{s,o}$;
and $p_{\textup{q}}$ says that if $o$ is a quantum subsystem involving $\leq k$ registers,
then we check if $r\in \Mq\bracks*{s,o}$.
Compared to~\cite{YFY13} (equivalent to setting $k=\abs*{\Objq}$),
\Cref{def:subsys-control} only authorises requests involving subsystem of size $\leq k$,
which achieves better efficiency by reducing the space complexity of storing the attribute $\Mq$,
as will be explicitly shown later in \Cref{thm:eff-subsystem}.

Typical choices of $k$ include $k=2$ and $k=\abs*{\Objq}$.
Note that the case $k=1$ forbids any entanglement between quantum registers,
recovering our first attempt to lift the right $\texttt{all}$ in \Cref{sub:security_breach_in_the_quantum_case}.

The $k$-subsystem control model provides the most direct control over
quantum operations performed on multiple quantum registers,
and therefore offers protection against threats from quantum entanglement (as illustrated in \Cref{sec:scenario_threat_from_quantum_entanglement}).
The security of this model is formalised in the following theorem.

\begin{theorem}[Security of $k$-subsystem control]
	\label{thm:security-k-subsys}
	For $2\leq k \leq \abs*{\Objq}$, the system described in \Cref{sub:problem_setting} can be lifted to a system with $k$-subsystem control
	such that the security guarantee in \Cref{thm:classical-security} is retained.
\end{theorem}

It is worth noting that although the security in \Cref{thm:security-k-subsys} 
(and in subsequent theorems about other models)
is stated with respect to the specific system described in \Cref{sub:problem_setting},
the access control model itself can be employed to protect against \textit{any threat from quantum entanglement}.
This is because, within the model, entanglement can be explicitly forbidden through specification.

\begin{proof} [Proof of \Cref{thm:security-k-subsys}]
	We only prove the theorem for $k=2$.
	The proof for other $k$ is similar and thus omitted.
	For better illustration of the flexibility of the $k$-subsystem control model,
	let us assume several additional quantum registers, 
	say $\Objq=\braces*{C_1,\ldots, C_n,D_1,\ldots, D_5}$;
	and we only show one possible way of lifting to this model.
	To prove that the lifted system retains the security guarantee in \Cref{thm:classical-security},
	it suffices to verify that no entanglement is allowed to be generated among $C_1,\ldots, C_n$.

	The lifted system has $\Objc=\braces*{A,B,\Mcc,\Mq}$.
	We construct the lifting as follows.
	\begin{itemize}
		\item 
			Let $\Mcc\bracks*{v,\Mcc}=\Mcc\bracks*{v,\Mq}=\braces*{\texttt{all}}$,
			meaning that $v$ can modify the attributes $\Mcc$ and $\Mq$ like 
			that it can modify $\Macc$ in \Cref{fig:userv}.
		\item
			For $X\in \braces*{A,B}$,
			we define $\Mcc\bracks*{s,X}=\Macc\bracks*{s,X}$.
			For $X\in \braces*{C_1,\ldots,C_n}$, let $\Mq\bracks*{s,\braces*{X}}= \Macc\bracks*{s,X}$.
			For $X\in \braces*{D_1,\ldots,D_5}$, let $\Mq\bracks*{s,\braces*{X}}=\braces*{\texttt{all}}$.
			We also modify Line $1$ and $6$ of $P_v$ in \Cref{fig:userv} to write $\Mcc\bracks*{s,X}$ and $\Mq\bracks*{s,\braces*{X}}$
			instead of $\Macc\bracks*{s,X}$.
		\item
			Let $\Mq\bracks*{w_1,\braces*{C_1,D_1}}=\Mq\bracks*{w_2,\braces*{C_2,D_2}}=
			\Mq\bracks*{w_3,\braces*{D_3,D_4}}=\Mq\bracks*{w_3,\braces*{D_4,D_5}}=\braces*{\texttt{all}}$.
		\item
			Those $\Mcc\bracks*{s,o}$ and $\Mq\bracks*{s,o}$ unspecified above are defined to be $\emptyset$.
			In particular, we have $\Mq\bracks*{w_j,\braces*{C_l,C_r}}=\emptyset$ for $l\neq r$,
			implying that quantum entanglement cannot be generated among $C_1,\ldots,C_n$.
	\end{itemize}
	Note that the above lifting only forbids entanglement generated among $C_1,\ldots,C_n$,
	but allows entanglement generated between $C_1$ and $D_1$,
	$C_2$ and $D_2$, $D_3$ and $D_4$, and $D_4$ and $D_5$.
	For illustration, we visualise each subsystem on which quantum operations are allowed in~\Cref{fig:SubSecEx}.

	\begin{figure}
		\centering
		\includegraphics[width =0.39\textwidth]{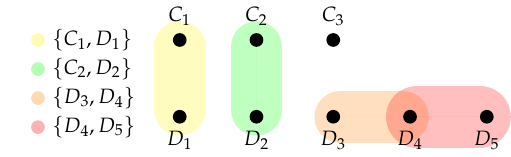}
		\caption{Illustration of allowed quantum operations on multiple registers
			in a system in the $2$-subsystem control model (see the proof of \Cref{thm:security-k-subsys}; take $n=3$).
		Each $2$-subsystem on which some user has access right \texttt{all} is colored.}
		\label{fig:SubSecEx}
	\end{figure}
\end{proof}

Now we analyse the efficiency of $k$-subsystem control.
Remember that the efficiency concerns the space and time complexities.
Here and throughout this paper,
the space complexity of implementing an access control model is
measured by the number of classical memory locations (each capable of storing a bounded integer)
required to store all the attributes.
The time complexity for handling an access request is measured by the number of elementary operations 
(including arithmetic, logical and memory access operations) in the standard word RAM model.

\begin{theorem}[Efficiency of $k$-subsystem control]
	\label{thm:eff-subsystem}
	Suppose that $\abs*{\Sub}=M$, $\abs*{\Objc}=N_{\textup{c}}$ and $\abs*{\Objq}=N_{\textup{q}}$,
	then the $k$-subsystem control model uses $O\parens*{M\cdot \parens*{N_{\textup{c}}+\sum_{j=1}^k \binom{N_{\textup{q}}}{k}}}$ space
	for access control,
	and it takes $O\parens*{x}$ time to authorise an access request of length $x$.
\end{theorem}

Compared to the original idea in~\cite{YFY13}, 
our \Cref{thm:eff-subsystem}, together with \Cref{thm:flex-hierarchy} later in \Cref{sub:comparison_of_flexibility},
demonstrates a trade-off between flexibility and efficiency.
In particular, taking smaller $k$ in the $k$-subsystem control model
leads to greater efficiency but reduced flexibility (see \Cref{thm:flex-hierarchy}). 
For example, focusing on the dependence on $N_{\textup{q}}$,
then for $k=2$, the space complexity is $O\parens*{N_{\textup{q}}^2}$.
However, for $k=N_{\textup{q}}$, the case originally suggested by~\cite{YFY13},
the space complexity is $O\parens*{2^{N_{\textup{q}}}}$,
which is exponentially large.
%may be infeasible if one wishes to enforce qubit-wise access control even on existing quantum computers.

It is also worth mentioning that the space or time complexity in \Cref{thm:eff-subsystem}
and subsequent theorems about other models is regarding the worst case.
We do not bother considering more efficient data structures (like ACL)~\cite{SS94} to store the attributes,
and leave this for future works (see also \Cref{sec:discussion}).

\begin{proof} [Proof of \Cref{thm:eff-subsystem}]
    The space complexity for implementing $k$-subsystem control 
    is dominated by that for storing the attributes $\Mcc$ and $\Mq$ in \Cref{def:subsys-control}.
    The matrix representation of $\Mcc$ has $\abs*{\Sub}$ rows and $\abs*{\Objc}$ columns,
    while that of $\Mq$ has $\abs*{\Sub}$ rows and $\abs*{\calP_{\leq k}\parens*{\Objq}}=\sum_{j=1}^k \binom{N_{\textup{q}}}{k}$ columns. 
    %whose complete matrix representation has $\abs*{\Sub}$ rows and $\abs*{\Obj}$ columns.
   % A simple counting argument shows that $\abs*{\calP_{\leq k}\parens*{\Objq}}$.
    
    The time complexity for handling an access request $(s,o,r)\in \Req$
    is dominated by, according to the authorisation rule in \Cref{def:subsys-control},
    reading the whole request and
    checking the size of the subsystem $o\subseteq \Objq$,
    which scales as the length of the request. 
\end{proof}

\subsubsection{Group Control}
\label{sub:group_control}

In \Cref{sub:subsystem_control}, $k$-subsystem control provides direct control of quantum operations
on subsystem of size $\leq k$. 
However, the space complexity for implementing $k$-subsystem control (even for the smallest nontrivial $k=2$)
could be formidable when the number $N_{\textup{q}}$ of quantum objects is large.
In practical classical systems, 
the number of objects can be in the tens of millions~\cite{HFK06}.
While it may take a long time to build quantum computers at such a scale,
we can still consider models with lower space requirements,
such as the following $k$-group control model.

\begin{definition}[$k$-group control]
	\label{def:group-control}
	Suppose that $1\leq k\leq \abs*{\Obj_q}$.
	The $k$-group control model, denoted by $\GRP^k$,
	extends \Cref{def:core-model} by setting  $\Atr=\braces*{\Mcc,\Mq,G,L}$,
	$\Rule=\braces*{\Auth}$,
	%where $\Mcc: \Sub\times \Objc\rightarrow \calP\parens*{\Rt}$,
	 $\Mq:\Sub\times \Objq\rightarrow \calP\parens*{\Rtq}$,
	$G:\Objq \rightarrow [k]$,
	and 
	\begin{equation*}
		\begin{split}
			\Auth\parens*{s,o,r}\equiv \ & p_{\textup{c}}\wedge p_{\textup{q}},\ \mbox{where}:\\
			p_{\textup{c}}\equiv \ &o\in \Objc \rightarrow r\in \Mcc\bracks*{s,o},\\
			p_{\textup{q}}\equiv \ &o\in \calP_+\parens*{\Objq}\rightarrow \parens*{\forall X,Y\in o: G\bracks*{X}=G\bracks*{Y}} \wedge\\
            & 
            \parens*{\forall X\in o: r\in \Mq\bracks*{s,X}}.
		\end{split}
	\end{equation*}
\end{definition}

Intuitively, the attribute $G$ assigns a group label to every object.
In the authorisation rule, $p_{\textup{c}}$ is standard;
and $p_{\textup{q}}$ says that if $o$ is a quantum subsystem,
then the request is authorised only if  all quantum registers in $o$ has the same group label, and the right $r$ appears in $\Mq\bracks*{s,X}$
for any quantum register $X\in o$. 
Note that the attribute $\Mq$ in \Cref{def:group-control} is different from that in \Cref{def:subsys-control}:
$\Mq$ in the $k$-group control model has a smaller domain.

Note that \Cref{def:group-control} can be slightly modified (by introducing a group label $0$)
to define an abstraction of the entangling zone,
which is employed in some architectures of quantum hardware~\cite{BEG+23}.
In this case, two-qubit quantum operations can only be performed on qubits in the entangling zone.

The $k$-group control model also provides explicit control over quantum operations
performed on multiple quantum registers,
through the attribute $G$ that assigns group labels.
The security of this model is formalised as follows. 

\begin{theorem}[Security of $k$-group control] 
    \label{thm:security-k-group}
	Let $n$ be as defined in \Cref{sub:problem_setting}.
	For $n+1\leq k\leq \abs*{\Objq}$,
	the system described in \Cref{sub:problem_setting} can be lifted to a system in $\GRP^k$
	such that the security guarantee in \Cref{thm:classical-security} is retained.
\end{theorem}

\begin{proof}
	We only prove the theorem for $k=n+1$.
	The proof for other $k$ is similar and thus omitted.
	Like in the proof of \Cref{thm:security-k-subsys},
	let us assume several additional quantum registers, 
	say $\Objq=\braces*{C_1,\ldots, C_n,D_1,\ldots, D_5}$;
	and we only show one possible way of lifting to this model.
	To prove that the lifted system retains the security guarantee in \Cref{thm:classical-security},
	it suffices to verify that no entanglement is allowed to be generated among $C_1,\ldots, C_n$.

	The lifted system has $\Objc = \braces*{A,B,\Mcc,\Mq,G}$.
	We construct the lifting as follows.
	\begin{itemize}
		\item 
			Let $\Mcc\bracks*{v,\Mcc}=\Mcc\bracks*{v,\Mq}=\braces*{\texttt{all}}$.
		\item
			For $X\in \braces*{A,B}$,
			we define $\Mcc\bracks*{s,X}=\Macc\bracks*{s,X}$.
			For $X\in \braces*{C_1,\ldots,C_n}$, let $\Mq\bracks*{s,X}= \Macc\bracks*{s,X}$.
			For $X\in \braces*{D_1,\ldots,D_5}$, let $\Mq\bracks*{s,X}=\braces*{\texttt{all}}$.
			We also modify Line $1$ and $6$ of $P_v$ in \Cref{fig:userv} to write $\Mcc\bracks*{s,X}$ and $\Mq\bracks*{s,X}$
			instead of $\Macc\bracks*{s,X}$.
		\item
			Let $G\bracks*{C_1}=G\bracks*{D_1}=1$,
			$G\bracks*{C_2}=G\bracks*{D_2}=2$,
			$G\bracks*{C_j}=j$ for $j>2$,
			and $G\bracks*{D_2}=G\bracks*{D_3}=G\bracks*{D_4}=n+1$.
	\end{itemize}
	By \Cref{def:group-control}, the above lifting forbids entanglement generated among $C_1,\ldots, C_n$,
	but allows entanglement generated between $C_1$ and $D_1$,
	$C_2$ and $D_2$, and among $D_3$, $D_4$ and $D_5$.
	For illustration, we visualise each group within which quantum operations are allowed in~\Cref{fig:GCloSecEx}.

	\begin{figure}
		\centering
		\includegraphics[width =0.38\textwidth]{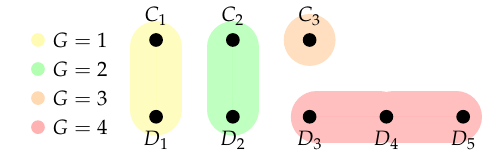}
		\caption{Illustration of allowed quantum operations on multiple quantum registers
		in a system in the $n+1$-group control model (see the proof of \Cref{thm:security-k-group}; take $n=3$).}
		\label{fig:GCloSecEx}
	\end{figure}
\end{proof}

Now we analyse the efficiency of the $k$-group control model.
Focusing on the dependence on $N_{\textup{q}}$,
the space complexity is $O(N_{\textup{q}})$,
which is much smaller than that of the $k$-subsystem control model.

\begin{theorem}[Efficiency of $k$-group control]
	Suppose that $\abs*{\Sub}=M$, $\abs*{\Objc}=N_{\textup{c}}$ and $\abs*{\Objq}=N_{\textup{q}}$,
	then the $k$-group control model uses $O\parens*{M\cdot \parens*{N_{\textup{c}}+N_{\textup{q}}}}$ space for access control,
	and it takes $O\parens*{x}$ time to handle an access request of length $x$.
\end{theorem}

\begin{proof}
    Similar to the proof of \Cref{thm:eff-subsystem},
    the space complexity is dominated
    by that for storing the attributes $\Mcc$ and $\Mq$ in \Cref{def:group-control}.
    The matrix representation of $\Mcc$ has $\abs*{\Sub}$ rows and $\abs*{\Objc}$ columns,
    while that of $\Mq$ has $\abs*{\Sub}$ rows and $\abs*{\Objq}$ columns.
    
    The time complexity is dominated by, according to the authorisation rule in \Cref{def:group-control},
    checking if all $X\in o$ have the same group label.
    This can be done by (a) picking an $X\in o$;
    (b) scanning other $Y\in o$; (c) checking if $G\bracks*{X}=G\bracks*{Y}$. 
    The conclusion immediately follows.
\end{proof}

\subsection{Control of Entanglement}
\label{sub:control_of_entanglement}

The subsystem control and group control models in \Cref{sub:control_of_quantum_operations}
offer explicit control over quantum operations on multiple quantum registers that can generate entanglement.
However, within these models,
it is not possible to explicitly control entanglement as a resource:
for example, we cannot make a specification to 
``forbid any entanglement to exist between quantum registers $A$ and $B$''
after entanglement has been established between $A$ and $B$,
because no information about existing entanglements is recorded.
Thus, we propose the following model to control the resource of entanglement.

\begin{definition}[$1$-entanglement control]
	\label{def:1-ent-control}
	The $1$-entanglement control model, denoted by $\ENT^1$,
	extends \Cref{def:core-model} by letting
	$\Atr=\braces*{\Mcc,\Mq,\Me,D,L}$,
	$\Rule=\braces*{\Auth,\PU}$,
	%where $\Mcc: \Sub\times \Objc\rightarrow \calP\parens*{\Rt}$,
	$\Mq:\Sub\times \Objq\rightarrow \calP\parens*{\Rtq}$,
	$\Me, D:\Objq\rightarrow \braces*{\mathit{true},\mathit{false}}$,
	and 
	\begin{align*}
		\begin{split}
		\Auth\parens*{s,o,r}  \equiv\ & p_{\textup{c}}\wedge p_{\textup{e}}\wedge p_{\textup{q}},\text{ where:}\\
        p_{\textup{c}}\equiv \ &o\in \Objc \rightarrow r\in \Mcc\bracks*{s,o},\\
		p_{\textup{e}}\equiv \ &o=\Me\bracks*{X}\rightarrow \parens*{\neg D\bracks*{X}\wedge \Me\bracks*{X}\rightarrow r= \texttt{read} },\\
	p_{\textup{q}}\equiv \ &o\in \calP_+\parens*{\Objq}\rightarrow \parens*{\forall X\in o: r\in \Mq\bracks*{s,X}}\wedge\\
    &\big(\abs*{o}>1\rightarrow \bigwedge_{X\in o}\Me\bracks*{X}\big),
		\end{split}
		\\
		\begin{split}
			\PU\parens*{s,o,r} \equiv\ & \mathbf{if}\ o\in \calP_+\parens*{\Objq}\ \mathbf{then}\\
									   & \quad \mathbf{if}\ r=\texttt{measure}\ \mathbf{then}\\
									   & \quad\quad \mathbf{for}\ X\in o\ \mathbf{do}\ D\bracks*{X}:=\mathit{true}\ \mathbf{od}\\
									   & \quad\mathbf{else}\ \mathbf{if}\ \abs*{o}>1\ \mathbf{then}\\
									   & \quad\quad \mathbf{for}\ X\in o\ \mathbf{do}\ D\bracks*{X}:=\mathit{false}\ \mathbf{od}\\
									   & \quad \mathbf{fi}\\
									   &\mathbf{fi}
		\end{split}
	\end{align*}
	Here, $\texttt{measure}\in \Rtq$ means the ability to perform a complete measurement (see \Cref{sub:quantum_computing}).
        Recall that $\PU$ denotes the post-update rule (see \Cref{def:post-update}). 
	%Here, $\Bracks*{r}$ represents the semantics of the operation $r$.
\end{definition}

In \Cref{def:1-ent-control}, we introduce two attributes $\Me$ and $D$.
For quantum register $X\in \Objq$, $\Me\bracks*{X}$ represents 
whether $X$ is allowed to be entangled with other quantum registers;
and $D\bracks*{X}$ represents whether $X$ is promised to remain disentangled from other quantum registers.
More precisely, $D\bracks*{X}=\mathit{true}$ means $X$ is promised to be disentangled,
and $D\bracks*{X}=\mathit{false}$ means $X$ can be probably entangled.
The authorisation and post-update rules are explained as follows.
\begin{itemize}
	\item 
		For the authorisation rule, $p_{\textup{c}}$ is standard.
		$p_{\textup{e}}$ is used to prevent the case $D\bracks*{X}=\mathit{false}\wedge \Me\bracks*{X}=\mathit{true}$,
		which means quantum register $X$ is not allowed to but being entangled with other registers.
		So, in $p_{\textup{e}}$, if $D\bracks*{X}=\mathit{false}$ and $\Me\bracks*{X}=\mathit{true}$,
		then the current request can only read but not modify $\Me\bracks*{X}$.
		$p_{\textup{q}}$ states that to exercise right $r$ on a quantum subsystem $o$,
		$r$ needs to be appear in $\Mq\bracks*{X}$ for any $X\in o$;
		and if $o$ involves multiple registers, then every $X\in o$ should be allowed to be entangled.
	\item
		The post-update rule updates the attribute $D$ after an authorised request.
		If the request performs a complete measurement on a quantum subsystem, 
		then every registers within are promised to be disentangled.
		Otherwise, if the subsystem involves multiple registers,
		the registers within can probably be entangled (in the worst case).
\end{itemize}

It is worth pointing out that the attribute $D$ only serves as an \textit{approximated knowledge} of existing entanglements.
As an approximation, it is possible that $D\bracks*{X}=\mathit{false}$ while $X$ is actually disentangled.
In this case, due to the above authorisation rule, before a user tries to modify $\Me\bracks*{X}$ to $\mathit{false}$,
some user in the system must perform a measurement on $X$ to force it to be disentangled, which is redundant.
Nevertheless, we suspect that it is impractical, without tracing the explicit state of quantum registers, 
to have accurate control (instead of approximation) of entanglement.
Meanwhile, tracing the explicit state is often beyond the scope of access control.

Another point worth mentioning for the post-update rule is that
we use complete measurement as a promise for disentanglement.
An open question here is whether there is other weaker condition of promising disentanglement 
other than complete measurement (see also \Cref{sec:discussion}).

We can further refine \Cref{def:1-ent-control} into the following model
that records more information about existing entanglements.

\begin{definition}[$2$-entanglement control]
	\label{def:2-ent-control}
	The $2$-entanglement control model, denoted by $\ENT^2$,
	extends \Cref{def:core-model} as follows.
	Let $\Atr=\braces*{\Mcc,\Mq,\Me,D,L}$,
	$\Rule=\braces*{\Auth,\PU}$,
	%where $\Mcc: \Sub\times \Objc\rightarrow \calP\parens*{\Rt}$,
	where $\Mq:\Sub\times \Objq\rightarrow \calP\parens*{\Rtq}$,
	$\Me, D:\calP_{2}\parens*{\Objq}\rightarrow \braces*{\mathit{true},\mathit{false}}$,
	and 
	\begin{align*}
		\begin{split}
		\Auth\parens*{s,o,r}  \equiv\ & p_{\textup{c}}\wedge p_{\textup{e}}\wedge p_{\textup{q}},\text{ where:}\\
        p_{\textup{c}} \equiv\ &o\in \Objc \rightarrow r\in \Mcc\bracks*{s,o},\\
										  p_{\textup{e}}\equiv \ &o=\Me\bracks*{X,Y}\rightarrow \parens*{\neg D\bracks*{X,Y}\wedge \Me\bracks*{X,Y}\rightarrow r=\texttt{read}},\\
										 p_{\textup{q}}\equiv \ &o\in \calP_+\parens*{\Objq}\rightarrow \parens*{\forall X\in o: r\in \Mq\bracks*{s,X}}\wedge\\
                                         &\big(\abs*{o}>1\rightarrow \bigwedge_{X\neq Y\in o}\Me\bracks*{X,Y}\big)
		\end{split}
		\\
		\begin{split}
			\PU\parens*{s,o,r} \equiv\ & \mathbf{if}\ o\in \calP_+\parens*{\Objq}\ \mathbf{then}\\
									   & \quad \mathbf{if}\ r=\texttt{measure}\ \mathbf{then}\\
									   & \quad\quad \mathbf{for}\ X\in o\wedge Y\in \Objq\ \mathbf{do}\ D\bracks*{X,Y}:=\mathit{true}\ \mathbf{od}\\
									   & \quad \mathbf{else}\ \mathbf{if}\ \abs*{o}>1\ \mathbf{then}\\
									   & \quad\quad \mathbf{for}\ X\neq Y\in o\ \mathbf{do}\ D\bracks*{X,Y}:=\mathit{false}\ \mathbf{od}\\
									   & \quad \mathbf{fi}\\
									   &\mathbf{fi}
		\end{split}
	\end{align*}
	Here, $\calP_2\parens*{\cdot}$ denotes the set of subsets of cardinality $2$.
\end{definition}

Compare to $\ENT^1$ in \Cref{def:1-ent-control},
we extend the attributes $\Me$ and $E$ to be functions on $\calP_{2}\parens*{\Objq}$.
Specifically, $\Me\bracks*{X,Y}$ represents whether $X,Y$ are allowed to be entangled;
and $D\bracks*{X,Y}$ represents whether $X$ is promised to be disentangled from $Y$.

The authorisation and post-update rules in \Cref{def:2-ent-control} are similar to but more fine-grained 
(regarding entanglement between two quantum registers)
than those in \Cref{def:1-ent-control}.
Note that in the post-update rule,
we modify $D\bracks*{X,Y}$ to be $\mathit{true}$ for all $Y\in \Objq$ when $X$ is completely measured.

In the above, we only define $k$-entanglement control for $k=1,2$.
A similar definition for higher $k$ is possible,
but it seems less useful due to the following intuitive reason.
$\ENT^2$ is more flexible than $\ENT^1$ because it records ``whether two quantum registers can be entangled'',
which is more fine-grained than ``whether one quantum register can be entangled with others''.
For example, saying ``$X_1,X_2$ are entangled'' is more fine-grained than 
saying ``$X_1$ is entangled with some register and $X_2$ is also entangled''.
However, when we consider $k=3$, 
it is unclear whether saying ``$X_1,X_2,X_3$ are entangled'' is more fine-grained
than saying ``$X_1,X_2$ are entangled and $X_1,X_3$ are also entangled''.

While the entanglement control greatly differs from models in \Cref{sub:control_of_quantum_operations}, 
it also offers protection against threats from entanglement,
as stated in the following theorem. 

\begin{theorem}[Security of $k$-entanglement control]
	\label{thm:security-1-ent}
	The system described in \Cref{sub:problem_setting} can be lifted to a system in $\ENT^1$ (or $\ENT^2$)
	such that the security guarantee in \Cref{thm:classical-security} is retained.
\end{theorem}

\begin{proof}
	We only prove the theorem for $\ENT^1$,
	and the proof for $\ENT^2$ is similar.
	Like in the proof of \Cref{thm:security-k-subsys},
	let us assume several additional quantum registers, 
	say $\Objq=\braces*{C_1,\ldots, C_n,D_1,\ldots, D_5}$;
	and we only show one possible way of lifting.

	The lifted system has $\Objc = \braces*{A,B,\Mcc,\Mq,\Me,D}$.
	We construct the lifting as follows.
	\begin{itemize}
		\item 
			Let $\Mcc\bracks*{v,\Mcc}=\Mcc\bracks*{v,\Mq}=\Mcc\bracks*{v,\Me}=\braces*{\texttt{all}}$.
		\item
			For $X\in \braces*{A,B}$,
			we define $\Mcc\bracks*{s,X}=\Macc\bracks*{s,X}$.
			For $X\in \braces*{C_1,\ldots,C_n}$, let $\Mq\bracks*{s,X}= \Macc\bracks*{s,X}$.
			For $X\in \braces*{D_1,\ldots,D_5}$, let $\Mq\bracks*{s,X}=\braces*{\texttt{all}}$.
			We also modify Line $1$ and $6$ of $P_v$ in \Cref{fig:userv} to write $\Mcc\bracks*{s,X}$ and $\Mq\bracks*{s,X}$
			instead of $\Macc\bracks*{s,X}$.
		\item
			For $j\in [n]$, let $\Me\bracks*{C_j}$ be initialised to $1$ (where by convention we use $1$ to represent $\mathit{true}$ and $0$ to represent $\mathit{false}$).
			We add the following line before Line $1$ of $P_v$ in \Cref{fig:userv}:
			For $j\in [n]$, measure $C_j$ in the computational basis and flip $\Me\bracks*{C_j}$.
			This new line forbids future entanglement among $C_1,\ldots, C_n$.
	\end{itemize}
	Before $v$ modifies each $\Me\bracks*{C_j}$ to $0$,
	according to the authorisation rule in \Cref{def:1-ent-control},
	$D\bracks*{C_j}$ has to be $1$,
	meaning that $C_j$ is promised to be disentangled from other quantum registers.
	Meanwhile, $\Me\bracks*{D_l}=1$, so each $D_l$ is allowed to be entangled with other quantum registers.
\end{proof}

Finally, let us analyse the efficiency of the $k$-entanglement control model. 

\begin{theorem}[Efficiency of $k$-entanglement control]
	Suppose that $\abs*{\Sub}=M$, $\abs*{\Objc}=N_{\textup{c}}$ and $\abs*{\Objq}=N_{\textup{q}}$,
	then the $k$-entanglement control model uses $O\parens*{M\cdot \parens*{N_{\textup{c}}+N_{\textup{q}}^k}}$ space for access control for $k=1,2$,
	and it takes $O\parens*{x+xN_{\textup{q}}\parens*{k-1}}$ time to handle an access request of length $x$.
\end{theorem}

\begin{proof}
    The space complexity is dominated by that for storing the attributes $\Mcc$, $\Mq$, $\Me$ and $D$ 
    in \Cref{def:1-ent-control,def:2-ent-control}.
    The matrix representation of $\Mcc$ has $\abs*{\Sub}$ rows and $\abs*{\Objc}$ columns,
    while that of $\Mq$ has $\abs*{\Sub}$ rows and $\abs*{\Objq}$ columns.
    $\Me$ and $D$ have $\abs*{\calP_{\leq k}\parens*{\Objq}}=O\parens*{N_{\textup{q}}^k}$ rows and $1$ column, for $k=1,2$.

    The time complexity is dominated by the first for-loop in the post-update rule.
    For $k=1$ (see \Cref{def:1-ent-control}), 
    the loop goes through every $X\in o$ and has time complexity $O\parens*{x}$.
    For $k=2$ (see \Cref{def:2-ent-control}),
    the loop goes through every $X\in o$ and $Y\in \Objq$ and 
    has time complexity $O\parens*{x\cdot N_{\textup{q}}}$.
\end{proof}

\subsection{Comparison of Flexibility}
\label{sub:comparison_of_flexibility}

In this subsection, we compare the flexibility of different models introduced in \Cref{sub:control_of_quantum_operations,sub:control_of_entanglement}.
The results are already visualised in \Cref{fig:compare-flex}.
In practice, one can also consider a hybrid of these models
to achieve a better trade-off between flexibility and efficiency.

Our first theorem shows that for each model in $\braces*{\SUBSYS,\GRP,\ENT}$,
as the parameter $k$ becomes larger, the model becomes more flexible.

\begin{theorem}[Flexibility hierarchy]
	\label{thm:flex-hierarchy}
	For $\mathsf{M}\in \braces*{\SUBSYS, \GRP}$ and any $k\geq 2$,
	or $\mathsf{M}=\ENT$ and $k=2$,
	we have $\mathsf{M}^{k-1} < \mathsf{M}^{k}$.
\end{theorem}

\begin{proof}
    For illustration, we only prove $\SUBSYS^{k-1} < \SUBSYS^{k}$ here,
    and leave the proofs of $\GRP^{k-1} < \GRP^{k}$
    and $\ENT^1 < \ENT^2$ to \Cref{sub:proof_of_flex_hierarchy}.
        
        \begin{enumerate}
		\item 
			We first prove $\SUBSYS^k\not\leq \SUBSYS^{k-1}$.
			The proof idea is using the existence of quantum operations acting non-trivially on $k$ quantum registers.
			%i.e., $\calO\calP\parens*{\calH}\subsetneq \calO\calP\parens*{\calH'}$ for $\calH\subsetneq \calH'$.
			For concreteness,
			let us consider $\mathit{QFT}_k$,
			the quantum Fourier transform on $k$ qubits,
			and use $\texttt{QFT}_k$ to denote the right to implement a $\mathit{QFT}_k$ quantum circuit.

			Let us consider a system $\calA=\parens*{\Sub, \Obj,\Rt,\Atr, \Rule}\in \SUBSYS^k$,
			where $\Sub=\braces*{u,v}$, $\Objc=\emptyset$, $\Objq=\braces*{X_1,\ldots, X_k}$,
			$\Rtc=\emptyset$ and $\Rtq=\braces*{\texttt{QFT}_k}$.
			Attributes $\Mcc,\Mq$ are initialised as follows.
			Since $\Objc=\emptyset$, we set $\Mcc=\emptyset$.
			Denote subsystem $q=\Objq$.
			For $s\in \Sub,o\subseteq \Objq$:
			\begin{equation}
				\Mq\bracks*{s,o} = 
				\begin{cases}
					\braces*{\texttt{QFT}_k}, & s=u\wedge o=q,\\
					\emptyset, & o.w.
				\end{cases}
				\label{eq:M-k-subsystem-prf}
			\end{equation}

			Assume for contradiction that there exists another system $\calA'=\parens*{\Sub, \Obj', \Rt', \Atr', \Rule'}\in \SUBSYS^{k-1}$
			with $\Mcc',\Mq'\in \Atr'$
			such that $\calA'\simeq \calA$.
			We can further assume that $\Objc'=\emptyset$ and $\Rtc'=\emptyset$,
			because otherwise $\calA$ and $\calA'$ will be obviously inequivalent.
			As a result, $\Mq'$ cannot be dynamically modified.

			Consider an execution $\parens*{S,P}$
			with $P_u\equiv \mathit{QFT}_k\bracks*{q}$
			and $P_v\equiv \bot$,
			where $\bot$ denotes termination without doing anything.
			By our construction of $\calA$,
			the history generated by $\parens*{S,P}$ in $\calA$ is simply $\parens*{u,q,\texttt{QFT}_k}$ and is authorised.

			Meanwhile, a request accessing quantum register $o$ in $\calA'$ is only authorised if $\abs*{o}\leq k-1$,
			according to the authorisation rule in \Cref{def:subsys-control}.
			Since $\mathit{QFT}_k$ non-trivially acts on all $k$ quantum registers,
			the history $\alpha$ generated by $\parens*{S,P}$ in $\calA'$ contains more than one requests.
			The above implies that $\alpha\parens*{0}=\parens*{u,o,r}$,
			%there exists $t\in \N$ such that $\alpha\parens*{t}=\parens*{u,o,r}$
			%$\alpha=\alpha',\parens*{w,o,r},\beta$ for some sequences $\alpha',\beta$ of rights,
			where $o\subseteq q$ is a quantum register with $\abs*{o}\leq k-1$ and $r\neq \texttt{QFT}_k$ 
			is the ability to perform some quantum circuit $U\neq \mathit{QFT}_k$.
			As we assume $\calA\simeq \calA'$, $\alpha$ is also authorised.
			
			%Suppose $\alpha=\parens*{s_1,o_1,r_1},\ldots,\parens*{s_m,o_m,r_m}$,
			%where $s_j\in \Sub'$, $o_j\subseteq q$ with $\abs*{o_j}\leq k-1$,
			%and $r_j\in \Rt'$ are abilities to perform quantum circuits $U_j$ which satisfy $U_m\ldots U_1=\mathit{QFT}_k$.
 
			Now we consider another execution $\parens*{S,P'}$
			with $P_{u}'\equiv U\bracks*{o}$
			and $P_v' \equiv \bot$.
			%where $Q$ is a program corresponding to the prefix $\alpha\parens*{0},\ldots, \alpha\parens*{t-1}$ of history $\alpha$.
			The history generated by $\parens*{S,P}$ in $\calA'$ is $\parens*{u,o,r}$, %$\alpha\parens*{0},\ldots, \alpha\parens*{t}$,
			which is therefore authorised as a prefix of the authorised history $\alpha$.
			However, $\parens*{S,P'}$ cannot generate a valid history in $\calA$ because $r\notin\Rt=\braces*{\texttt{QFT}_k}$.
			Hence, we obtain a contradiction and the conclusion follows.
		\item
			Next, we prove that $\SUBSYS^{k-1}\leq \SUBSYS^{k}$.
			Suppose that $\calA=\parens*{\Sub, \Obj,\Rt,\Atr, \Rule}\in \SUBSYS^{k-1}$ with $\Mcc,\Mq\in \Atr$.
			Then, we can define $\calA'=\parens*{\Sub,\Obj,\Rt, \Atr',\Rule}\in \SUBSYS^{k}$ with $\Mcc',\Mq'\in \Atr$
			such that $\Mcc'=\Mcc$ and for any $s\in \Sub, o\subseteq \Objq$:
			\begin{equation*}
				\Mq'\bracks*{s,o}=
				\begin{cases}
					\Mq\bracks*{s,o}, &  \abs*{o}\leq k-1,\\
					\emptyset, & o.w.
				\end{cases}
			\end{equation*}
			It is easy to see that $\calA\simeq \calA'$ from this construction.
	\end{enumerate}
\end{proof}

Our second theorem presents a comparison between the flexibility of subsystem control, group control and entanglement control. 
Let $\SUBSYS=\bigcup_{k} \SUBSYS^k$,
and define $\GRP$ and $\ENT$ similarly.
Let $\SUBSYS^{< N}= \bigcup_{k} \SUBSYS^k \cap \braces*{\calA:\calA\text{ has }\abs*{\Objq}>k}$
be the set of systems with $k$-subsystem control and $k$ less than the size of $\Objq$.

\begin{theorem}[Comparison of Flexibility]
	\label{thm:flex-compare}
	The flexibility of $\SUBSYS$, $\GRP$, $\ENT$ can be compared as follows.
	\begin{enumerate}
		\item
			\label{thm:subsys-not-leq-grp-ent}
			$\SUBSYS \not\leq \GRP,\ENT$.
		\item
			\label{thm:grp-not-leq-ent-subsys}
			$\GRP \not\leq \ENT,\SUBSYS^{<N}$ and $\GRP\leq \SUBSYS$.
		\item
			\label{thm:ent-not-leq-subsys-grp}
			$\ENT \not\leq \SUBSYS,\GRP$.
	\end{enumerate}
\end{theorem}

\begin{proof} For illustration, here,  we only prove Item (3), 
%\Cref{thm:ent-not-leq-subsys-grp} in 
%\Cref{thm:flex-compare},
leaving the proofs of other items %\Cref{thm:subsys-not-leq-grp-ent,thm:grp-not-leq-ent-subsys} 
to \Cref{sub:proof_of_flex_compare}. 
%[Proof of \Cref{thm:flex-compare}~\Cref{thm:ent-not-leq-subsys-grp}]
	Let us only prove $\ENT^1\not\leq \SUBSYS$.
	Then, $\ENT\not\leq \GRP$ easily follows from $\GRP \leq \SUBSYS$ in 
	\Cref{thm:flex-compare}~\Cref{thm:grp-not-leq-ent-subsys}.
	%Like the proof of \Cref{thm:subsys-not-leq-grp-ent},
	The proof idea is essentially using the difference between control of quantum operations and control of entanglement.
	In particular, $\ENT$ uses attribute $D$ to record promises of disentanglement,
	which implies that a system in $\ENT$ 
	can make authorisation decision based on more information about existing entanglements.
	In contrast, during the execution, 
	a system in $\SUBSYS$ cannot (even approximately) distinguish whether entanglement has been established or not,
	of which its authorisation rule is independent.

	Let us consider a system $\calA=\parens*{\Sub,\Obj, \Rt,\Atr, \Rule}\in \ENT$,
	where $\Sub=\braces*{u,v}$, $\Objc=\braces*{\Me}$, $\Objq=\braces*{X_1,X_2}$,
	$\Rtc=\braces*{\texttt{read},\texttt{write}}$, and $\Rtq=\braces*{\texttt{CNOT},\texttt{measure}}$.
	Here, $\texttt{CNOT}$ means the ability to perform a $\mathit{CNOT}$ gate,
	and $\texttt{measure}$ means the ability to perform a computational basis measurement.
	$\Me\in \Objc$ implies that attribute $\Me$ can be dynamically modified by users.
	Attributes $\Mcc,\Mq,\Me,D\in \Atr$ in $\calA$ are initialised as follows. For $s\in \Sub, o\in \Objc$:
	\begin{equation*}
		\label{eq:ent-minus-sub-mc}
		\Mcc\bracks*{s,o}=
		\begin{cases}
			\braces*{\texttt{read},\texttt{write}}, & s=u,\\
			\emptyset,  & o.w.\\
		\end{cases}
	\end{equation*}
	For $s\in \Sub,o\subseteq \Objq$:
	$\Mq\bracks*{s,o}=\braces*{\texttt{CNOT,\texttt{measure}}}$,
	$\Me\bracks*{o}  = \mathit{true}$,
	and $D\bracks*{o}  = \mathit{true}$.

	Assume for contradiction that there exists another system
	$\calA'=\parens*{\Sub, \Obj', \Rt',\Atr', \Rule'}\in \SUBSYS$
	with $\Mcc',\Mq'\in \Atr'$ such that $\calA'\simeq \calA$.
	Note that we assume $\calA'$ has the same $\Sub$ as that of $\calA$ 
	because otherwise they will be obviously inequivalent.

	Consider an execution $\parens*{S,P}$ with
	$P_u\equiv \mathit{disent}\parens*{X_1}$
	and $P_v\equiv \bot$, where $\mathit{disent}\parens*{X_1}$ means
	to modify attributes such that quantum register $X_1$ is disentangled from others,
	and $\bot$ denotes termination without doing anything.
	By our construction of $\calA$,
	the history generated by $\parens*{S,P}$ in $\calA$ is 
	$\parens*{u,\Me\bracks*{X_1},\texttt{read}},\parens*{u,\Me\bracks*{X_1},\texttt{write}}$ and is authorised.
	Note that during the execution, the value of $\Me\bracks*{X_1}$ will be modified from $\mathit{true}$ to $\mathit{false}$,
	and the value of $D\bracks*{X_1}$ is always $\mathit{true}$.

	Consider another execution $\parens*{S,P'}$ with
        \begin{equation*}
            P_u'\equiv H\bracks*{X_1};\mathit{CNOT}\bracks*{X_1,X_2};\mathit{disent}\parens*{X_1}
        \end{equation*}
	and $P_v'\equiv \bot$.
	The history generated by $\parens*{S,P'}$ in $\calA$ is 
	\begin{equation*}
		\parens*{u,\braces*{X_1},\texttt{H}},\parens*{u,\braces*{X_1,X_2},\texttt{CNOT}}, \parens*{u,\Me\bracks*{X_1},\texttt{read}}, \parens*{u, \Me\bracks*{X_1}, \texttt{write}}.
	\end{equation*}
	This history is unauthorised because the post-update rule
	in \Cref{def:1-ent-control} modifies $D\bracks*{X_1}$ and $D\bracks*{X_2}$ to $\mathit{false}$ after the second request,
	when the quantum state of $X_1,X_2$ becomes $\frac{1}{\sqrt{2}}\parens*{\ket{0}_{X_1}\ket{0}_{X_2}+\ket{1}_{X_1}\ket{1}_{X_2}}$,
	which is entangled.
	Then, the last request modifying $\Me\bracks*{X_1}$ will be denied by the authorisation rule.

	On the other hand, suppose that the histories generated by $\parens*{S,P}$ and $\parens*{S,P'}$ in $\calA'$
	are $\alpha$ and $\alpha'$, respectively.
	Since we assume $\calA\simeq \calA'$, by \Cref{def:equivalent-sys},
	$\alpha$ is authorised and $\alpha'$ is unauthorised.
	Observe that $\alpha$ is a suffix of $\alpha'$:
	we have $\alpha'=\beta,\alpha$ for some sequence $\beta$ of requests generated from executing $H\bracks*{X_1};\mathit{CNOT}\bracks*{X_1,X_2}$ in $P_u'$.
	This is because the authorisation rule of $\calA'$ (see \Cref{def:subsys-control}) is based on attributes $\Mcc',\Mq'$, 
	which are unchanged by $\beta$.
	Further, this implies $\alpha'$ should be authorised,
	because the prefix $\beta$ does not change $\Mcc',\Mq'$ and will not affect whether the suffix $\alpha$ is authorised.
	Hence, we obtain a contradiction and the conclusion follows.

	%As $u$ only executes $\mathit{disent}\parens*{X_1}$,
	%Suppose that $\alpha$ does not contain redundant requests 
	%(e.g., $\alpha\parens*{0}=\parens*{u,X_1,\texttt{H}}$ and $\alpha\parens*{1}=\parens*{u,X_1,\texttt{H}}$,
	%with \texttt{Had} being the ability to perform a Hadamard gate).
	%Then, for any $t\in \N$, we must have $\alpha\parens*{t}=\parens*{u,o,r}$
	%for some classical object $o\in \Objc$.
	%This is because if $\alpha$ can non-trivially modifies (e.g., $o=X_1$ and $r=\texttt{measure}$)
	%the state of quantum objects, then the history will be inconsistent with the semantics of the program $P$.

	%However, the history generated by $\parens*{S,P'}$ in $\calA'$ 
	%should be $\parens*{u,\braces*{X_1,X_2}, \texttt{CNOT}},\alpha$, 
	%the concatenation of $\parens*{u,\braces*{X_1,X_2}, \texttt{CNOT}}$ and the history generated by $\parens*{S,P}$ in $\calA'$.
	%Due to our previous observation, this history should be authorised.
	%Thus, $\calA\not\simeq \calA'$ and the conclusion follows.
\end{proof}

\section{Related Works}
\label{sec:related_works}

\paragraph{Quantum Access Control}

The work~\cite{YFY13} first studied access control in quantum computing from the perspective of information-flow security. Their observation that rights should be specified for quantum subsystems motivated our \Cref{def:core-model} and the $k=\abs*{\Objq}$ case of \Cref{def:subsys-control}, as mentioned in \Cref{sec:protection_access_control_in_quantum_computing}. However, they did not provide any explicit scenario of access control showing \textit{entanglement can leak secret beyond direct communication}.
In contrast, our \Cref{sec:scenario_threat_from_quantum_entanglement} presents the \textit{first} explicit scenario of how a classically secure access control system becomes insecure when adapted to the quantum setting, with a rigorous proof.
This identification of threat from entanglement enables us to design effective quantum access control models and analyse them in \Cref{sec:protection_access_control_in_quantum_computing}.
Other related work~\cite{GI19} has studied entanglement accessibility in the context of the quantum internet, a different topic from the access control in computer security we address here.

\paragraph{Quantum Operating Systems}

Operating systems are a major area where access control mechanisms have been extensively studied and implemented.
In quantum computing, there have been already numerous efforts devoted to tackle specific issues relevant to operating systems. 
These include task decomposition (due to the scarcity of qubits in existing quantum hardware, and typically via quantum circuit cutting or knitting, e.g., \cite{BSS16,PHOW20,MF21,TTS+21,EMG+22,PS24,LGB+24}), job scheduling (e.g., \cite{RSMC21,LD24,OC24,LSL+25}), multiprogramming (e.g., \cite{DTNQ19,LD21,OSV22,RGH+21,NT23,KSZ+23,LD24}), memory management (e.g., \cite{HPJ+15,MSR+19,DPW20,LWS+23}), and concurrency (e.g., \cite{JL04,GN05,FDY11,YZLF22,FLY22,HSHT21,WZLSXD22,TKM12,AGM17,AGM17,ZY24c}).
Meanwhile, some other works have considered more holistic approaches to designing quantum operating systems~\cite{CWB17,HLAC20,KWH+21,GRTB24,DIV+25}.
It can be expected that quantum access control (considered in this paper) will become more indispensable to the security of quantum and classical-quantum hybrid computer systems 
when various quantum operating systems are deployed in the future.

\paragraph{Security and Bell-Type Inequalities}

The violation of Bell-type inequalities (including the Mermin inequality~\cite{Mermin90} used in this paper),
which essentially reflects the exotic nature of quantum mechanics, 
has been applied in a number of security protocols that utilize quantum properties.
For example, the celebrated E91 protocol proposed in~\cite{Ekert91} modifies the Bell test 
to detect eavesdropping and securely generate private keys for cryptography.
This technique was later greatly extended into a line of works on device-independent quantum cryptography~\cite{MY98,BHK05,PAB+09,RUV13,VV14,MS16b,ARV19,DFR20,DF19}. 
Similar ideas have also been employed in randomness expansion~\cite{Colbeck09,PAM+10,VV12,CY14,MS16b,MS17}
and randomness amplification~\cite{CR12,GMD+13,CSW15,KA20}.
Most of the above works focus on quantum cryptography and 
leverage the quantum entanglement as an advantage for enhancing security.
In contrast, this paper considers the access control security of quantum computer systems,  identifies entanglement as a source of security threats, and proposes new access control models to protect against such threats.

\section{Conclusion}
\label{sec:discussion}

We reveal that the access control security can be threatened if existing computer systems integrate with quantum computing. This is demonstrated by presenting the first explicit scenario of a security breach when a classically secure access control system is straightforwardly adapted to the quantum setting. The threat essentially comes from the phenomenon of quantum entanglement. To address such threat, we propose several new models of quantum access control, including subsystem control, group control and entanglement control. Their security, flexibility and efficiency are rigorously analysed. While all the proposed models are secure against threats from entanglement, their flexibility and efficiency vary. In practice, specific requirements for the latter two factors determine which model is the most suitable for practical uses, and one can also consider a hybrid of these models.

The research reported in this paper is merely the first step toward access control of quantum computers. In the following, we list several topics for future research. 
Firstly, to prevent from security breach from quantum entanglement,
an immediate next step is to integrate new quantum access control mechanisms into the design of future classical-quantum hybrid systems (including quantum-centric supercomputing systems~\cite{Gambetta22a,Gambetta22b,AAB24,MCG24,PC24}).
This involves further refining the quantum access control models proposed in this paper to accommodate the actual requirements of the specific computer system to be protected.
Secondly, as mentioned in \Cref{sub:access_control,sub:subsystem_control},
for simplicity we have not considered how attributes in our proposed models are stored.
Like in the classical case~\cite{SS94}, 
it is worth investigating how to store the attributes using more efficient data structures,
whose design might also leverage the unique properties of quantum systems.
Thirdly, as mentioned in \Cref{sub:control_of_entanglement}, 
in the model $\ENT$ (see \Cref{def:1-ent-control,def:2-ent-control}) for entanglement control,
we focus on a single approach to recording approximated knowledge about existing entanglements.
This approximation is coarse-grain: only complete measurements are regarded as promise of disentanglement,
while any other quantum operations involving multiple quantum registers are assumed to create potential entanglement.
An interesting question is if there are other approaches that offers finer approximations and greater flexibility
(perhaps at the cost of reduced efficiency).

%Brute-force measurement.
%Discuss quantum monogamy: transform proof of separability to proof of entanglement.

\newpage
\bibliographystyle{ACM-Reference-Format}
\bibliography{main}


\begin{thebibliography}{124}

%%% the Web address.

\ifx \showCODEN    \undefined \def \showCODEN     #1{\unskip}     \fi
\ifx \showDOI      \undefined \def \showDOI       #1{#1}\fi
\ifx \showISBNx    \undefined \def \showISBNx     #1{\unskip}     \fi
\ifx \showISBNxiii \undefined \def \showISBNxiii  #1{\unskip}     \fi
\ifx \showISSN     \undefined \def \showISSN      #1{\unskip}     \fi
\ifx \showLCCN     \undefined \def \showLCCN      #1{\unskip}     \fi
\ifx \shownote     \undefined \def \shownote      #1{#1}          \fi
\ifx \showarticletitle \undefined \def \showarticletitle #1{#1}   \fi
\ifx \showURL      \undefined \def \showURL       {\relax}        \fi
% invisible to TeX
\providecommand\bibfield[2]{#2}
\providecommand\bibinfo[2]{#2}
\providecommand\natexlab[1]{#1}
\providecommand\showeprint[2][]{arXiv:#2}

\bibitem[Aharonov et~al\mbox{.}(2008)]%
        {ABE08}
\bibfield{author}{\bibinfo{person}{Dorit Aharonov}, \bibinfo{person}{Michael {Ben-Or}}, {and} \bibinfo{person}{Elad Eban}.} \bibinfo{year}{2008}\natexlab{}.
\newblock \bibinfo{title}{Interactive proofs for quantum computations}.
\newblock
\newblock
\showeprint[arxiv]{0810.5375}~[quant-ph]


\bibitem[Aharonov et~al\mbox{.}(2017)]%
        {AGM17}
\bibfield{author}{\bibinfo{person}{Dorit Aharonov}, \bibinfo{person}{Maor Ganz}, {and} \bibinfo{person}{Loick Magnin}.} \bibinfo{year}{2017}\natexlab{}.
\newblock \bibinfo{title}{Dining philosophers, leader election and ring size problems, in the quantum setting}.
\newblock
\newblock
\showeprint[arxiv]{1707.01187}~[quant-ph]


\bibitem[Alexeev et~al\mbox{.}(2024)]%
        {AAB24}
\bibfield{author}{\bibinfo{person}{Yuri Alexeev}, \bibinfo{person}{Maximilian Amsler}, \bibinfo{person}{Marco~Antonio Barroca}, \bibinfo{person}{Sanzio Bassini}, \bibinfo{person}{Torey Battelle}, \bibinfo{person}{Daan Camps}, \bibinfo{person}{David Casanova}, \bibinfo{person}{Young~Jay Choi}, \bibinfo{person}{Frederic~T Chong}, \bibinfo{person}{Charles Chung}, {et~al\mbox{.}}} \bibinfo{year}{2024}\natexlab{}.
\newblock \showarticletitle{Quantum-centric supercomputing for materials science: A perspective on challenges and future directions}.
\newblock \bibinfo{journal}{\emph{Future Generation Computer Systems}}  \bibinfo{volume}{160} (\bibinfo{year}{2024}), \bibinfo{pages}{666--710}.
\newblock


\bibitem[Arnon-Friedman et~al\mbox{.}(2019)]%
        {ARV19}
\bibfield{author}{\bibinfo{person}{Rotem Arnon-Friedman}, \bibinfo{person}{Renato Renner}, {and} \bibinfo{person}{Thomas Vidick}.} \bibinfo{year}{2019}\natexlab{}.
\newblock \showarticletitle{Simple and tight device-independent security proofs}.
\newblock \bibinfo{journal}{\emph{SIAM J. Comput.}} \bibinfo{volume}{48}, \bibinfo{number}{1} (\bibinfo{year}{2019}), \bibinfo{pages}{181--225}.
\newblock


\bibitem[Aspect et~al\mbox{.}(1982)]%
        {ADR82}
\bibfield{author}{\bibinfo{person}{Alain Aspect}, \bibinfo{person}{Jean Dalibard}, {and} \bibinfo{person}{G{\'e}rard Roger}.} \bibinfo{year}{1982}\natexlab{}.
\newblock \showarticletitle{Experimental test of {Bell's} inequalities using {time-varying} analyzers}.
\newblock \bibinfo{journal}{\emph{Physical Review Letters}} \bibinfo{volume}{49}, \bibinfo{number}{25} (\bibinfo{year}{1982}), \bibinfo{pages}{1804}.
\newblock


\bibitem[Barrett et~al\mbox{.}(2005)]%
        {BHK05}
\bibfield{author}{\bibinfo{person}{Jonathan Barrett}, \bibinfo{person}{Lucien Hardy}, {and} \bibinfo{person}{Adrian Kent}.} \bibinfo{year}{2005}\natexlab{}.
\newblock \showarticletitle{No signaling and quantum key distribution}.
\newblock \bibinfo{journal}{\emph{Physical Review Letters}} \bibinfo{volume}{95}, \bibinfo{number}{1} (\bibinfo{year}{2005}), \bibinfo{pages}{010503}.
\newblock


\bibitem[Bell and {La Padula}(1976)]%
        {BLP76}
\bibfield{author}{\bibinfo{person}{D.~Elliott Bell} {and} \bibinfo{person}{Leonard~J. {La Padula}}.} \bibinfo{year}{1976}\natexlab{}.
\newblock \bibinfo{booktitle}{\emph{Secure computer system: Unified exposition and {Multics} interpretation}}.
\newblock \bibinfo{type}{{T}echnical {R}eport} ESD-TR-75-306. \bibinfo{institution}{The MITRE Corporation, Bedford, MA}.
\newblock


\bibitem[Bell(1964)]%
        {Bell64}
\bibfield{author}{\bibinfo{person}{John~S. Bell}.} \bibinfo{year}{1964}\natexlab{}.
\newblock \showarticletitle{On the {Einstein Podolsky Rosen} paradox}.
\newblock \bibinfo{journal}{\emph{Physics Physique Fizika}} \bibinfo{volume}{1}, \bibinfo{number}{3} (\bibinfo{year}{1964}), \bibinfo{pages}{195}.
\newblock


\bibitem[Bertino et~al\mbox{.}(2000)]%
        {BBF00}
\bibfield{author}{\bibinfo{person}{Elisa Bertino}, \bibinfo{person}{Piero~Andrea Bonatti}, {and} \bibinfo{person}{Elena Ferrari}.} \bibinfo{year}{2000}\natexlab{}.
\newblock \showarticletitle{TRBAC: A temporal role-based access control model}. In \bibinfo{booktitle}{\emph{Proceedings of the fifth ACM workshop on Role-based access control}}. \bibinfo{pages}{21--30}.
\newblock


\bibitem[Biba(1977)]%
        {Biba77}
\bibfield{author}{\bibinfo{person}{Kenneth~J. Biba}.} \bibinfo{year}{1977}\natexlab{}.
\newblock \bibinfo{booktitle}{\emph{Integrity considerations for secure computer systems}}.
\newblock \bibinfo{type}{{T}echnical {R}eport} ESD-TR-76-372. \bibinfo{institution}{The MITRE Corporation, Bedford, MA}.
\newblock


\bibitem[Bluvstein et~al\mbox{.}(2023)]%
        {BEG+23}
\bibfield{author}{\bibinfo{person}{Dolev Bluvstein}, \bibinfo{person}{Simon~J. Evered}, \bibinfo{person}{Alexandra~A. Geim}, \bibinfo{person}{Sophie~H. Li}, \bibinfo{person}{Hengyun Zhou}, \bibinfo{person}{Tom Manovitz}, \bibinfo{person}{Sepehr Ebadi}, \bibinfo{person}{Madelyn Cain}, \bibinfo{person}{Marcin Kalinowski}, \bibinfo{person}{Dominik Hangleiter}, \bibinfo{person}{J.~Pablo Bonilla~Ataides}, \bibinfo{person}{Nishad Maskara}, \bibinfo{person}{Iris Cong}, \bibinfo{person}{Xun Gao}, \bibinfo{person}{Pedro Sales~Rodriguez}, \bibinfo{person}{Thomas Karolyshyn}, \bibinfo{person}{Giulia Semeghini}, \bibinfo{person}{Michael~J. Gullans}, \bibinfo{person}{Markus Greiner}, \bibinfo{person}{Vladan Vuletić}, {and} \bibinfo{person}{Mikhail~D. Lukin}.} \bibinfo{year}{2023}\natexlab{}.
\newblock \showarticletitle{Logical quantum processor based on reconfigurable atom arrays}.
\newblock \bibinfo{journal}{\emph{Nature}} \bibinfo{volume}{626}, \bibinfo{number}{7997} (\bibinfo{year}{2023}), \bibinfo{pages}{58–65}.
\newblock


\bibitem[Brassard et~al\mbox{.}(2005)]%
        {BBT05}
\bibfield{author}{\bibinfo{person}{Gilles Brassard}, \bibinfo{person}{Anne Broadbent}, {and} \bibinfo{person}{Alain Tapp}.} \bibinfo{year}{2005}\natexlab{}.
\newblock \showarticletitle{Recasting {Mermin}'s multi-player game into the framework of pseudo-telepathy}.
\newblock \bibinfo{journal}{\emph{Quantum Information and Computation}} \bibinfo{volume}{5}, \bibinfo{number}{7} (\bibinfo{year}{2005}), \bibinfo{pages}{538–550}.
\newblock


\bibitem[Bravyi et~al\mbox{.}(2016)]%
        {BSS16}
\bibfield{author}{\bibinfo{person}{Sergey Bravyi}, \bibinfo{person}{Graeme Smith}, {and} \bibinfo{person}{John~A. Smolin}.} \bibinfo{year}{2016}\natexlab{}.
\newblock \showarticletitle{Trading classical and quantum computational resources}.
\newblock \bibinfo{journal}{\emph{Physical Review X}} \bibinfo{volume}{6}, \bibinfo{number}{2} (\bibinfo{year}{2016}), \bibinfo{pages}{021043}.
\newblock


\bibitem[Broadbent et~al\mbox{.}(2009)]%
        {BFK09}
\bibfield{author}{\bibinfo{person}{Anne Broadbent}, \bibinfo{person}{Joseph~F. Fitzsimons}, {and} \bibinfo{person}{Elham Kashefi}.} \bibinfo{year}{2009}\natexlab{}.
\newblock \showarticletitle{Universal blind quantum computation}. In \bibinfo{booktitle}{\emph{2009 50th annual IEEE symposium on foundations of computer science}}. \bibinfo{pages}{517--526}.
\newblock


\bibitem[Childs(2005)]%
        {Childs05}
\bibfield{author}{\bibinfo{person}{Andrew~M. Childs}.} \bibinfo{year}{2005}\natexlab{}.
\newblock \showarticletitle{Secure assisted quantum computation}.
\newblock \bibinfo{journal}{\emph{Quantum Information \& Computation}} \bibinfo{volume}{5}, \bibinfo{number}{6} (\bibinfo{year}{2005}), \bibinfo{pages}{456--466}.
\newblock


\bibitem[Chung et~al\mbox{.}(2015)]%
        {CSW15}
\bibfield{author}{\bibinfo{person}{{Kai-Min} Chung}, \bibinfo{person}{Yaoyun Shi}, {and} \bibinfo{person}{Xiaodi Wu}.} \bibinfo{year}{2015}\natexlab{}.
\newblock \bibinfo{title}{Physical randomness extractors: Generating random numbers with minimal assumptions}.
\newblock
\newblock
\showeprint[arxiv]{1402.4797}~[quant-ph]


\bibitem[Clark and Wilson(1987)]%
        {CW87}
\bibfield{author}{\bibinfo{person}{David~D. Clark} {and} \bibinfo{person}{David~R. Wilson}.} \bibinfo{year}{1987}\natexlab{}.
\newblock \showarticletitle{A comparison of commercial and military computer security policies}. In \bibinfo{booktitle}{\emph{1987 IEEE Symposium on Security and Privacy}}. \bibinfo{pages}{184--184}.
\newblock


\bibitem[Clauser et~al\mbox{.}(1969)]%
        {CHSH69}
\bibfield{author}{\bibinfo{person}{John~F. Clauser}, \bibinfo{person}{Michael~A. Horne}, \bibinfo{person}{Abner Shimony}, {and} \bibinfo{person}{Richard~A. Holt}.} \bibinfo{year}{1969}\natexlab{}.
\newblock \showarticletitle{Proposed experiment to test local {hidden-variable} theories}.
\newblock \bibinfo{journal}{\emph{Physical Review Letters}} \bibinfo{volume}{23}, \bibinfo{number}{15} (\bibinfo{year}{1969}), \bibinfo{pages}{880}.
\newblock


\bibitem[Colbeck(2009)]%
        {Colbeck09}
\bibfield{author}{\bibinfo{person}{Roger Colbeck}.} \bibinfo{year}{2009}\natexlab{}.
\newblock \emph{\bibinfo{title}{Quantum and relativistic protocols for secure {multi-party} computation}}.
\newblock \bibinfo{thesistype}{Ph.\,D. Dissertation}. \bibinfo{school}{University of Cambridge}.
\newblock


\bibitem[Colbeck and Renner(2012)]%
        {CR12}
\bibfield{author}{\bibinfo{person}{Roger Colbeck} {and} \bibinfo{person}{Renato Renner}.} \bibinfo{year}{2012}\natexlab{}.
\newblock \showarticletitle{Free randomness can be amplified}.
\newblock \bibinfo{journal}{\emph{Nature Physics}} \bibinfo{volume}{8}, \bibinfo{number}{6} (\bibinfo{year}{2012}), \bibinfo{pages}{450--453}.
\newblock


\bibitem[{Corrigan-Gibbs} et~al\mbox{.}(2017)]%
        {CWB17}
\bibfield{author}{\bibinfo{person}{Henry {Corrigan-Gibbs}}, \bibinfo{person}{David~J. Wu}, {and} \bibinfo{person}{Dan Boneh}.} \bibinfo{year}{2017}\natexlab{}.
\newblock \showarticletitle{Quantum operating systems}. In \bibinfo{booktitle}{\emph{Proceedings of the 16th Workshop on Hot Topics in Operating Systems}}. \bibinfo{pages}{76--81}.
\newblock


\bibitem[Coudron and Yuen(2014)]%
        {CY14}
\bibfield{author}{\bibinfo{person}{Matthew Coudron} {and} \bibinfo{person}{Henry Yuen}.} \bibinfo{year}{2014}\natexlab{}.
\newblock \showarticletitle{Infinite randomness expansion with a constant number of devices}. In \bibinfo{booktitle}{\emph{Proceedings of the forty-sixth annual ACM symposium on Theory of computing}}. \bibinfo{pages}{427--436}.
\newblock


\bibitem[Dai et~al\mbox{.}(2020)]%
        {DPW20}
\bibfield{author}{\bibinfo{person}{Wenhan Dai}, \bibinfo{person}{Tianyi Peng}, {and} \bibinfo{person}{Moe~Z. Win}.} \bibinfo{year}{2020}\natexlab{}.
\newblock \showarticletitle{Quantum queuing delay}.
\newblock \bibinfo{journal}{\emph{IEEE Journal on Selected Areas in Communications}} \bibinfo{volume}{38}, \bibinfo{number}{3} (\bibinfo{year}{2020}), \bibinfo{pages}{605--618}.
\newblock


\bibitem[Das et~al\mbox{.}(2019)]%
        {DTNQ19}
\bibfield{author}{\bibinfo{person}{Poulami Das}, \bibinfo{person}{Swamit~S. Tannu}, \bibinfo{person}{Prashant~J. Nair}, {and} \bibinfo{person}{Moinuddin Qureshi}.} \bibinfo{year}{2019}\natexlab{}.
\newblock \showarticletitle{A case for {multi-programming} quantum computers}. In \bibinfo{booktitle}{\emph{Proceedings of the 52nd Annual IEEE/ACM International Symposium on Microarchitecture}}. \bibinfo{pages}{291--303}.
\newblock


\bibitem[Delle~Donne et~al\mbox{.}(2025)]%
        {DIV+25}
\bibfield{author}{\bibinfo{person}{C. Delle~Donne}, \bibinfo{person}{M. Iuliano}, \bibinfo{person}{B. {van der Vecht}}, \bibinfo{person}{G.~M. Ferreira}, \bibinfo{person}{H. Jirovsk{\'a}}, \bibinfo{person}{T.~J.~W. {van der Steenhoven}}, \bibinfo{person}{A. Dahlberg}, \bibinfo{person}{M. Skrzypczyk}, \bibinfo{person}{D. Fioretto}, \bibinfo{person}{M. Teller}, {et~al\mbox{.}}} \bibinfo{year}{2025}\natexlab{}.
\newblock \showarticletitle{An operating system for executing applications on quantum network nodes}.
\newblock \bibinfo{journal}{\emph{Nature}} \bibinfo{volume}{639}, \bibinfo{number}{8054} (\bibinfo{year}{2025}), \bibinfo{pages}{321--328}.
\newblock


\bibitem[Denning(1976)]%
        {Denning76}
\bibfield{author}{\bibinfo{person}{Dorothy~E. Denning}.} \bibinfo{year}{1976}\natexlab{}.
\newblock \showarticletitle{A lattice model of secure information flow}.
\newblock \bibinfo{journal}{\emph{Commun. ACM}} \bibinfo{volume}{19}, \bibinfo{number}{5} (\bibinfo{year}{1976}), \bibinfo{pages}{236--243}.
\newblock


\bibitem[Denning(1971)]%
        {Denning71}
\bibfield{author}{\bibinfo{person}{Peter~J. Denning}.} \bibinfo{year}{1971}\natexlab{}.
\newblock \showarticletitle{Third generation computer systems}.
\newblock \bibinfo{journal}{\emph{ACM Computing Surveys (CSUR)}} \bibinfo{volume}{3}, \bibinfo{number}{4} (\bibinfo{year}{1971}), \bibinfo{pages}{175--216}.
\newblock


\bibitem[Downs et~al\mbox{.}(1985)]%
        {DRKJ85}
\bibfield{author}{\bibinfo{person}{Deborah~D. Downs}, \bibinfo{person}{Jerzy~R. Rub}, \bibinfo{person}{Kenneth~C. Kung}, {and} \bibinfo{person}{Carole~S. Jordan}.} \bibinfo{year}{1985}\natexlab{}.
\newblock \showarticletitle{Issues in discretionary access control}. In \bibinfo{booktitle}{\emph{1985 IEEE symposium on security and privacy}}. \bibinfo{pages}{208--208}.
\newblock


\bibitem[Dunjko et~al\mbox{.}(2012)]%
        {DKL12}
\bibfield{author}{\bibinfo{person}{Vedran Dunjko}, \bibinfo{person}{Elham Kashefi}, {and} \bibinfo{person}{Anthony Leverrier}.} \bibinfo{year}{2012}\natexlab{}.
\newblock \showarticletitle{Blind quantum computing with weak coherent pulses}.
\newblock \bibinfo{journal}{\emph{Physical Review Letters}} \bibinfo{volume}{108}, \bibinfo{number}{20} (\bibinfo{year}{2012}), \bibinfo{pages}{200502}.
\newblock


\bibitem[Dupuis and Fawzi(2019)]%
        {DF19}
\bibfield{author}{\bibinfo{person}{Frederic Dupuis} {and} \bibinfo{person}{Omar Fawzi}.} \bibinfo{year}{2019}\natexlab{}.
\newblock \showarticletitle{Entropy accumulation with improved {second-order} term}.
\newblock \bibinfo{journal}{\emph{IEEE Transactions on Information Theory}} \bibinfo{volume}{65}, \bibinfo{number}{11} (\bibinfo{year}{2019}), \bibinfo{pages}{7596–7612}.
\newblock


\bibitem[Dupuis et~al\mbox{.}(2020)]%
        {DFR20}
\bibfield{author}{\bibinfo{person}{Frederic Dupuis}, \bibinfo{person}{Omar Fawzi}, {and} \bibinfo{person}{Renato Renner}.} \bibinfo{year}{2020}\natexlab{}.
\newblock \showarticletitle{Entropy accumulation}.
\newblock \bibinfo{journal}{\emph{Communications in Mathematical Physics}} \bibinfo{volume}{379}, \bibinfo{number}{3} (\bibinfo{year}{2020}), \bibinfo{pages}{867--913}.
\newblock


\bibitem[Eddins et~al\mbox{.}(2022)]%
        {EMG+22}
\bibfield{author}{\bibinfo{person}{Andrew Eddins}, \bibinfo{person}{Mario Motta}, \bibinfo{person}{Tanvi~P. Gujarati}, \bibinfo{person}{Sergey Bravyi}, \bibinfo{person}{Antonio Mezzacapo}, \bibinfo{person}{Charles Hadfield}, {and} \bibinfo{person}{Sarah Sheldon}.} \bibinfo{year}{2022}\natexlab{}.
\newblock \showarticletitle{Doubling the size of quantum simulators by entanglement forging}.
\newblock \bibinfo{journal}{\emph{PRX Quantum}} \bibinfo{volume}{3}, \bibinfo{number}{1} (\bibinfo{year}{2022}), \bibinfo{pages}{010309}.
\newblock


\bibitem[Ekert(1991)]%
        {Ekert91}
\bibfield{author}{\bibinfo{person}{Artur~K. Ekert}.} \bibinfo{year}{1991}\natexlab{}.
\newblock \showarticletitle{Quantum cryptography based on {Bell}’s theorem}.
\newblock \bibinfo{journal}{\emph{Physical Review Letters}} \bibinfo{volume}{67}, \bibinfo{number}{6} (\bibinfo{year}{1991}), \bibinfo{pages}{661}.
\newblock


\bibitem[Feng et~al\mbox{.}(2012)]%
        {FDY11}
\bibfield{author}{\bibinfo{person}{Yuan Feng}, \bibinfo{person}{Runyao Duan}, {and} \bibinfo{person}{Mingsheng Ying}.} \bibinfo{year}{2012}\natexlab{}.
\newblock \showarticletitle{Bisimulation for quantum processes}.
\newblock \bibinfo{journal}{\emph{ACM Transactions on Programming Languages and Systems (TOPLAS)}} \bibinfo{volume}{34}, \bibinfo{number}{4} (\bibinfo{year}{2012}), \bibinfo{pages}{1--43}.
\newblock


\bibitem[Feng et~al\mbox{.}(2022)]%
        {FLY22}
\bibfield{author}{\bibinfo{person}{Yuan Feng}, \bibinfo{person}{Sanjiang Li}, {and} \bibinfo{person}{Mingsheng Ying}.} \bibinfo{year}{2022}\natexlab{}.
\newblock \showarticletitle{Verification of distributed quantum programs}.
\newblock \bibinfo{journal}{\emph{ACM Transactions on Computational Logic (TOCL)}} \bibinfo{volume}{23}, \bibinfo{number}{3} (\bibinfo{year}{2022}), \bibinfo{pages}{1--40}.
\newblock


\bibitem[Ferraiolo and Kuhn(1992)]%
        {FK92}
\bibfield{author}{\bibinfo{person}{David~F. Ferraiolo} {and} \bibinfo{person}{D.~Richard Kuhn}.} \bibinfo{year}{1992}\natexlab{}.
\newblock \showarticletitle{{Role-based} access controls}.
\newblock \bibinfo{journal}{\emph{15th National Computer Security Conference (1992)}}, \bibinfo{pages}{554--563}.
\newblock


\bibitem[Ferraiolo et~al\mbox{.}(2001)]%
        {FSGKC01}
\bibfield{author}{\bibinfo{person}{David~F. Ferraiolo}, \bibinfo{person}{Ravi Sandhu}, \bibinfo{person}{Serban Gavrila}, \bibinfo{person}{D.~Richard Kuhn}, {and} \bibinfo{person}{Ramaswamy Chandramouli}.} \bibinfo{year}{2001}\natexlab{}.
\newblock \showarticletitle{Proposed {NIST} standard for {role-based} access control}.
\newblock \bibinfo{journal}{\emph{ACM Transactions on Information and System Security (TISSEC)}} \bibinfo{volume}{4}, \bibinfo{number}{3} (\bibinfo{year}{2001}), \bibinfo{pages}{224--274}.
\newblock


\bibitem[Fitzsimons(2017)]%
        {Fitz17}
\bibfield{author}{\bibinfo{person}{Joseph~F. Fitzsimons}.} \bibinfo{year}{2017}\natexlab{}.
\newblock \showarticletitle{Private quantum computation: an introduction to blind quantum computing and related protocols}.
\newblock \bibinfo{journal}{\emph{npj Quantum Information}} \bibinfo{volume}{3}, \bibinfo{number}{1} (\bibinfo{year}{2017}), \bibinfo{pages}{23}.
\newblock


\bibitem[Fitzsimons and Kashefi(2017)]%
        {FK17}
\bibfield{author}{\bibinfo{person}{Joseph~F. Fitzsimons} {and} \bibinfo{person}{Elham Kashefi}.} \bibinfo{year}{2017}\natexlab{}.
\newblock \showarticletitle{Unconditionally verifiable blind quantum computation}.
\newblock \bibinfo{journal}{\emph{Physical Review A}} \bibinfo{volume}{96}, \bibinfo{number}{1} (\bibinfo{year}{2017}), \bibinfo{pages}{012303}.
\newblock


\bibitem[Freedman and Clauser(1972)]%
        {FC72}
\bibfield{author}{\bibinfo{person}{Stuart~J. Freedman} {and} \bibinfo{person}{John~F. Clauser}.} \bibinfo{year}{1972}\natexlab{}.
\newblock \showarticletitle{Experimental test of local {hidden-variable} theories}.
\newblock \bibinfo{journal}{\emph{Physical Review Letters}} \bibinfo{volume}{28}, \bibinfo{number}{14} (\bibinfo{year}{1972}), \bibinfo{pages}{938}.
\newblock


\bibitem[Gallego et~al\mbox{.}(2013)]%
        {GMD+13}
\bibfield{author}{\bibinfo{person}{Rodrigo Gallego}, \bibinfo{person}{Lluis Masanes}, \bibinfo{person}{Gonzalo {De La Torre}}, \bibinfo{person}{Chirag Dhara}, \bibinfo{person}{Leandro Aolita}, {and} \bibinfo{person}{Antonio Ac{\'\i}n}.} \bibinfo{year}{2013}\natexlab{}.
\newblock \showarticletitle{Full randomness from arbitrarily deterministic events}.
\newblock \bibinfo{journal}{\emph{Nature communications}} \bibinfo{volume}{4}, \bibinfo{number}{1} (\bibinfo{year}{2013}), \bibinfo{pages}{2654}.
\newblock


\bibitem[Gambetta(2022a)]%
        {Gambetta22b}
\bibfield{author}{\bibinfo{person}{Jay Gambetta}.} \bibinfo{year}{2022}\natexlab{a}.
\newblock \showarticletitle{Expanding the IBM Quantum roadmap to anticipate the future of quantum-centric supercomputing}.
\newblock \bibinfo{journal}{\emph{IBM Research Blog}} (\bibinfo{year}{2022}).
\newblock


\bibitem[Gambetta(2022b)]%
        {Gambetta22a}
\bibfield{author}{\bibinfo{person}{Jay Gambetta}.} \bibinfo{year}{2022}\natexlab{b}.
\newblock \showarticletitle{Quantum-centric supercomputing: {The} next wave of computing}.
\newblock \bibinfo{journal}{\emph{IBM Research Blog}} (\bibinfo{year}{2022}).
\newblock


\bibitem[Gay and Nagarajan(2005)]%
        {GN05}
\bibfield{author}{\bibinfo{person}{Simon~J. Gay} {and} \bibinfo{person}{Rajagopal Nagarajan}.} \bibinfo{year}{2005}\natexlab{}.
\newblock \showarticletitle{Communicating quantum processes}. In \bibinfo{booktitle}{\emph{Proceedings of the 32nd ACM SIGPLAN-SIGACT Symposium on Principles of Programming languages}}. \bibinfo{pages}{145--157}.
\newblock


\bibitem[Giortamis et~al\mbox{.}(2024)]%
        {GRTB24}
\bibfield{author}{\bibinfo{person}{Emmanouil Giortamis}, \bibinfo{person}{Francisco Rom{\~a}o}, \bibinfo{person}{Nathaniel Tornow}, {and} \bibinfo{person}{Pramod Bhatotia}.} \bibinfo{year}{2024}\natexlab{}.
\newblock \bibinfo{title}{{QOS}: A quantum operating system}.
\newblock
\newblock
\showeprint[arxiv]{2406.19120}~[quant-ph]


\bibitem[Giovannetti et~al\mbox{.}(2013)]%
        {GMMR13}
\bibfield{author}{\bibinfo{person}{Vittorio Giovannetti}, \bibinfo{person}{Lorenzo Maccone}, \bibinfo{person}{Tomoyuki Morimae}, {and} \bibinfo{person}{Terry~G. Rudolph}.} \bibinfo{year}{2013}\natexlab{}.
\newblock \showarticletitle{Efficient universal blind quantum computation}.
\newblock \bibinfo{journal}{\emph{Physical Review Letters}} \bibinfo{volume}{111}, \bibinfo{number}{23} (\bibinfo{year}{2013}), \bibinfo{pages}{230501}.
\newblock


\bibitem[Graham and Denning(1971)]%
        {GD71}
\bibfield{author}{\bibinfo{person}{G.~Scott Graham} {and} \bibinfo{person}{Peter~J. Denning}.} \bibinfo{year}{1971}\natexlab{}.
\newblock \showarticletitle{Protection: principles and practice}. In \bibinfo{booktitle}{\emph{Proceedings of the May 16-18, 1972, spring joint computer conference}}. \bibinfo{pages}{417--429}.
\newblock


\bibitem[Greenberger et~al\mbox{.}(1990)]%
        {GHSZ90}
\bibfield{author}{\bibinfo{person}{Daniel~M. Greenberger}, \bibinfo{person}{Michael~A. Horne}, \bibinfo{person}{Abner Shimony}, {and} \bibinfo{person}{Anton Zeilinger}.} \bibinfo{year}{1990}\natexlab{}.
\newblock \showarticletitle{{Bell’s} theorem without inequalities}.
\newblock \bibinfo{journal}{\emph{American Journal of Physics}} \bibinfo{volume}{58}, \bibinfo{number}{12} (\bibinfo{year}{1990}), \bibinfo{pages}{1131--1143}.
\newblock


\bibitem[Grover(1996)]%
        {Grover96}
\bibfield{author}{\bibinfo{person}{Lov~K. Grover}.} \bibinfo{year}{1996}\natexlab{}.
\newblock \showarticletitle{A fast quantum mechanical algorithm for database search}. In \bibinfo{booktitle}{\emph{Proceedings of the 28th Annual ACM Symposium on Theory of Computing}} \emph{(\bibinfo{series}{STOC '96})}. \bibinfo{pages}{212--219}.
\newblock


\bibitem[Gyongyosi and Imre(2019)]%
        {GI19}
\bibfield{author}{\bibinfo{person}{Laszlo Gyongyosi} {and} \bibinfo{person}{Sandor Imre}.} \bibinfo{year}{2019}\natexlab{}.
\newblock \showarticletitle{Entanglement access control for the quantum internet}.
\newblock \bibinfo{journal}{\emph{Quantum Information Processing}}  \bibinfo{volume}{18} (\bibinfo{year}{2019}), \bibinfo{pages}{1--17}.
\newblock


\bibitem[H{\"a}ner et~al\mbox{.}(2021)]%
        {HSHT21}
\bibfield{author}{\bibinfo{person}{Thomas H{\"a}ner}, \bibinfo{person}{Damian~S. Steiger}, \bibinfo{person}{Torsten Hoefler}, {and} \bibinfo{person}{Matthias Troyer}.} \bibinfo{year}{2021}\natexlab{}.
\newblock \showarticletitle{Distributed quantum computing with {QMPI}}. In \bibinfo{booktitle}{\emph{Proceedings of the International Conference for High Performance Computing, Networking, Storage and Analysis}}. \bibinfo{pages}{1--13}.
\newblock


\bibitem[Harrison et~al\mbox{.}(1976)]%
        {HRU76}
\bibfield{author}{\bibinfo{person}{Michael~A. Harrison}, \bibinfo{person}{Walter~L. Ruzzo}, {and} \bibinfo{person}{Jeffrey~D. Ullman}.} \bibinfo{year}{1976}\natexlab{}.
\newblock \showarticletitle{Protection in operating systems}.
\newblock \bibinfo{journal}{\emph{Commun. ACM}} \bibinfo{volume}{19}, \bibinfo{number}{8} (\bibinfo{year}{1976}), \bibinfo{pages}{461--471}.
\newblock


\bibitem[{H}arrow et~al\mbox{.}(2009)]%
        {HHL09}
\bibfield{author}{\bibinfo{person}{Aram~W. {H}arrow}, \bibinfo{person}{Avinatan {H}assidim}, {and} \bibinfo{person}{Seth {L}loyd}.} \bibinfo{year}{2009}\natexlab{}.
\newblock \showarticletitle{Quantum algorithm for linear systems of equations}.
\newblock \bibinfo{journal}{\emph{Physical Review Letters}} \bibinfo{volume}{103}, \bibinfo{number}{15} (\bibinfo{year}{2009}), \bibinfo{pages}{150502}.
\newblock


\bibitem[Heckey et~al\mbox{.}(2015)]%
        {HPJ+15}
\bibfield{author}{\bibinfo{person}{Jeff Heckey}, \bibinfo{person}{Shruti Patil}, \bibinfo{person}{Ali JavadiAbhari}, \bibinfo{person}{Adam Holmes}, \bibinfo{person}{Daniel Kudrow}, \bibinfo{person}{Kenneth~R. Brown}, \bibinfo{person}{Diana Franklin}, \bibinfo{person}{Frederic~T. Chong}, {and} \bibinfo{person}{Margaret Martonosi}.} \bibinfo{year}{2015}\natexlab{}.
\newblock \showarticletitle{Compiler management of communication and parallelism for quantum computation}. In \bibinfo{booktitle}{\emph{Proceedings of the Twentieth International Conference on Architectural Support for Programming Languages and Operating Systems}}. \bibinfo{pages}{445--456}.
\newblock


\bibitem[Honan et~al\mbox{.}(2020)]%
        {HLAC20}
\bibfield{author}{\bibinfo{person}{Reid Honan}, \bibinfo{person}{Trent~W. Lewis}, \bibinfo{person}{Scott Anderson}, {and} \bibinfo{person}{Jake Cooke}.} \bibinfo{year}{2020}\natexlab{}.
\newblock \showarticletitle{A quantum computer operating system}. In \bibinfo{booktitle}{\emph{Algorithms and Architectures for Parallel Processing: 20th International Conference, ICA3PP 2020}}. \bibinfo{pages}{415--431}.
\newblock


\bibitem[Hu et~al\mbox{.}(2006)]%
        {HFK06}
\bibfield{author}{\bibinfo{person}{Vincent~C. Hu}, \bibinfo{person}{David Ferraiolo}, {and} \bibinfo{person}{D.~Richard Kuhn}.} \bibinfo{year}{2006}\natexlab{}.
\newblock \bibinfo{booktitle}{\emph{Assessment of access control systems}}.
\newblock \bibinfo{publisher}{US Department of Commerce, National Institute of Standards and Technology}.
\newblock


\bibitem[Hu et~al\mbox{.}(2013)]%
        {HFK13}
\bibfield{author}{\bibinfo{person}{Vincent~C. Hu}, \bibinfo{person}{David Ferraiolo}, \bibinfo{person}{Rick Kuhn}, \bibinfo{person}{Arthur~R. Friedman}, \bibinfo{person}{Alan~J. Lang}, \bibinfo{person}{Margaret~M. Cogdell}, \bibinfo{person}{Adam Schnitzer}, \bibinfo{person}{Kenneth Sandlin}, \bibinfo{person}{Robert Miller}, {and} \bibinfo{person}{Karen Scarfone}.} \bibinfo{year}{2013}\natexlab{}.
\newblock \showarticletitle{Guide to attribute based access control ({ABAC}) definition and considerations}.
\newblock \bibinfo{journal}{\emph{NIST Special Publication}} \bibinfo{volume}{800}, \bibinfo{number}{162} (\bibinfo{year}{2013}), \bibinfo{pages}{1--54}.
\newblock


\bibitem[Hu and Kent(2012)]%
        {HK12}
\bibfield{author}{\bibinfo{person}{Vincent~C. Hu} {and} \bibinfo{person}{Karen~Ann Kent}.} \bibinfo{year}{2012}\natexlab{}.
\newblock \bibinfo{booktitle}{\emph{Guidelines for access control system evaluation metrics}}.
\newblock \bibinfo{publisher}{US Department of Commerce, National Institute of Standards and Technology}.
\newblock


\bibitem[Jorrand and Lalire(2004)]%
        {JL04}
\bibfield{author}{\bibinfo{person}{Philippe Jorrand} {and} \bibinfo{person}{Marie Lalire}.} \bibinfo{year}{2004}\natexlab{}.
\newblock \showarticletitle{Toward a quantum process algebra}. In \bibinfo{booktitle}{\emph{Proceedings of the 1st Conference on Computing Frontiers}}. \bibinfo{pages}{111--119}.
\newblock


\bibitem[Jozsa and Linden(2003)]%
        {JL03}
\bibfield{author}{\bibinfo{person}{Richard Jozsa} {and} \bibinfo{person}{Noah Linden}.} \bibinfo{year}{2003}\natexlab{}.
\newblock \showarticletitle{On the role of entanglement in quantum-computational speed-up}.
\newblock \bibinfo{journal}{\emph{Proceedings of the Royal Society of London. Series A: Mathematical, Physical and Engineering Sciences}} \bibinfo{volume}{459}, \bibinfo{number}{2036} (\bibinfo{year}{2003}), \bibinfo{pages}{2011--2032}.
\newblock


\bibitem[Kessler and {Arnon-Friedman}(2020)]%
        {KA20}
\bibfield{author}{\bibinfo{person}{Max Kessler} {and} \bibinfo{person}{Rotem {Arnon-Friedman}}.} \bibinfo{year}{2020}\natexlab{}.
\newblock \showarticletitle{{Device-independent} randomness amplification and privatization}.
\newblock \bibinfo{journal}{\emph{IEEE Journal on Selected Areas in Information Theory}} \bibinfo{volume}{1}, \bibinfo{number}{2} (\bibinfo{year}{2020}), \bibinfo{pages}{568--584}.
\newblock


\bibitem[Khadirsharbiyani et~al\mbox{.}(2023)]%
        {KSZ+23}
\bibfield{author}{\bibinfo{person}{Soheil Khadirsharbiyani}, \bibinfo{person}{Movahhed Sadeghi}, \bibinfo{person}{Mostafa~Eghbali Zarch}, \bibinfo{person}{Jagadish Kotra}, {and} \bibinfo{person}{Mahmut~Taylan Kandemir}.} \bibinfo{year}{2023}\natexlab{}.
\newblock \showarticletitle{{TRIM}: {crossTalk-awaRe qubIt Mapping} for multiprogrammed quantum systems}. In \bibinfo{booktitle}{\emph{2023 IEEE International Conference on Quantum Software (QSW)}}. \bibinfo{pages}{138--148}.
\newblock


\bibitem[Kong et~al\mbox{.}(2021)]%
        {KWH+21}
\bibfield{author}{\bibinfo{person}{Weicheng Kong}, \bibinfo{person}{Junchao Wang}, \bibinfo{person}{Yongjian Han}, \bibinfo{person}{Yuchun Wu}, \bibinfo{person}{Yu Zhang}, \bibinfo{person}{Menghan Dou}, \bibinfo{person}{Yuan Fang}, {and} \bibinfo{person}{Guoping Guo}.} \bibinfo{year}{2021}\natexlab{}.
\newblock \bibinfo{title}{{Origin Pilot}: A quantum operating system for effecient usage of quantum resources}.
\newblock
\newblock
\showeprint[arxiv]{2105.10730}~[quant-ph]


\bibitem[Lampson(1969)]%
        {Lampson69}
\bibfield{author}{\bibinfo{person}{Butler~W. Lampson}.} \bibinfo{year}{1969}\natexlab{}.
\newblock \showarticletitle{Dynamic protection structures}. In \bibinfo{booktitle}{\emph{Proceedings of the November 18-20, 1969, fall joint computer conference}}. \bibinfo{pages}{27--38}.
\newblock


\bibitem[Lampson(1974)]%
        {Lampson74}
\bibfield{author}{\bibinfo{person}{Butler~W. Lampson}.} \bibinfo{year}{1974}\natexlab{}.
\newblock \showarticletitle{Protection}.
\newblock \bibinfo{journal}{\emph{ACM SIGOPS Operating Systems Review}} \bibinfo{volume}{8}, \bibinfo{number}{1} (\bibinfo{year}{1974}), \bibinfo{pages}{18--24}.
\newblock


\bibitem[LaPadula and Bell(1973)]%
        {BLP73}
\bibfield{author}{\bibinfo{person}{Leonard~J. LaPadula} {and} \bibinfo{person}{D.~Elliot Bell}.} \bibinfo{year}{1973}\natexlab{}.
\newblock \bibinfo{booktitle}{\emph{Secure computer systems: A mathematical model}}.
\newblock \bibinfo{type}{{T}echnical {R}eport} ESD–TR–73–278–I. \bibinfo{institution}{The MITRE Corporation, Bedford, MA}.
\newblock


\bibitem[Li et~al\mbox{.}(2025)]%
        {LSL+25}
\bibfield{author}{\bibinfo{person}{Jinyang Li}, \bibinfo{person}{Yuhong Song}, \bibinfo{person}{Yipei Liu}, \bibinfo{person}{Jianli Pan}, \bibinfo{person}{Lei Yang}, \bibinfo{person}{Travis Humble}, {and} \bibinfo{person}{Weiwen Jiang}.} \bibinfo{year}{2025}\natexlab{}.
\newblock \bibinfo{title}{{QuSplit}: Achieving both high fidelity and throughput via job splitting on noisy quantum computers}.
\newblock
\newblock
\showeprint[arxiv]{2501.12492}~[quant-ph]


\bibitem[Li et~al\mbox{.}(2024)]%
        {LGB+24}
\bibfield{author}{\bibinfo{person}{Zirui Li}, \bibinfo{person}{Minghao Guo}, \bibinfo{person}{Mayank Barad}, \bibinfo{person}{Wei Tang}, \bibinfo{person}{Eddy~Z. Zhang}, {and} \bibinfo{person}{Yipeng Huang}.} \bibinfo{year}{2024}\natexlab{}.
\newblock \bibinfo{title}{A case for quantum circuit cutting for {NISQ} applications: Impact of topology, determinism, and sparsity}.
\newblock
\newblock
\showeprint[arxiv]{2412.17929}~[quant-ph]


\bibitem[Liu et~al\mbox{.}(2023)]%
        {LWS+23}
\bibfield{author}{\bibinfo{person}{Chenxu Liu}, \bibinfo{person}{Meng Wang}, \bibinfo{person}{Samuel~A Stein}, \bibinfo{person}{Yufei Ding}, {and} \bibinfo{person}{Ang Li}.} \bibinfo{year}{2023}\natexlab{}.
\newblock \bibinfo{title}{Quantum memory: A missing piece in quantum computing units}.
\newblock
\newblock
\showeprint[arxiv]{2309.14432}~[quant-ph]


\bibitem[Liu and Dou(2021)]%
        {LD21}
\bibfield{author}{\bibinfo{person}{Lei Liu} {and} \bibinfo{person}{Xinglei Dou}.} \bibinfo{year}{2021}\natexlab{}.
\newblock \showarticletitle{{QuCloud}: A new qubit mapping mechanism for {multi-programming} quantum computing in cloud environment}. In \bibinfo{booktitle}{\emph{2021 IEEE International Symposium on High-Performance Computer Architecture (HPCA)}}. \bibinfo{pages}{167--178}.
\newblock


\bibitem[Liu and Dou(2024)]%
        {LD24}
\bibfield{author}{\bibinfo{person}{Lei Liu} {and} \bibinfo{person}{Xinglei Dou}.} \bibinfo{year}{2024}\natexlab{}.
\newblock \showarticletitle{{QuCloud+}: A holistic qubit mapping scheme for single/{multi-programming} on {2D/3D} {NISQ} quantum computers}.
\newblock \bibinfo{journal}{\emph{ACM Transactions on Architecture and Code Optimization}} \bibinfo{volume}{21}, \bibinfo{number}{1} (\bibinfo{year}{2024}), \bibinfo{pages}{1--27}.
\newblock


\bibitem[{L}loyd(1996)]%
        {Lloyd96}
\bibfield{author}{\bibinfo{person}{Seth {L}loyd}.} \bibinfo{year}{1996}\natexlab{}.
\newblock \showarticletitle{Universal quantum simulators}.
\newblock \bibinfo{journal}{\emph{Science}} \bibinfo{volume}{273}, \bibinfo{number}{5278} (\bibinfo{year}{1996}), \bibinfo{pages}{1073--1078}.
\newblock


\bibitem[Mandelbaum et~al\mbox{.}(2024)]%
        {MCG24}
\bibfield{author}{\bibinfo{person}{Ryan Mandelbaum}, \bibinfo{person}{Antonio~D. C{\'o}rcoles}, {and} \bibinfo{person}{Jay Gambetta}.} \bibinfo{year}{2024}\natexlab{}.
\newblock \showarticletitle{IBM's big bet on the quantum-centric supercomputer: recent advances point the way to useful classical-quantum hybrids}.
\newblock \bibinfo{journal}{\emph{IEEE Spectrum}} \bibinfo{volume}{61}, \bibinfo{number}{9} (\bibinfo{year}{2024}), \bibinfo{pages}{24--33}.
\newblock


\bibitem[Mantri et~al\mbox{.}(2013)]%
        {MPF13}
\bibfield{author}{\bibinfo{person}{Atul Mantri}, \bibinfo{person}{Carlos~A. {P{\'e}rez-Delgado}}, {and} \bibinfo{person}{Joseph~F. Fitzsimons}.} \bibinfo{year}{2013}\natexlab{}.
\newblock \showarticletitle{Optimal blind quantum computation}.
\newblock \bibinfo{journal}{\emph{Physical Review Letters}} \bibinfo{volume}{111}, \bibinfo{number}{23} (\bibinfo{year}{2013}), \bibinfo{pages}{230502}.
\newblock


\bibitem[Mayers and Yao(1998)]%
        {MY98}
\bibfield{author}{\bibinfo{person}{Dominic Mayers} {and} \bibinfo{person}{Andrew Yao}.} \bibinfo{year}{1998}\natexlab{}.
\newblock \showarticletitle{Quantum cryptography with imperfect apparatus}. In \bibinfo{booktitle}{\emph{Proceedings 39th Annual Symposium on Foundations of Computer Science (Cat. No. 98CB36280)}}. \bibinfo{pages}{503--509}.
\newblock


\bibitem[Mermin(1990)]%
        {Mermin90}
\bibfield{author}{\bibinfo{person}{N.~David Mermin}.} \bibinfo{year}{1990}\natexlab{}.
\newblock \showarticletitle{Extreme quantum entanglement in a superposition of macroscopically distinct states}.
\newblock \bibinfo{journal}{\emph{Physical Review Letters}} \bibinfo{volume}{65}, \bibinfo{number}{15} (\bibinfo{year}{1990}), \bibinfo{pages}{1838–1840}.
\newblock


\bibitem[Meuli et~al\mbox{.}(2019)]%
        {MSR+19}
\bibfield{author}{\bibinfo{person}{Giulia Meuli}, \bibinfo{person}{Mathias Soeken}, \bibinfo{person}{Martin Roetteler}, \bibinfo{person}{Nikolaj Bjorner}, {and} \bibinfo{person}{Giovanni {De Micheli}}.} \bibinfo{year}{2019}\natexlab{}.
\newblock \showarticletitle{Reversible pebbling game for quantum memory management}. In \bibinfo{booktitle}{\emph{2019 Design, Automation \& Test in Europe Conference \& Exhibition (DATE)}}. \bibinfo{pages}{288--291}.
\newblock


\bibitem[Mi et~al\mbox{.}(2022)]%
        {MDS22}
\bibfield{author}{\bibinfo{person}{Allen Mi}, \bibinfo{person}{Shuwen Deng}, {and} \bibinfo{person}{Jakub Szefer}.} \bibinfo{year}{2022}\natexlab{}.
\newblock \showarticletitle{Securing reset operations in nisq quantum computers}. In \bibinfo{booktitle}{\emph{Proceedings of the 2022 ACM SIGSAC Conference on Computer and Communications Security}}. \bibinfo{pages}{2279--2293}.
\newblock


\bibitem[Miller and Shi(2016)]%
        {MS16b}
\bibfield{author}{\bibinfo{person}{Carl~A. Miller} {and} \bibinfo{person}{Yaoyun Shi}.} \bibinfo{year}{2016}\natexlab{}.
\newblock \showarticletitle{Robust protocols for securely expanding randomness and distributing keys using untrusted quantum devices}.
\newblock \bibinfo{journal}{\emph{Journal of the ACM (JACM)}} \bibinfo{volume}{63}, \bibinfo{number}{4} (\bibinfo{year}{2016}), \bibinfo{pages}{1--63}.
\newblock


\bibitem[Miller and Shi(2017)]%
        {MS17}
\bibfield{author}{\bibinfo{person}{Carl~A. Miller} {and} \bibinfo{person}{Yaoyun Shi}.} \bibinfo{year}{2017}\natexlab{}.
\newblock \showarticletitle{Universal security for randomness expansion from the {spot-checking} protocol}.
\newblock \bibinfo{journal}{\emph{SIAM J. Comput.}} \bibinfo{volume}{46}, \bibinfo{number}{4} (\bibinfo{year}{2017}), \bibinfo{pages}{1304--1335}.
\newblock


\bibitem[Mitarai and Fujii(2021)]%
        {MF21}
\bibfield{author}{\bibinfo{person}{Kosuke Mitarai} {and} \bibinfo{person}{Keisuke Fujii}.} \bibinfo{year}{2021}\natexlab{}.
\newblock \showarticletitle{Constructing a virtual {two-qubit} gate by sampling {single-qubit} operations}.
\newblock \bibinfo{journal}{\emph{New Journal of Physics}} \bibinfo{volume}{23}, \bibinfo{number}{2} (\bibinfo{year}{2021}), \bibinfo{pages}{023021}.
\newblock


\bibitem[Morimae(2012)]%
        {Morimae12}
\bibfield{author}{\bibinfo{person}{Tomoyuki Morimae}.} \bibinfo{year}{2012}\natexlab{}.
\newblock \showarticletitle{Continuous-variable blind quantum computation}.
\newblock \bibinfo{journal}{\emph{Physical Review Letters}} \bibinfo{volume}{109}, \bibinfo{number}{23} (\bibinfo{year}{2012}), \bibinfo{pages}{230502}.
\newblock


\bibitem[Morimae(2014)]%
        {Morimae14}
\bibfield{author}{\bibinfo{person}{Tomoyuki Morimae}.} \bibinfo{year}{2014}\natexlab{}.
\newblock \showarticletitle{Verification for measurement-only blind quantum computing}.
\newblock \bibinfo{journal}{\emph{Physical Review A}} \bibinfo{volume}{89}, \bibinfo{number}{6} (\bibinfo{year}{2014}), \bibinfo{pages}{060302}.
\newblock


\bibitem[Morimae et~al\mbox{.}(2015)]%
        {MDK15}
\bibfield{author}{\bibinfo{person}{Tomoyuki Morimae}, \bibinfo{person}{Vedran Dunjko}, {and} \bibinfo{person}{Elham Kashefi}.} \bibinfo{year}{2015}\natexlab{}.
\newblock \showarticletitle{Ground state blind quantum computation on {AKLT} state}.
\newblock \bibinfo{journal}{\emph{Quantum Information \& Computation}} \bibinfo{volume}{15}, \bibinfo{number}{3–4} (\bibinfo{year}{2015}), \bibinfo{pages}{200–234}.
\newblock


\bibitem[Morimae and Fujii(2012)]%
        {MF12}
\bibfield{author}{\bibinfo{person}{Tomoyuki Morimae} {and} \bibinfo{person}{Keisuke Fujii}.} \bibinfo{year}{2012}\natexlab{}.
\newblock \showarticletitle{Blind topological measurement-based quantum computation}.
\newblock \bibinfo{journal}{\emph{Nature communications}} \bibinfo{volume}{3}, \bibinfo{number}{1} (\bibinfo{year}{2012}), \bibinfo{pages}{1036}.
\newblock


\bibitem[Morimae and Fujii(2013)]%
        {MF13}
\bibfield{author}{\bibinfo{person}{Tomoyuki Morimae} {and} \bibinfo{person}{Keisuke Fujii}.} \bibinfo{year}{2013}\natexlab{}.
\newblock \showarticletitle{Blind quantum computation protocol in which {Alice} only makes measurements}.
\newblock \bibinfo{journal}{\emph{Physical Review A}} \bibinfo{volume}{87}, \bibinfo{number}{5} (\bibinfo{year}{2013}), \bibinfo{pages}{050301}.
\newblock


\bibitem[Morimae and Koshiba(2013)]%
        {MK13}
\bibfield{author}{\bibinfo{person}{Tomoyuki Morimae} {and} \bibinfo{person}{Takeshi Koshiba}.} \bibinfo{year}{2013}\natexlab{}.
\newblock \bibinfo{title}{Composable security of {measuring-Alice} blind quantum computation}.
\newblock
\newblock
\showeprint[arxiv]{1306.2113}~[quant-ph]


\bibitem[{N}ielsen and {C}huang(2010)]%
        {NC10}
\bibfield{author}{\bibinfo{person}{Michael~A. {N}ielsen} {and} \bibinfo{person}{Isaac~L. {C}huang}.} \bibinfo{year}{2010}\natexlab{}.
\newblock \bibinfo{booktitle}{\emph{Quantum Computation and Quantum Information: 10th Anniversary Edition}}.
\newblock \bibinfo{publisher}{Cambridge University Press}.
\newblock


\bibitem[Niu and {Todri-Sanial}(2023)]%
        {NT23}
\bibfield{author}{\bibinfo{person}{Siyuan Niu} {and} \bibinfo{person}{Aida {Todri-Sanial}}.} \bibinfo{year}{2023}\natexlab{}.
\newblock \showarticletitle{Enabling {multi-programming} mechanism for quantum computing in the {NISQ} era}.
\newblock \bibinfo{journal}{\emph{Quantum}}  \bibinfo{volume}{7} (\bibinfo{year}{2023}), \bibinfo{pages}{925}.
\newblock


\bibitem[Ohkura et~al\mbox{.}(2022)]%
        {OSV22}
\bibfield{author}{\bibinfo{person}{Yasuhiro Ohkura}, \bibinfo{person}{Takahiko Satoh}, {and} \bibinfo{person}{Rodney {Van Meter}}.} \bibinfo{year}{2022}\natexlab{}.
\newblock \showarticletitle{Simultaneous execution of quantum circuits on current and {near-future} {NISQ} systems}.
\newblock \bibinfo{journal}{\emph{IEEE Transactions on Quantum Engineering}}  \bibinfo{volume}{3} (\bibinfo{year}{2022}), \bibinfo{pages}{1–10}.
\newblock


\bibitem[Orenstein and Chaudhary(2024)]%
        {OC24}
\bibfield{author}{\bibinfo{person}{Aaron Orenstein} {and} \bibinfo{person}{Vipin Chaudhary}.} \bibinfo{year}{2024}\natexlab{}.
\newblock \showarticletitle{{QGroup}: Parallel quantum job scheduling using dynamic programming}. In \bibinfo{booktitle}{\emph{2024 IEEE International Conference on Quantum Computing and Engineering (QCE)}}, Vol.~\bibinfo{volume}{1}. \bibinfo{pages}{990--999}.
\newblock


\bibitem[Park and Sandhu(2004)]%
        {PS04}
\bibfield{author}{\bibinfo{person}{Jaehong Park} {and} \bibinfo{person}{Ravi Sandhu}.} \bibinfo{year}{2004}\natexlab{}.
\newblock \showarticletitle{The UCON$_{\text{ABC}}$ usage control model}.
\newblock \bibinfo{journal}{\emph{ACM Transactions on Information and System Security (TISSEC)}} \bibinfo{volume}{7}, \bibinfo{number}{1} (\bibinfo{year}{2004}), \bibinfo{pages}{128--174}.
\newblock


\bibitem[Pascuzzi and C{\'o}rcoles(2024)]%
        {PC24}
\bibfield{author}{\bibinfo{person}{Vincent~R. Pascuzzi} {and} \bibinfo{person}{Antonio~D. C{\'o}rcoles}.} \bibinfo{year}{2024}\natexlab{}.
\newblock \bibinfo{title}{Quantum-centric supercomputing for physics research}.
\newblock
\newblock
\showeprint[arxiv]{2408.11741}~[quant-ph]


\bibitem[{P}earl(2000)]%
        {Pearl00}
\bibfield{author}{\bibinfo{person}{Judea {P}earl}.} \bibinfo{year}{2000}\natexlab{}.
\newblock \bibinfo{booktitle}{\emph{Causality: Models, Reasoning, and Inference}}.
\newblock \bibinfo{publisher}{Cambridge University Press}, \bibinfo{address}{USA}.
\newblock


\bibitem[Peng et~al\mbox{.}(2020)]%
        {PHOW20}
\bibfield{author}{\bibinfo{person}{Tianyi Peng}, \bibinfo{person}{Aram~W. Harrow}, \bibinfo{person}{Maris Ozols}, {and} \bibinfo{person}{Xiaodi Wu}.} \bibinfo{year}{2020}\natexlab{}.
\newblock \showarticletitle{Simulating large quantum circuits on a small quantum computer}.
\newblock \bibinfo{journal}{\emph{Physical review letters}} \bibinfo{volume}{125}, \bibinfo{number}{15} (\bibinfo{year}{2020}), \bibinfo{pages}{150504}.
\newblock


\bibitem[Pironio et~al\mbox{.}(2009)]%
        {PAB+09}
\bibfield{author}{\bibinfo{person}{Stefano Pironio}, \bibinfo{person}{Antonio Ac{\'\i}n}, \bibinfo{person}{Nicolas Brunner}, \bibinfo{person}{Nicolas Gisin}, \bibinfo{person}{Serge Massar}, {and} \bibinfo{person}{Valerio Scarani}.} \bibinfo{year}{2009}\natexlab{}.
\newblock \showarticletitle{{Device-independent} quantum key distribution secure against collective attacks}.
\newblock \bibinfo{journal}{\emph{New Journal of Physics}} \bibinfo{volume}{11}, \bibinfo{number}{4} (\bibinfo{year}{2009}), \bibinfo{pages}{045021}.
\newblock


\bibitem[Pironio et~al\mbox{.}(2010)]%
        {PAM+10}
\bibfield{author}{\bibinfo{person}{Stefano Pironio}, \bibinfo{person}{Antonio Ac{\'\i}n}, \bibinfo{person}{Serge Massar}, \bibinfo{person}{A.~Boyer {de La Giroday}}, \bibinfo{person}{Dzmitry~N. Matsukevich}, \bibinfo{person}{Peter Maunz}, \bibinfo{person}{Steven Olmschenk}, \bibinfo{person}{David Hayes}, \bibinfo{person}{Lefroy Luo}, \bibinfo{person}{T.~Andrew Manning}, {et~al\mbox{.}}} \bibinfo{year}{2010}\natexlab{}.
\newblock \showarticletitle{Random numbers certified by {Bell}’s theorem}.
\newblock \bibinfo{journal}{\emph{Nature}} \bibinfo{volume}{464}, \bibinfo{number}{7291} (\bibinfo{year}{2010}), \bibinfo{pages}{1021--1024}.
\newblock


\bibitem[Piveteau and Sutter(2024)]%
        {PS24}
\bibfield{author}{\bibinfo{person}{Christophe Piveteau} {and} \bibinfo{person}{David Sutter}.} \bibinfo{year}{2024}\natexlab{}.
\newblock \showarticletitle{Circuit knitting with classical communication}.
\newblock \bibinfo{journal}{\emph{IEEE Transactions on Information Theory}} \bibinfo{volume}{70}, \bibinfo{number}{4} (\bibinfo{year}{2024}), \bibinfo{pages}{2734–2745}.
\newblock


\bibitem[Puterman(2014)]%
        {Puterman14}
\bibfield{author}{\bibinfo{person}{Martin~L. Puterman}.} \bibinfo{year}{2014}\natexlab{}.
\newblock \bibinfo{booktitle}{\emph{Markov decision processes: discrete stochastic dynamic programming}}.
\newblock \bibinfo{publisher}{John Wiley \& Sons}.
\newblock


\bibitem[Rabin(1980)]%
        {Rabin80}
\bibfield{author}{\bibinfo{person}{Michael~O. Rabin}.} \bibinfo{year}{1980}\natexlab{}.
\newblock \showarticletitle{$N$-process synchronization by $4\cdot\log_2 N$-valued shared variable}. In \bibinfo{booktitle}{\emph{21st Annual Symposium on Foundations of Computer Science (sfcs 1980)}}. \bibinfo{pages}{407--410}.
\newblock


\bibitem[Ravi et~al\mbox{.}(2021)]%
        {RSMC21}
\bibfield{author}{\bibinfo{person}{Gokul~Subramanian Ravi}, \bibinfo{person}{Kaitlin~N. Smith}, \bibinfo{person}{Prakash Murali}, {and} \bibinfo{person}{Frederic~T. Chong}.} \bibinfo{year}{2021}\natexlab{}.
\newblock \showarticletitle{Adaptive job and resource management for the growing quantum cloud}. In \bibinfo{booktitle}{\emph{2021 IEEE International Conference on Quantum Computing and Engineering (QCE)}}. \bibinfo{pages}{301--312}.
\newblock


\bibitem[Reichardt et~al\mbox{.}(2013)]%
        {RUV13}
\bibfield{author}{\bibinfo{person}{Ben~W. Reichardt}, \bibinfo{person}{Falk Unger}, {and} \bibinfo{person}{Umesh Vazirani}.} \bibinfo{year}{2013}\natexlab{}.
\newblock \showarticletitle{Classical command of quantum systems}.
\newblock \bibinfo{journal}{\emph{Nature}} \bibinfo{volume}{496}, \bibinfo{number}{7446} (\bibinfo{year}{2013}), \bibinfo{pages}{456--460}.
\newblock


\bibitem[Resch et~al\mbox{.}(2021)]%
        {RGH+21}
\bibfield{author}{\bibinfo{person}{Salonik Resch}, \bibinfo{person}{Anthony Gutierrez}, \bibinfo{person}{Joon~Suk Huh}, \bibinfo{person}{Srikant Bharadwaj}, \bibinfo{person}{Yasuko Eckert}, \bibinfo{person}{Gabriel Loh}, \bibinfo{person}{Mark Oskin}, {and} \bibinfo{person}{Swamit Tannu}.} \bibinfo{year}{2021}\natexlab{}.
\newblock \bibinfo{title}{Accelerating variational quantum algorithms using circuit concurrency}.
\newblock
\newblock
\showeprint[arxiv]{2109.01714}~[cs.ET]


\bibitem[Saltzer(1974)]%
        {Saltzer74}
\bibfield{author}{\bibinfo{person}{Jerome~H. Saltzer}.} \bibinfo{year}{1974}\natexlab{}.
\newblock \showarticletitle{Protection and the control of information sharing in {Multics}}.
\newblock \bibinfo{journal}{\emph{Commun. ACM}} \bibinfo{volume}{17}, \bibinfo{number}{7} (\bibinfo{year}{1974}), \bibinfo{pages}{388--402}.
\newblock


\bibitem[Saltzer and Schroeder(1975)]%
        {SS75}
\bibfield{author}{\bibinfo{person}{Jerome~H. Saltzer} {and} \bibinfo{person}{Michael~D. Schroeder}.} \bibinfo{year}{1975}\natexlab{}.
\newblock \showarticletitle{The protection of information in computer systems}.
\newblock \bibinfo{journal}{\emph{Proc. IEEE}} \bibinfo{volume}{63}, \bibinfo{number}{9} (\bibinfo{year}{1975}), \bibinfo{pages}{1278--1308}.
\newblock


\bibitem[Sandhu and Park(2003)]%
        {SP03}
\bibfield{author}{\bibinfo{person}{Ravi Sandhu} {and} \bibinfo{person}{Jaehong Park}.} \bibinfo{year}{2003}\natexlab{}.
\newblock \showarticletitle{Usage control: A vision for next generation access control}. In \bibinfo{booktitle}{\emph{Computer Network Security: Second International Workshop on Mathematical Methods, Models, and Architectures for Computer Network Security, MMM-ACNS 2003}}. \bibinfo{pages}{17--31}.
\newblock


\bibitem[Sandhu et~al\mbox{.}(1996)]%
        {SCFY96}
\bibfield{author}{\bibinfo{person}{Ravi~S. Sandhu}, \bibinfo{person}{Edward~J. Coyne}, \bibinfo{person}{Hal~L. Feinstein}, {and} \bibinfo{person}{Charles~E. Youman}.} \bibinfo{year}{1996}\natexlab{}.
\newblock \showarticletitle{{Role-based} access control models}.
\newblock In \bibinfo{booktitle}{\emph{IEEE Computer}}. Vol.~\bibinfo{volume}{29}. \bibinfo{pages}{38--47}.
\newblock


\bibitem[Sandhu and Samarati(1994)]%
        {SS94}
\bibfield{author}{\bibinfo{person}{Ravi~S. Sandhu} {and} \bibinfo{person}{Pierangela Samarati}.} \bibinfo{year}{1994}\natexlab{}.
\newblock \showarticletitle{Access control: principle and practice}.
\newblock \bibinfo{journal}{\emph{IEEE communications magazine}} \bibinfo{volume}{32}, \bibinfo{number}{9} (\bibinfo{year}{1994}), \bibinfo{pages}{40--48}.
\newblock


\bibitem[Shor(1994)]%
        {Shor94}
\bibfield{author}{\bibinfo{person}{Peter~W. Shor}.} \bibinfo{year}{1994}\natexlab{}.
\newblock \showarticletitle{Algorithms for quantum computation: discrete logarithms and factoring}. In \bibinfo{booktitle}{\emph{Proceedings 35th Annual IEEE Symposium on Foundations of Computer Science}} \emph{(\bibinfo{series}{FOCS '94})}. \bibinfo{pages}{124--134}.
\newblock


\bibitem[Tang et~al\mbox{.}(2021)]%
        {TTS+21}
\bibfield{author}{\bibinfo{person}{Wei Tang}, \bibinfo{person}{Teague Tomesh}, \bibinfo{person}{Martin Suchara}, \bibinfo{person}{Jeffrey Larson}, {and} \bibinfo{person}{Margaret Martonosi}.} \bibinfo{year}{2021}\natexlab{}.
\newblock \showarticletitle{{CutQC}: Using small quantum computers for large quantum circuit evaluations}. In \bibinfo{booktitle}{\emph{Proceedings of the 26th ACM International Conference on Architectural Support for Programming Languages and Operating Systems}}. \bibinfo{pages}{473--486}.
\newblock


\bibitem[Tani et~al\mbox{.}(2012)]%
        {TKM12}
\bibfield{author}{\bibinfo{person}{Seiichiro Tani}, \bibinfo{person}{Hirotada Kobayashi}, {and} \bibinfo{person}{Keiji Matsumoto}.} \bibinfo{year}{2012}\natexlab{}.
\newblock \showarticletitle{Exact quantum algorithms for the leader election problem}.
\newblock \bibinfo{journal}{\emph{ACM Transactions on Computation Theory (TOCT)}} \bibinfo{volume}{4}, \bibinfo{number}{1} (\bibinfo{year}{2012}), \bibinfo{pages}{1--24}.
\newblock


\bibitem[Trochatos et~al\mbox{.}(2024)]%
        {TDX+24}
\bibfield{author}{\bibinfo{person}{Theodoros Trochatos}, \bibinfo{person}{Sanjay Deshpande}, \bibinfo{person}{Chuanqi Xu}, \bibinfo{person}{Yao Lu}, \bibinfo{person}{Yongshan Ding}, {and} \bibinfo{person}{Jakub Szefer}.} \bibinfo{year}{2024}\natexlab{}.
\newblock \showarticletitle{Dynamic pulse switching for protection of quantum computation on untrusted clouds}. In \bibinfo{booktitle}{\emph{2024 IEEE International Symposium on Hardware Oriented Security and Trust (HOST)}}. \bibinfo{pages}{404--414}.
\newblock


\bibitem[Trochatos and Szefer(2024)]%
        {TJ24}
\bibfield{author}{\bibinfo{person}{Theodoros Trochatos} {and} \bibinfo{person}{Jakub Szefer}.} \bibinfo{year}{2024}\natexlab{}.
\newblock \bibinfo{title}{Quantum operating system support for quantum trusted execution environments}.
\newblock
\newblock
\showeprint[arxiv]{2410.08486}~[quant-ph]


\bibitem[Trochatos et~al\mbox{.}(2023a)]%
        {TXD+b23}
\bibfield{author}{\bibinfo{person}{Theodoros Trochatos}, \bibinfo{person}{Chuanqi Xu}, \bibinfo{person}{Sanjay Deshpande}, \bibinfo{person}{Yao Lu}, \bibinfo{person}{Yongshan Ding}, {and} \bibinfo{person}{Jakub Szefer}.} \bibinfo{year}{2023}\natexlab{a}.
\newblock \bibinfo{title}{Hardware architecture for a quantum computer trusted execution environment}.
\newblock
\newblock
\showeprint[arxiv]{2308.03897}~[cs.ET]


\bibitem[Trochatos et~al\mbox{.}(2023b)]%
        {TXD+23}
\bibfield{author}{\bibinfo{person}{Theodoros Trochatos}, \bibinfo{person}{Chuanqi Xu}, \bibinfo{person}{Sanjay Deshpande}, \bibinfo{person}{Yao Lu}, \bibinfo{person}{Yongshan Ding}, {and} \bibinfo{person}{Jakub Szefer}.} \bibinfo{year}{2023}\natexlab{b}.
\newblock \showarticletitle{A quantum computer trusted execution environment}.
\newblock \bibinfo{journal}{\emph{IEEE Computer Architecture Letters}} \bibinfo{volume}{22}, \bibinfo{number}{2} (\bibinfo{year}{2023}), \bibinfo{pages}{177--180}.
\newblock


\bibitem[Vazirani and Vidick(2012)]%
        {VV12}
\bibfield{author}{\bibinfo{person}{Umesh Vazirani} {and} \bibinfo{person}{Thomas Vidick}.} \bibinfo{year}{2012}\natexlab{}.
\newblock \showarticletitle{Certifiable quantum dice: or, true random number generation secure against quantum adversaries}. In \bibinfo{booktitle}{\emph{Proceedings of the forty-fourth annual ACM symposium on Theory of computing}}. \bibinfo{pages}{61--76}.
\newblock


\bibitem[Vazirani and Vidick(2014)]%
        {VV14}
\bibfield{author}{\bibinfo{person}{Umesh Vazirani} {and} \bibinfo{person}{Thomas Vidick}.} \bibinfo{year}{2014}\natexlab{}.
\newblock \showarticletitle{Fully {device-independent} quantum key distribution}.
\newblock \bibinfo{journal}{\emph{Physical Review Letters}}  \bibinfo{volume}{113} (\bibinfo{year}{2014}), \bibinfo{pages}{140501}.
\newblock
Issue 14.


\bibitem[Wu et~al\mbox{.}(2022)]%
        {WZLSXD22}
\bibfield{author}{\bibinfo{person}{Anbang Wu}, \bibinfo{person}{Hezi Zhang}, \bibinfo{person}{Gushu Li}, \bibinfo{person}{Alireza Shabani}, \bibinfo{person}{Yuan Xie}, {and} \bibinfo{person}{Yufei Ding}.} \bibinfo{year}{2022}\natexlab{}.
\newblock \showarticletitle{{AutoComm}: a framework for enabling efficient communication in distributed quantum programs}. In \bibinfo{booktitle}{\emph{2022 55th IEEE/ACM International Symposium on Microarchitecture (MICRO)}}. \bibinfo{pages}{1027--1041}.
\newblock


\bibitem[Xu et~al\mbox{.}(2023a)]%
        {XCMS23}
\bibfield{author}{\bibinfo{person}{Chuanqi Xu}, \bibinfo{person}{Jessie Chen}, \bibinfo{person}{Allen Mi}, {and} \bibinfo{person}{Jakub Szefer}.} \bibinfo{year}{2023}\natexlab{a}.
\newblock \showarticletitle{Securing nisq quantum computer reset operations against higher energy state attacks}. In \bibinfo{booktitle}{\emph{Proceedings of the 2023 ACM SIGSAC Conference on Computer and Communications Security}}. \bibinfo{pages}{594--607}.
\newblock


\bibitem[Xu et~al\mbox{.}(2023b)]%
        {XES23}
\bibfield{author}{\bibinfo{person}{Chuanqi Xu}, \bibinfo{person}{Ferhat Erata}, {and} \bibinfo{person}{Jakub Szefer}.} \bibinfo{year}{2023}\natexlab{b}.
\newblock \showarticletitle{Exploration of power {side-channel} vulnerabilities in quantum computer controllers}. In \bibinfo{booktitle}{\emph{Proceedings of the 2023 ACM SIGSAC Conference on Computer and Communications Security}}. \bibinfo{pages}{579--593}.
\newblock


\bibitem[Ying et~al\mbox{.}(2013)]%
        {YFY13}
\bibfield{author}{\bibinfo{person}{Mingsheng Ying}, \bibinfo{person}{Yuan Feng}, {and} \bibinfo{person}{Nengkun Yu}.} \bibinfo{year}{2013}\natexlab{}.
\newblock \showarticletitle{Quantum information-flow security: Noninterference and access control}. In \bibinfo{booktitle}{\emph{2013 IEEE 26th Computer Security Foundations Symposium}}. \bibinfo{pages}{130--144}.
\newblock


\bibitem[Ying et~al\mbox{.}(2022)]%
        {YZLF22}
\bibfield{author}{\bibinfo{person}{Mingsheng Ying}, \bibinfo{person}{Li Zhou}, \bibinfo{person}{Yangjia Li}, {and} \bibinfo{person}{Yuan Feng}.} \bibinfo{year}{2022}\natexlab{}.
\newblock \showarticletitle{A proof system for disjoint parallel quantum programs}.
\newblock \bibinfo{journal}{\emph{Theoretical Computer Science}}  \bibinfo{volume}{897} (\bibinfo{year}{2022}), \bibinfo{pages}{164--184}.
\newblock


\bibitem[Zhang et~al\mbox{.}(2005)]%
        {ZPSP05}
\bibfield{author}{\bibinfo{person}{Xinwen Zhang}, \bibinfo{person}{Francesco {Parisi-Presicce}}, \bibinfo{person}{Ravi Sandhu}, {and} \bibinfo{person}{Jaehong Park}.} \bibinfo{year}{2005}\natexlab{}.
\newblock \showarticletitle{Formal model and policy specification of usage control}.
\newblock \bibinfo{journal}{\emph{ACM Transactions on Information and System Security (TISSEC)}} \bibinfo{volume}{8}, \bibinfo{number}{4} (\bibinfo{year}{2005}), \bibinfo{pages}{351--387}.
\newblock


\bibitem[Zhang and Ying(2024)]%
        {ZY24c}
\bibfield{author}{\bibinfo{person}{Zhicheng Zhang} {and} \bibinfo{person}{Mingsheng Ying}.} \bibinfo{year}{2024}\natexlab{}.
\newblock \bibinfo{title}{Atomicity in distributed quantum computing}.
\newblock
\newblock
\showeprint[arxiv]{2404.18592}~[quant-ph]


\end{thebibliography}

\newpage

\appendix

\section{Proof Details}
\label{sec:proof_details}

\subsection{Proof of \Cref{thm:classical-security}}
\label{sub:proof_of_classical-security}

In this appendix, we present the full proof of \Cref{thm:classical-security}.
The proof uses notations and tools in probabilistic graphical models,
of which background is provided in \Cref{sec:background_on_probabilistic_graphical_models}.

\begin{proof}[Proof of \Cref{thm:classical-security}]
	Let us fix any scheduler $S$ of the system.
	Since we allow the program $P$ to be probabilistic (e.g., $P_v$ in \Cref{fig:userv} is already probabilistic),
	the value of any register can be regarded as a random variable.
	%where the randomness comes from all subjects using random numbers.
	For example, $A(t)$, the value of the register $A$ at time $t$,
	is a random variable.
	Similarly, $\alpha(t)$, the request at time $t$ in the history,
	can also be seen as a random variable.

	Let us first identify several time points $t_1,t_{v,j},t_v,t_2$ with respect to program $P_v$ in \Cref{fig:userv}, by supposing:
	\begin{itemize}
		\item 
			At $t=t_1-1$: $v$ issues the write request in Line $1$
		\item
			At $t=t_{v,j}-1$: $v$ issues the write request for $j$ in Line $3$.
		\item
			At $t=t_v-1$: $v$ issues the write request in Line $5$.
		\item
			At $t=t_2-1$: $v$ issues the write request in Line $6$.
	\end{itemize}

	For convenience, let us denote $D_j:=C_j,L\bracks*{w_j}$.
	Note that from $M_0,M_1,M_2$, we have
	\begin{equation*}
		\Obs\parens*{w_j,t}=
		\begin{cases}
			D_j\parens*{t}, & t\leq t_2,\\
			B\parens*{t},D_j(t), & t\geq t_2+1.
		\end{cases}
	\end{equation*}
	where we denote a set by an ordered list and will use this convention throughout the proof.

	We can restrict $t_u=t_1$ and $t_w\geq t_2+1$ in \Cref{thm:classical-security}.
	This is because if $t_u\neq t_1$ and $A\parens*{t_u}\neq A\parens*{t_1}$, or if $t_w\leq t_2$,
	then there is no information flow from $A\parens*{t_u}$ to $\Obs\parens*{w_j,t_w}$ and thus $A\parens*{t_u}\CI\Obs\parens*{w_j,t_w}$.
	%In this case, the secret information of $u$ never leaks to any $w_j$.
	Moreover, since the access matrix $\Macc$ (taking values in $M_0,M_1,M_2$) is always symmetric for all $w_j$,
	we can only prove for the case $j=1$ without loss of generality.
	%without loss of generality, our goal can be reduced to proving
	Now proving \Cref{thm:classical-security} reduces to proving
	\begin{equation}
		I\parens*{A\parens*{t_1}; B\parens*{t_w},D_1\parens*{t_w}} \leq 2^{-\parens*{n-7}/2}
		\label{eq:classical-goal}
	\end{equation}
	for any $t_w\geq t_2+1$.
	
	We identify and define all random variables of concern in our proof as follows.
	\begin{itemize}
		\item 
			$A(t_1)$ stores the secret information written by $u$ into $A$.
		\item
			Let $X_j:=C_j^1(t_{v,j})$,
			where $C_j^1$ denotes the first bit of $C_j$.
			Let $\Lambda_j:=C_j^{\overline{1}}(t_{v,j}),L\bracks*{w_j}\parens*{t_{v,j}}$,
			where $C_j^{\overline{1}}$ denotes the remaining bits (except for the first bit) of $C_j$.
		\item
			$B(t_{v})$ stores the information written by $v$ into $B$, which also encodes the secret information of $u$. 
		\item
			For each $j\in [n]$, let 
			\begin{equation*}
				Y_j:= \abs*{\braces*{t\in [t_v+1,t_2-1]:\alpha(t)=(w_j, \texttt{flip},B)}}\bmod 2
			\end{equation*}
			denote the parity of the number of \texttt{flip} exercised by $w_j$ on $B$ for $t\in [t_v+1,t_2-1]$.
		\item
			$B(t_2)$ is obtained from $B\parens*{t_v}$ after each $w_j$ exercises a number of \texttt{flip}.
		\item
			$B(t_w)$ and $D_j(t_w)$ contain all information accessible to $w_j$ at time $t_w\geq t_2+1$.
	\end{itemize}

	For convenience, we also define the following notations:
	\begin{itemize}
		\item 
			Let $X:=X_1, \ldots, X_n$. %The same convention applies to $Y$ and $\Lambda$.
		\item
			Let $X_{\overline{j}}:=X_1,\ldots, X_{j-1},X_{j+1},\ldots,X_n$ for $j\in [n]$.
		\item
			Let $X'=X_{\overline{1}}$.
	\end{itemize}
	The above notations apply when $X$ is replaced by $Y$ or $\Lambda$.

	By the program $P_v$ of $v$ described in~\Cref{fig:userv},
	the change of access matrix $\Macc(t)$ in~\Cref{fig:acc-matrix-1,fig:acc-matrix-2,fig:acc-matrix-3},
	and the temporal ordering of requests,
	the relations between concerned random variables can be summarised in the probabilistic graphical model in~\Cref{fig:graphical}.

	\begin{figure*}
		\centering
		\includegraphics[width=0.6\textwidth]{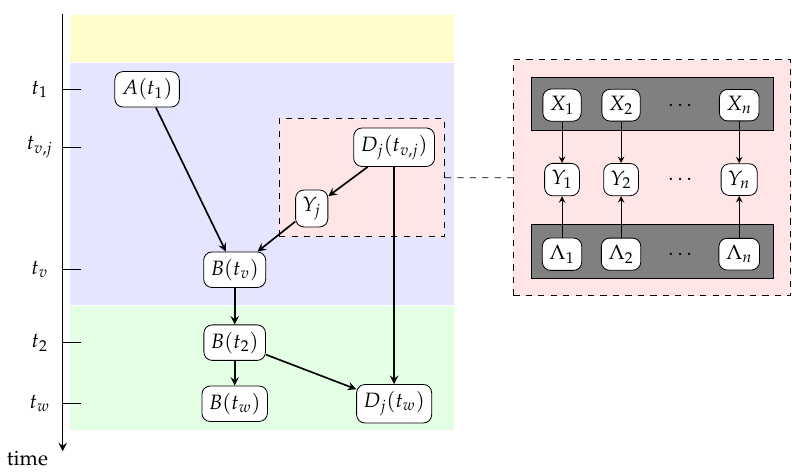}
		\caption{
			Probabilistic graphical model of concerned random variables in the system described in~\Cref{sub:problem_setting}.
                As usual, a directed edge represents a causal relation, and a bidirected edge represents a mutual dependence.
			The LHS depicts the relations between $A\parens*{t_1},D_j\parens*{t_{v,j}},Y_j,B\parens*{t_v},B\parens*{t_2},B\parens*{t_w},D_j\parens*{t_w}$.
			The RHS depicts the relations between $X,Y,\Lambda$, where nodes in each gray area (e.g., $X_1,\ldots,X_n$) are fully connected (by bidirected edges).
		}
		\label{fig:graphical}
	\end{figure*}

	Let us also fix some $a,b\in \braces*{0,1}$ and integer $d$.
	From~\Cref{fig:graphical}, we can decompose the joint probability distribution of these concerned random variables as
	\begin{align}
		\label{eq:joint-original}
        \begin{split}
		&\Pr[A(t_1)=a,B(t_v)=b_1, B(t_2)=b,X=x,\\
        &\qquad\qquad\qquad\qquad\qquad \qquad \Lambda=\lambda,Y=y, D_1(t_w)=d]
        \end{split}\\
		\begin{split}
			=&\Pr\bracks*{A(t_1)=a}\Pr\bracks*{B(t_v)=b_1\midv A(t_1)=a,X=x}\\
            &\quad \Pr\bracks*{X=x,Y=y, \Lambda=\lambda}\Pr\bracks*{B(t_2)=b\midv B(t_v)=b_1,Y=y}\\
			&\qquad\Pr\bracks*{D_1(t_w)=d\midv B(t_2)=b,X=x,\Lambda=\lambda}.
		\end{split}
		\label{eq:joint-pr}
	\end{align}
	Additionally, from~\Cref{fig:graphical}, the following conditional independence relations hold:
	\begin{itemize}
		\item 
			$D_1(t_w)\CI X_{\overline{1}},\Lambda_{\overline{1}}\mid B(t_2),X_1,\Lambda_1$; 
		\item
			$X\CI \Lambda$; and
		\item
			For any $j\in [n]$,
			\begin{equation}
				Y_j\CI X_{\overline{j}},Y_{\overline{j}},\Lambda_{\overline{j}} \mid X_j,\Lambda_j\text{\ \ and\ \ }
				Y_{\overline{j}} \CI X_j,Y_j,\Lambda_j \mid X_{\overline{j}}, \Lambda_{\overline{j}}.
				\label{eq:CI-XY}
			\end{equation}
			%where $X_{\overline{j}}=X_1,\ldots,X_{j-1},X_{j+1},\ldots, X_n$ and $Y_{\overline{j}}, \Lambda_{\overline{j}}$ are similarly defined.
	\end{itemize}
	By the fixed program $P_v$ of user $v$ in \Cref{fig:userv}, we further have:
	\begin{itemize}
		\item 
			$\Pr\bracks*{X=x}=\frac{1}{2^{n-1}}$ for $x\in \braces*{0,1}^n$ with $\abs*{x}\bmod 2 = 0$.
		\item
			$B(t_v)=A(t_1)\oplus \parens*{\frac{\abs*{X}}{2}\bmod 2}$.
			As a result,
            \begin{equation*}
                \Pr\bracks*{B(t_v)=b_1\midv A(t_1)=a,X=x}\neq 0
            \end{equation*}
		  iff $\parens*{-1}^{b_1}=\parens*{-1}^{a+\abs*{x}}$.
		\item 
			$B(t_v)=B(t_2)\oplus \bigoplus_j Y_j$.
			As a result,
            \begin{equation*}
                \Pr\bracks*{B(t_2)=b \midv B(t_v)=b_1,Y=y}\neq 0
            \end{equation*}
			iff $\parens*{-1}^{\abs*{y}}=\parens*{-1}^{b+b_1}$.

			%$\Pr\bracks*{B(t_2)=b\midv \parens*{-1}^{B(t_2)}=\parens*{-1}^{a+\sum_j x_j/2 +y_j},Y=y}$
			%is non-zero only when $\parens*{-1}^{\sum_j x_j/2+y_j}=\parens*{-1}^{a+b}$.
	\end{itemize}

	Combining the above observations,
	\Cref{eq:joint-pr} can be simplified as
	\begin{equation}
		\label{eq:joint-simp}
		\begin{split}
			&\Pr\bracks*{A(t_1)=a}\Id\bracks*{\parens*{-1}^{b_1}=\parens*{-1}^{a+\abs*{x}}} \Pr\bracks*{X=x,Y=y, \Lambda=\lambda}\\
			&\quad\Id\bracks*{\parens*{-1}^{\abs*{y}}=\parens*{-1}^{b+b_1}}\Pr\bracks*{D_1(t_w)=d\midv B(t_2)=b,X_1=x_1,\Lambda_1=\lambda_1}.
		\end{split}
	\end{equation}

	To prove our goal in \Cref{eq:classical-goal}, let us start with calculating the quantity
	\begin{align}
		&\Pr\bracks*{A(t_1)=a, B(t_w)=b,D_1(t_w)=d}\\
        \begin{split}
            =&\sum_{x,y,\lambda,b_1}\Pr\big[A(t_1)=a,B(t_v)=b_1,B(t_w)=b,X=x,\\
            &\qquad\qquad\qquad\qquad\qquad\qquad\Lambda=\lambda, Y=y,D_1(t_w)=d\big]
        \end{split}\\
		\begin{split}
			=&\Pr\bracks*{A(t_1)=a}\sum_{x,y,\lambda}\Id\bracks*{\parens*{-1}^{\abs*{x}/2+\abs*{y}+a+b}=1}\Pr\bracks*{X=x,Y=y, \Lambda=\lambda}\\
			 &\qquad\qquad\qquad\Pr\bracks*{D_1(t_w)=d\midv B(t_2)=b,X_1=x_1,\Lambda_1=\lambda_1},
		\end{split}
		\label{eq:ABC-sum}
	\end{align}
	where in the last equality we replace the joint probability distribution \Cref{eq:joint-original} by \Cref{eq:joint-simp}.
	%Let $X':=X_2,\ldots ,X_n$,
	%$Y':=Y_2,\ldots ,Y_n$
	%and $\Lambda':=\Lambda_2\ldots \Lambda_n$.
	%Moreover, as
	%\begin{align*}
	%	&\Pr\bracks*{X=x}\Pr\bracks*{\Lambda=\lambda}\Pr\bracks*{Y'=y'\midv X'=x',\Lambda'=\lambda'}\Pr\bracks*{Y_1=y_1\midv X_1=x_1,\Lambda_1=\lambda_1}\\
	%	\begin{split}
	%		=&\Pr\bracks*{X'=x'\midv X_1=x_1}\Pr\bracks*{\Lambda'=\lambda'\midv \Lambda_1=\lambda_1}\Pr\bracks*{Y'=y'\midv X'=x',\Lambda'=\lambda'}\\
	%		 &\qquad\qquad \Pr\bracks*{X_1=x_1}\Pr\bracks*{\Lambda_1=\lambda_1}\Pr\bracks*{Y_1=y_1\midv X_1=x_1,\Lambda_1=\lambda_1}
	%	\end{split}
	%	\\
	%	=&\Pr\bracks*{X'=x',Y'=y',\Lambda'=\lambda'\midv X_1=x_1,\Lambda_1=\lambda_1}\Pr\bracks*{X_1=x_1,Y_1=y_1,\Lambda_1=\lambda_1}
	%\end{align*}

	Using the conditions in \Cref{eq:CI-XY} gives the term
	\begin{align*}
		&\Pr\bracks*{X=x,Y=y,\Lambda=\lambda}\\
        \begin{split}
            =&\Pr\bracks*{X'=x',Y'=y',\Lambda'=\lambda'\midv X_1=x_1,Y_1=y_1,\Lambda_1=\lambda_1}\\
		&\qquad\qquad \qquad\qquad\qquad\qquad\Pr\bracks*{X_1=x_1,Y_1=y_1,\Lambda_1=\lambda_1}
        \end{split}
		\\
        \begin{split}
            =&\Pr\bracks*{X'=x',Y'=y',\Lambda'=\lambda'\midv X_1=x_1,\Lambda_1=\lambda_1}\\
		&\qquad\qquad\qquad\qquad\qquad\qquad\Pr\bracks*{X_1=x_1,Y_1=y_1,\Lambda_1=\lambda_1}.
        \end{split}
	\end{align*}
	Consequently, \Cref{eq:ABC-sum} can be rewritten as 
	\begin{align}
		\begin{split}
		&\Pr\bracks*{A(t_1)=a}\sum_{x,y,\lambda}\Id\bracks*{\parens*{-1}^{\abs*{x}/2+\abs*{y}+a+b}=1}\\
        &\qquad\Pr\bracks*{X_1=x_1,Y_1=y_1,\Lambda_1=\lambda_1}\\
        &\qquad\qquad\Pr\bracks*{X'=x',Y'=y',\lambda'=\lambda'\midv X_1=x_1,\Lambda_1=\lambda_1}\\
		&\qquad\qquad\qquad\Pr\bracks*{D_1(t_w)=d\midv B(t_2)=b, X_1=x_1,\Lambda_1=\lambda_1}.
		\end{split}
		%\begin{split}
		%	=&\Pr\bracks*{A(t_1)=a}\sum_{x_1,y_1,\lambda_1}\sum_{x',y',\lambda'}\Id\bracks*{\parens*{-1}^{\abs*{x}/2+\abs*{y}+a+b}=1}\Pr\bracks*{X'=x',Y'=y',\Lambda'=\lambda'\midv X_1=x_1,\Lambda_1=\lambda_1}\\
		%	 &\qquad\qquad\Pr\bracks*{X_1=x_1,Y_1=y_1,\Lambda_1=\lambda_1}\Pr\bracks*{C_1(t_w)=c\midv B(t_2)=b, X_1=x_1,\Lambda_1=\lambda_1}
		%\end{split}
		\label{eq:split-X1}
	\end{align}
	For convenience, let us define
	\begin{align}
        \begin{split}
            &f_a(x_1,y_1,\lambda_1):=
		\sum_{x',y',\lambda'}\Id\bracks*{\parens*{-1}^{\abs*{x}/2+\abs*{y} +a+b}=1}\\
        &\qquad\qquad\Pr\bracks*{X'=x',Y'=y',\Lambda'=\lambda'\midv X_1=x_1,\Lambda_1=\lambda_1},\label{eq:f-def}
        \end{split}
		\\
		&g:=
			\frac{4\Pr\bracks*{D_1(t_w)=d, B(t_2)=b}}
		{\sum_{x_1,\lambda_1}\Pr\bracks*{\Lambda_1=\lambda_1}\Pr\bracks*{D_1(t_w)=d\midv B(t_2)=b, X_1=x_1,\Lambda_1=\lambda_1}}.
		\label{eq:g-def}
	\end{align}
	Then, using the technical \Cref{lmm:f-bound,lmm:g-bound},
	we can rewrite \Cref{eq:split-X1} as
	\begin{align}
		\begin{split}
			&\Pr\bracks*{A(t_1)=a}\sum_{x_1,y_1,\lambda_1} f_a(x_1,y_1,\lambda_1)\Pr\bracks*{X_1=x_1,Y_1=y_1,\Lambda_1=\lambda_1}\\
			&\qquad\qquad\qquad\Pr\bracks*{D_1(t_w)=d\midv B(t_2)=b,X_1=x_1,\Lambda_1=\lambda_1}
		\end{split}
		\label{eq:ABC-rewrite}\\
		\begin{split}
			=& \Pr\bracks*{A\parens*{t_1}=a}\parens*{\frac{1}{2}+\delta}
			 \sum_{x_1,\lambda_1}\Pr\bracks*{X_1=x_1}\Pr\bracks*{\Lambda_1=\lambda_1}\\
			 &\qquad\qquad
			 \Pr\bracks*{D_1(t_w)=d\midv B(t_2)=b,X_1=x_1,\Lambda_1=\lambda_1}
		\end{split}
		\label{eq:replace-f}
		\\
		=& \Pr\bracks*{A\parens*{t_1}=a}\parens*{1+2\delta}\Pr\bracks*{D_1\parens*{t_w}=d,B\parens*{t_2}=b} g^{-1} \label{eq:plug-in-g}\\
		=& \parens*{1+2\delta}\parens*{1+\epsilon}^{-1}\Pr\bracks*{A\parens*{t_1}=a}\Pr\bracks*{D_1\parens*{t_w}=d,B\parens*{t_2}=b}
		\label{eq:replace-g}
	\end{align}
	for some $\abs*{\delta}\leq 2^{-(n-1)/2}$ and $\abs*{\epsilon}\leq 2^{-(n-3)/2}$.
	Here, \Cref{eq:replace-f} comes from \Cref{lmm:f-bound} and $X_1\CI \Lambda_1$;
	\Cref{eq:plug-in-g} comes from $\Pr\bracks*{X_1=x_1}=\frac{1}{2}$;
	and \Cref{eq:replace-g} comes from \Cref{lmm:g-bound}.

	All the above together yield $\Pr\bracks*{A(t_1)=a,B(t_w)=b,D_1(t_w)=d}=\Cref{eq:replace-g}$.
	Now we are ready to compute
	\begin{align*}
		&\frac{\Pr\bracks*{A(t_1)=a\midv B(t_w)=b,D_1(t_w)=d}}{\Pr\bracks*{A(t_1)=a}}\\
		=&\frac{\Pr\bracks*{A(t_1)=a,B(t_w)=b,D_1(t_w)=d}}{\Pr\bracks*{A(t_1)=a}\Pr\bracks*{B(t_w)=b,D_1(t_w)=d}}\\
		=&\parens*{1+2\delta}\parens*{1+\epsilon}^{-1},
	\end{align*}
	which can be upper bounded by $\frac{1+2^{-(n-3)/2}}{1-2^{-(n-3)/2}}\leq 1+2^{-(n-7)/2}$.
	Finally, using the inequality $\log\parens*{1+z}\leq z$ and the definition of mutual information leads to \Cref{eq:classical-goal}.
\end{proof}

In the following are two technical lemmas used in the proof of \Cref{thm:classical-security}.
Intuitively, \Cref{lmm:f-bound} says $f_a(x_1,y_1,\lambda_1)$ is close to $1/2$,
and \Cref{lmm:g-bound} says $g$ is close to $1$.

\begin{lemma}
	\label{lmm:f-bound}
	Let $f_a(x_1,y_1,\lambda_1)$ be defined as in \Cref{eq:f-def}.
	Then, for any $x_1,y_1\in \braces*{0,1}$, 
	\begin{equation}
		\abs*{f_a(x_1,y_1,\lambda_1)-\frac{1}{2}}\leq 2^{-(n-1)/2}.
		\label{eq:f_a-bound}
	\end{equation}
\end{lemma}

\begin{proof}
	Using \Cref{eq:CI-XY} and $X\CI \Lambda$, we have
	\begin{align*}
				&\Pr\bracks*{X'=x', Y'=y',\Lambda'=\lambda'\midv X_1=x_1,\Lambda_1=\lambda_1}\\
		=&\Pr\bracks*{Y'=y'\midv X'=x',\Lambda'=\lambda'}\\
        &\qquad\qquad\qquad\Pr\bracks*{X'=x'\midv X_1=x_1}\Pr\bracks*{\Lambda'=\lambda'\midv \Lambda_1=\lambda_1}.
	\end{align*}
	Let $X'',\Lambda''$ be random variables such that
        \begin{align*}
    	\Pr\bracks*{X''=x'}&=\Pr\bracks*{X'=x'\midv X_1=x_1}\\
    	\Pr\bracks*{\Lambda''=\lambda'}&=\Pr\bracks*{\Lambda'=\lambda'\midv \Lambda_1=\lambda_1}.
        \end{align*}
	It is easy to see that the probability distribution of $X''$ is uniform over the set 
	\begin{equation*}
		\braces*{x'\in \braces*{0,1}^{n-1}: \abs*{x'}\bmod 2 =x_1}.
	\end{equation*}
	In this case, we can rewrite \Cref{eq:f-def} as
	\begin{equation}
		\label{eq:f_a-rewrite}
		f_a(x_1,y_1,\lambda_1)=\Pr_{\Lambda''}\bracks*{\parens*{-1}^{\abs*{X''}/2+\abs*{Y'}+a+b+x_1/2+y_1}=1},
	\end{equation}
	where we use the subscript $\Lambda''$ to indicate this hidden random variable.
	%where $X'$ is uniformly chosen from $\overline{\calX}$.

	Let us write $X''=X_2''\ldots X_n''$
	and the same convention applies to $\Lambda''$.
	Similar to \Cref{eq:CI-XY},
	we have $Y_j\CI X_{\overline{j}}'',Y_{\overline{j}}'',\Lambda_{\overline{j}}''\mid X_j'',\Lambda_j''$,
	and consequently
	\begin{align*}
		\Pr\bracks*{Y'=y'\midv X''=x',\Lambda''=\lambda'}
				&=\prod_{j\geq 2}\Pr\bracks*{Y_j=y_j\mid X''=x',\Lambda''=\lambda'}\\
				&=\prod_{j\geq 2}\Pr\bracks*{Y_j=y_j\midv X_j''=x_j,\Lambda_j''=\lambda_j}.
	\end{align*}

	Hence, the conditions in \Cref{lmm:Mermin} are satisfies.
	By Mermin inequality in \Cref{lmm:Mermin}, we have 
	\begin{equation*}
		\abs*{\E\bracks*{\parens*{-1}^{\abs*{X''}/2+\abs*{Y'}+x_1/2}}}\leq 2^{-(n-1)/2+1}.
	\end{equation*}
	Note that 
	\begin{align*}
		&\Pr\bracks*{\parens*{-1}^{\abs*{X''}/2+\abs*{Y'}+x_1/2}=1}-\Pr\bracks*{\parens*{-1}^{\abs*{X''}/2+\abs*{Y'}+x_1/2}=-1}\\
        =&\E\bracks*{\parens*{-1}^{\abs*{X''}/2+\abs*{Y'}+x_1/2}}
        \end{align*}
        and
        \begin{equation*}
        \Pr\bracks*{\parens*{-1}^{\abs*{X''}/2+\abs*{Y'}+x_1/2}=1}+\Pr\bracks*{\parens*{-1}^{\abs*{X''}/2+\abs*{Y'}+x_1/2}=-1}
        =1.
        \end{equation*}
	Therefore, we can derive for any $a\in \braces*{0,1}$:
	\begin{equation}
		\abs*{\Pr\bracks*{\parens*{-1}^{\abs*{X''}/2+\abs*{Y'}+x_1/2+a+b+y_1}=1}-\frac{1}{2}}\leq 2^{-(n-1)/2},
		\label{eq:Mermin-pr}
	\end{equation}
	and \Cref{eq:f_a-bound} immediately follows from \Cref{eq:f_a-rewrite}.
	%we have $\abs*{f_a(x_1,y_1)-\frac{1}{2}}\leq 2^{-(n-1)/2}$.
\end{proof}

\begin{lemma}
	\label{lmm:g-bound}
	Let $g$ be defined as in \Cref{eq:g-def}.
	Then, we have
	\begin{equation}
		\abs*{g-1}\leq 2^{-(n-3)/2}.
		\label{eq:random-B}
	\end{equation}
\end{lemma}

\begin{proof}
	Note that 
	\begin{align*}
				&\Pr\bracks*{B(t_2)=b\midv X_1=x_1,\Lambda_1=\lambda_1}\\
		=&\Pr\bracks*{\parens*{-1}^{\abs*{X}/2+\abs*{Y}+A(t_1)+b}=1\midv X_1=x_1,\Lambda_1=\lambda_1}\\
        \begin{split}
            =&\sum_{a',y_1}\Pr\Big[\parens*{-1}^{\abs*{X'}/2+\abs*{Y'}+a'+b+x_1/2+y_1}=1\Big\vert\\ &\qquad\qquad A(t_1)=a', Y_1=y_1, X_1=x_1,\Lambda_1=\lambda_1\Big]\\
            &\qquad\qquad\qquad\Pr\bracks*{A(t_1)=a',Y_1=y_1\midv X_1=x_1,\Lambda_1=\lambda_1}
        \end{split}\\  
		=&\sum_{a',y_1} f_{a'}(x_1,y_1,\lambda_1) \Pr\bracks*{A(t_1)=a',Y_1=y_1\midv X_1=x_1,\Lambda_1=\lambda_1}.
	\end{align*}
	Thus, using \Cref{lmm:f-bound},
	we have 
        \begin{equation*}
            \abs*{\Pr\bracks*{B(t_2)=b\midv X_1=x_1,\Lambda_1=\lambda_1}-\frac{1}{2}}\leq 2^{-(n-1)/2}.
        \end{equation*}

	Next, by $X_1\CI \Lambda_1$, we can write
	\begin{align*}
		&\Pr\bracks*{D_1(t_w)=d, B(t_2)=b}\\
		%=&\sum_{x_1}\Pr\bracks*{C_1(t_w)=c,B(t_2)=b\midv X_1=x_1}\Pr\bracks*{X_1=x_1}\\
		=&\sum_{x_1,\lambda_1}\Pr\bracks*{D_1(t_w)=d\midv B(t_2)=b, X_1=x_1,\Lambda_1=\lambda_1}\\
		 &\qquad\Pr\bracks*{B(t_2)=b\midv X_1=x_1,\Lambda_1=\lambda_1}\Pr\bracks*{X_1=x_1}\Pr\bracks*{\Lambda_1=\lambda_1}.
	\end{align*}
	Combining the above with $\Pr\bracks*{X_1=x_1}=\frac{1}{2}$ and the definition of $g$ in \Cref{eq:g-def},
	our goal \Cref{eq:random-B} easily follows.
\end{proof}

\subsection{Proof of \Cref{lmm:Mermin}}
\label{sub:proof_of_Mermin}

In this subsection we provide a proof of the variant of Mermin inequality in \Cref{lmm:Mermin} for completeness.
The proof idea is almost the same as the one in~\cite{Mermin90}.

\begin{proof}[Proof of \Cref{lmm:Mermin}]
	First note that
	\begin{equation*}
		\E\bracks*{(-1)^{\abs*{X}/2 +\abs*{Y}+b/2}}=\sum_{\lambda} \Pr\bracks*{\Lambda=\lambda} \E\bracks*{\parens*{-1}^{\abs*{X}/2+\abs*{Y}+b/2}\midv \Lambda=\lambda}.
	\end{equation*}
	To prove the target inequality \Cref{eq:Mermin-ex}, it suffices to prove that
	\begin{equation}
		\abs*{\E\bracks*{\parens*{-1}^{\abs*{X}/2+\abs*{Y}+b/2}\midv \Lambda=\lambda}}\leq 2^{-n/2+1}.
		\label{eq:mermin-suffice}
	\end{equation}

	%Note that for any $k\in [n]$,
	%$\Ex\bracks*{(-1)^{Y_k}\midv X_k=x_k}=\Pr\bracks*{Y_k=0\midv X_k=x_k}-\Pr\bracks*{Y_k=1\midv X_k=x_k}$.
	%Let $\overline{x}=x_2\ldots x_k$.
	Let us consider the quantity
	\begin{equation}
		F_b:=
		\begin{cases}
            \begin{aligned}
                &\reop\Big(\prod_{j\in [n]} \big(\Ex\bracks*{(-1)^{Y_j}\midv X_j=0,\Lambda_j=\lambda_j}+\\
                &\qquad\qquad i\Ex\bracks*{(-1)^{Y_j}\midv X_j=1,\Lambda_j=\lambda_j}\big)\Big),
            \end{aligned}
			& b=0,\\
            \begin{aligned}
                &-\imop\Big(\prod_{j\in [n]} \big(\Ex\bracks*{(-1)^{Y_j}\midv X_j=0,\Lambda_j=\lambda_j}+\\
                &\qquad\qquad\Ex\bracks*{(-1)^{Y_j}\midv X_j=1,\Lambda_j=\lambda_j}\big)\Big),
            \end{aligned}
			& b=1.\\
		\end{cases}
		\label{eq:F-form-1}
	\end{equation}

	Since each term $\E\bracks*{\parens*{-1}^{Y_j}\midv X_j=a,\Lambda_j=\lambda_j}\in [-1,1]$,
	it is easy to see that
	\begin{equation}
		\abs*{F_b}\leq \parens*{\sqrt{2}}^n=2^{n/2}.
		\label{eq:F-bound}
	\end{equation}

	On the other hand, by calculation, we obtain
	\begin{equation}
		F_b=\sum_{x\in \calX_b}\parens*{-1}^{\abs*{x}/2+b/2}\prod_{j}\Ex\bracks*{\parens*{-1}^{Y_j}\midv X_j=x_j,\Lambda_j=\lambda_j}.
		\label{eq:F-b}
	\end{equation}
	Since $\E\bracks*{(-1)^{Y_j}\midv X_j=x_j,\Lambda_j=\lambda_j}=\Pr\bracks*{Y_j=0\midv X_j=x_j,\Lambda_j=\lambda_j}-\Pr\bracks*{Y_j=1\midv X_j=x_j,\Lambda_j=\lambda_j}$, we further have
	\begin{align}
		\Cref{eq:F-b} &=\sum_{x\in \calX_b}\parens*{-1}^{\abs*{x}/2+b/2}\sum_{y\in \calY} \parens*{-1}^{\abs*{y}}\Pr\bracks*{Y=y\midv X=x,\Lambda=\lambda} \\
					  &=\sum_{x\in \calX_b}\Ex\bracks*{\parens*{-1}^{\abs*{x}/2+ \abs*{Y}+b/2}\midv X=x,\Lambda=\lambda}.
					  \label{eq:F-form-2}
	\end{align}
	As $\Lambda$ is independent of $X$,
	$\Pr\bracks*{X=x\midv \Lambda=\lambda}=\Pr\bracks*{X=x}=\frac{1}{2^{n-1}}$ for any $x\in \calX$.
	Consequently, 
	\begin{align*}
		\Cref{eq:F-form-2}&=2^{n-1}\sum_{x\in \calX_b}\Ex\bracks*{\parens*{-1}^{\abs*{x}/2+\abs*{Y}+b/2}\midv X=x,\Lambda=\lambda}\Pr\bracks*{X=x\midv \Lambda=\lambda}\\ %\label{eq:F-form-2}
						  &=2^{n-1}\E\bracks*{\parens*{-1}^{\abs*{X}/2+\abs*{Y}+b/2}\midv \Lambda=\lambda}.
	\end{align*}
	Finally, combining the above with \Cref{eq:F-bound} yields \Cref{eq:mermin-suffice}.
\end{proof}

\subsection{Proof of Lemmas about Flexibility}
\label{sub:proof_of_lemmas_about_flexibility}

In this subsection, we present detailed proofs of several lemmas about flexibility in \Cref{sub:comparison_of_flexibility}.

\subsubsection{Proof of \Cref{thm:flex-hierarchy}}
\label{sub:proof_of_flex_hierarchy}

Let us first prove the remaining parts of \Cref{thm:flex-hierarchy},
which can be broken into the following two lemmas for $\mathsf{M}=\GRP$ and $\mathsf{M}=\ENT$, respectively.
%Let us prove them one by one.

%\begin{lemma}
%	\label{lmm:flex-subsystem}
%	For any $k\geq 2$, $\SUBSYS^{k-1}<\SUBSYS^{k}$.
%\end{lemma}

\begin{lemma}
	\label{lmm:flex-grp}
	For any $k\geq 2$, $\GRP^{k-1}<\GRP^k$.
\end{lemma}

\begin{proof}
	\qquad
	\begin{enumerate}
		\item 
			We first prove that $\GRP^k \not\leq \GRP^{k-1}$.
			The proof idea is by noticing that
			a system in $\GRP^k$ can assign all quantum registers into $k$ groups,
			while a system in $\GRP^{k-1}$ can only assign them into $k-1$ groups.
			Using the pigeonhole principle, 
			there will be two quantum registers that belong to different groups in the former system,
			and to the same group in the latter system.
			Then, intuitively, we can show that latter system authorise strictly more requests than the former.

			Let us consider a system $\calA=\parens*{\Sub,\Obj,\Rt, \Atr,\Rule}\in \GRP^k$,
			where $\Rtc=\emptyset$, $\Rtq=\braces*{\texttt{CNOT}}$,
			$\Sub=\braces*{u,v}$, $\Objc=\emptyset$, and $\Objq=\braces*{X_1,\ldots, X_{2k}}$.
			Here, $\texttt{CNOT}$ means the ability to perform a $\mathit{CNOT}$ gate.
			%Denote subsystem $q=\Objq$.
			Attributes $\Mcc,\Mq, G\in \Atr$ are initialised as follows.
			Since $\Objc=\emptyset$, we set $\Mcc=\emptyset$.
			For $s\in \Sub, o\in \Objq$:
			\begin{equation*}
				\Mq\bracks*{s,o}=
				\begin{cases}
					\braces*{\texttt{CNOT}}, & s=u,\\
					0, & o.w.
				\end{cases}
			\end{equation*}
			Let $G\bracks*{X_j}=\floor*{\parens*{j+1}/2}$ for $j\in [2k]$.

			Assume for contradiction that there exists another system $\calA'=\parens*{\Sub,\Obj',\Rt', \Atr',\Rule'}\in \GRP^{k-1}$ 
			with $\Mcc',\Mq', G'\in \Atr$ such that $\calA\simeq \calA'$. % Our goal is to prove $\calA\not\simeq \calA'$.
			We can further assume that $\Objc'=\emptyset$ and $\Rtc'=\emptyset$,
			because otherwise $\calA$ and $\calA'$ will be obviously inequivalent.
			Using similar reasoning to that in the proof of $\SUBSYS^{k}\not\leq\SUBSYS^{k-1}$ in \Cref{thm:flex-hierarchy},
			we can also restrict that $\Rtq'=\braces*{\texttt{CNOT}}$.

			Consider an execution $\parens*{S,P}$
			with 
            \begin{equation*}
                P_u\equiv \mathbf{for}\ l\in [k]\ \mathbf{do}\ \mathit{CNOT}\bracks*{X_{2l-1},X_{2l}}\ \mathbf{od}
            \end{equation*}
			and $P_v\equiv \bot$.
			By our construction of $\calA$,
			the history generated by $\parens*{S,P}$
			in $\calA$ is $\parens*{u,\braces*{X_{1},X_{2}}, \texttt{CNOT}},\ldots, \parens*{u,\parens*{X_{2k-1},X_{2k}},\texttt{CNOT}}$
			and authorised according to \Cref{def:group-control}.
			Since we assume $\calA\simeq\calA'$, the history generated by $\parens*{S,P}$ in $\calA'$
			is also authorised, which implies $\texttt{CNOT}\in \Mq'\bracks*{u,X_j}$ for $j\in [2k]$.

			Observe that by the pigeonhole principle,
			there must exist distinct $j_1,j_2,j_3\in [2k]$ such that 
			$G'\bracks*{X_{j_1}}=G'\bracks*{X_{j_2}}=G'\bracks*{X_{j_3}}$
			and $G\bracks*{X_{j_1}}\neq G\bracks*{X_{j_2}}$.

			Now consider another execution $\parens*{S,P'}$ with
			$P_u'\equiv \mathit{CNOT}\bracks*{X_{j_1},X_{j_2}}$ and $P_v'\equiv \bot$.
			The histories generated by $\parens*{S,P'}$ in $\calA$ and $\calA'$
			are the same $\parens*{u, \braces*{X_{j_1}, X_{j_2}},\texttt{CNOT}}$.
			By our construction of $\calA$, this history is unauthorised in $\calA$ as $G\bracks*{X_{j_1}}\neq G\bracks*{X_{j_2}}$ 
			(see the authorisation rule in \Cref{def:group-control}).
			However, it is authorised in $\calA'$ because $G'\bracks*{X_{j_1}}= G'\bracks*{X_{j_2}}$.
			Hence, we obtain a contradiction and the conclusion follows.
		\item
			Next we prove that $\GRP^{k-1}< \GRP^k$.

			Suppose that $\calA=\parens*{\Sub, \Obj,\Rt,\Atr, \Rule}\in \GRP^{k-1}$ with $\Mcc,\Mq,G\in \Atr$.
			Then, we can define another system $\calA'=\parens*{\Sub,\Obj,\Rt, \Atr',\Rule}\in \GRP^{k}$ with $\Mcc',\Mq',G'\in \Atr$,
			such that $\Mcc'=\Mcc$, $\Mq'=\Mq$, 
			and $G'\bracks*{o}=G\bracks*{o}$ for any $o\in \Objq$.
			It is easy to see that $\calA\simeq \calA'$ in this case.
	\end{enumerate}
\end{proof}

\begin{lemma}
	\label{lmm:flex-ent}
	$\ENT^1< \ENT^2$.
\end{lemma}

\begin{proof}
	\qquad
	\begin{enumerate}
		\item 
			First we prove $\ENT^2\not\leq \ENT^1$.
			The proof idea is by observing that 
			the attribute $\Me$ of a system in $\ENT^2$ records whether two quantum registers can be entangled,
			while $\Me$ of a system in $\ENT^1$ only records whether a quantum register can be entangled with others.
			Therefore, a system in $\ENT^2$ has a more fine-grained control of entanglement than a system in $\ENT^1$.

			Let us consider a system $\calA=\parens*{\Sub,\Obj, \Rt,\Atr,\Rule}\in \ENT^2$,
			where $\Rtc=\emptyset$, $\Rtq=\braces*{\texttt{CNOT},\texttt{measure}}$,
			$\Sub=\braces*{u,v}$, $\Objc=\emptyset$, and $\Objq=\braces*{X_1,X_2,X_3,X_4}$.
			Here, $\texttt{measure}$ means the ability to perform a complete measurement.
			Attributes $\Mcc,\Mq,\Me,D\in \Atr$ are initialised as follows.
			As $\Objc=\emptyset$, we set $\Mcc=\emptyset$.
			%\begin{equation*}
			%	\Mcc\bracks*{s,o} =
			%	\begin{cases}
			%		\braces*{\texttt{read},\texttt{write}}, & s=u,\\
			%		\emptyset, & o.w.
			%	\end{cases}
			%\end{equation*}
			%for $s\in \Sub, o\in \Objc$,
			%and 
			For $s\in \Sub, o\in \Objq$:
			$\Mq\bracks*{s,o}=\braces*{\texttt{CNOT},\texttt{measure}}$
			and $D\bracks*{o}=\mathit{true}$.
			For $s\in \Sub,o\in \calP_2\parens*{\Objq}$:
			\begin{equation*}
				\Me\bracks*{o}=
				\begin{cases}
					\mathit{true}, & o=\braces*{X_1,X_2}\vee o=\braces*{X_3,X_4},\\
					\mathit{false}, & o.w.
				\end{cases}
			\end{equation*}

			Assume for contradiction that there exists another system $\calA'=\parens*{\Sub,\Obj',\Rt',\Atr',\Rule'}\in \ENT^1$
			with $\Mcc',\Mq',\Me',D'\in \Atr'$ such that $\calA\simeq$
			We can further assume that $\Objc'=\emptyset$ and $\Rtc'=\emptyset$,
			because otherwise $\calA$ and $\calA'$ will be obviously inequivalent.
			Using similar reasoning to that in the proof of $\SUBSYS^{k}\not\leq\SUBSYS^{k-1}$ in \Cref{thm:flex-hierarchy},
			we can also restrict that $\Rtq'=\braces*{\texttt{CNOT},\texttt{measure}}$.

			Consider an execution $\parens*{S,P}$ with
            \begin{equation*}
                P_u\equiv \mathit{CNOT}\bracks*{X_1,X_2}; \mathit{CNOT}\bracks*{X_3,X_4}
            \end{equation*}
			and $P_v\equiv \bot$. 
			By our construction of $\calA$,
			the history generated by $\parens*{S,P}$ in $\calA$ is 
			$\parens*{u,\braces*{X_1,X_2},\texttt{CNOT}},\parens*{u,\braces*{X_3,X_4},\texttt{CNOT}}$
			and authorised according to \Cref{def:2-ent-control}.

			On the other hand, since we assume $\calA\simeq \calA'$,
			the history generated by $\parens*{S,P}$ in $\calA'$ is also authorised.
			By \Cref{def:1-ent-control}, this implies that $\Me'\bracks*{o}=\mathit{true}$ for $o\in \braces*{X_1,X_2,X_3,X_4}$
			(meaning any quantum register in $\calA'$ can be entangled with others)
			and $\texttt{CNOT}\in \Mq'\bracks*{u,X_1}\cap \Mq'\bracks*{u,X_3}$.
			%Moreover, there are two cases:
			%either $\texttt{CNOT}\in \Mq\bracks*{u,X_1},\Mq\bracks*{u,X_3}$; or $\Rt'$ contains some rights to performs
			%quantum circuits $U_1,\ldots, U_m$ such that $U_m\ldots U_1=\mathit{CNOT}$. 
			%The latter case is similar to the one in the proof of \Cref{lmm:flex-subsystem},
			%and will yields $\calA\simeq\calA'$. It suffices to consider the former case.

			Now consider another execution $\parens*{S,P'}$
			with $P_u'\equiv \mathit{CNOT}\bracks*{X_1,X_3}$ and $P_v'\equiv \bot$.
			The histories generated by $\parens*{S,P'}$ in $\calA$ and $\calA'$ 
			are the same $\parens*{u,\braces*{X_1,X_3},\texttt{CNOT}}$.
			By our construction of $\calA$, this history is unauthorised in $\calA$ 
			because $\Me\bracks*{X_1,X_3}=\mathit{false}$ (see the authorisation rule in \Cref{def:2-ent-control}).
			However, it is authorised in $\calA'$ due to $\Me'\bracks*{X_1}=\Me'\bracks*{X_2}=\mathit{true}$
			and $\texttt{CNOT}\in \Mq'\bracks*{u,X_1}\cap \Mq'\bracks*{u,X_3}$ (see the authorisation rule in \Cref{def:1-ent-control}).
			Hence, we obtain a contradiction and the conclusion follows.
		\item
			Next we prove $\ENT^1\leq \ENT^2$.
			Consider a system $\calA=\parens*{\Sub,\Obj,\Rt,\Atr,\Rule}\in \ENT^1$ with $\Mcc,\Mq,\Me,D\in \Atr$.
			Then, we can define $\calA'=\parens*{\Sub,\Obj,\Rt,\Atr',\Rule'}\in \ENT^2$ with $\Mcc',\Mq',\Me',D'\in \Atr'$
			such that $\Mcc'=\Mcc$, $\Mq'=\Mq$, and for any $o_1\neq o_2\in \Objq$: 
			$\Me'\bracks*{o_1,o_2}=\Me\bracks*{o_1}\wedge \Me\bracks*{o_2}$
			and $D'\bracks*{o_1,o_2}=D\bracks*{o_1}\vee D\bracks*{o_2}$.
			It is easy to see that $\calA\simeq\calA'$.
	\end{enumerate}
\end{proof}

\subsubsection{Proof of \Cref{thm:flex-compare}~\Cref{thm:subsys-not-leq-grp-ent,thm:grp-not-leq-ent-subsys}}
\label{sub:proof_of_flex_compare}

Now we prove \Cref{thm:subsys-not-leq-grp-ent,thm:grp-not-leq-ent-subsys} of \Cref{thm:flex-compare}.
First, \Cref{thm:subsys-not-leq-grp-ent} in \Cref{thm:flex-compare} can be restated as the following lemma.

\begin{lemma}
	\label{lmm:subsys-comp-grp-ent}
	$\SUBSYS\not\leq \GRP,\ENT$.
\end{lemma}

\begin{proof}
	\quad
	\begin{itemize}
		\item 
			We first prove that $\SUBSYS\not\leq \GRP$.
			The proof idea is similar to that for proving $\ENT^2\not\leq \ENT^1$. 
			Intuitively, the attribute $\Mq$ in a system in $\SUBSYS$ records 
			information about subsystems which consists of multiple quantum registers,
			while the attributes $\Mq, G$ in a system in $\GRP$ only records
			information about each individual quantum register.
			In some cases, the former provides a more fine-grained control of quantum operations than the latter.

			Let us consider a system $\calA=\parens*{\Sub, \Obj, \Rt, \Atr, \Rule}\in \SUBSYS$, 
			where $\Rtc=\emptyset$, $\Rtq=\braces*{\texttt{CNOT}}$,
			$\Sub= \braces*{u,v}$, $\Objc=\emptyset$, and $\Objq=\braces*{X_1,X_2,X_3}$.
			Attributes $\Mcc,\Mq\in \Atr$ are initialised as follows.
			As $\Objc=\emptyset$, we set $\Mcc=\emptyset$.
			For $s\in \Sub,o\subseteq \Objq$.
			\begin{equation*}
				\Mq\bracks*{s,o}= 
				\begin{cases}
					\braces*{\texttt{CNOT}}, & s=u \wedge \parens*{o= \braces*{X_1,X_2}\vee o =\braces*{X_2,X_3}},\\
					\emptyset, & o.w.
				\end{cases}
				%\label{eq:sub-minus-grp}
			\end{equation*}

			Assume for contradiction that there exists another system 
			$\calA'=\parens*{\Sub,\Obj',\Rt',\Atr',\Rule'}\in \GRP$
			with $\Mcc',\Mq',G'\in \Atr'$ such that $\calA'\simeq \calA$.
			We can further assume that $\Objc'=\emptyset$ and $\Rtc'=\emptyset$,
			because otherwise $\calA$ and $\calA'$ will be obviously inequivalent.
			Using similar reasoning to that in the proof of $\SUBSYS^{k}\not\leq\SUBSYS^{k-1}$ in \Cref{thm:flex-hierarchy},
			we can also restrict that $\Rtq'=\braces*{\texttt{CNOT}}$.

			Consider an execution $\parens*{S,P}$
			with 
            \begin{equation*}
                P_u\equiv \mathit{CNOT}\bracks*{X_1,X_2}; \mathit{CNOT}\bracks*{X_2,X_3}
            \end{equation*}
			and $P_v\equiv \bot$.
			By our construction of $\calA$,
			the history generated by $\parens*{S,P}$
			in $\calA$ is $\parens*{u,\braces*{X_1,X_2},\texttt{CNOT}},\braces*{u,\braces*{X_2,X_3}, \texttt{CNOT}}$ and authorised.
			Since we assume $\calA\simeq \calA'$,
			the history generated by $\parens*{S,P}$ in $\calA'$
			is also be authorised.
			By the authorisation rule in \Cref{def:group-control}, 
			this implies that $\texttt{CNOT}\in \Mq'\bracks*{X_1}\cap \Mq'\bracks*{X_3}$ and $G'\bracks*{X_1}=G'\bracks*{X_2}=G'\bracks*{X_3}$.
			%Moreover, there are two cases:
			%either $\texttt{CNOT}\in\Mq\bracks*{u,X_1}, \Mq\bracks*{u,X_3}$;
			%or $\Rt'$ contains some rights to perform quantum circuits $U_1,\ldots, U_m$
			%such that $U_m\ldots U_1=\mathit{CNOT}$.
			%The latter case is similar to the one in the proof of \Cref{lmm:flex-subsystem},
			%and will yield $\calA'\simeq \calA$,
			%so it suffices to consider the former case.

			Now we consider another execution $\parens*{S,P'}$,
			with $P_u'\equiv \mathit{CNOT}\bracks*{X_1,X_3}$ and $P_v'\equiv \bot$.
			The histories generated by $\parens*{S,P'}$ in $\calA$ and $\calA'$ 
			are the same $\parens*{u,\braces*{X_1,X_3}, \texttt{CNOT}}$.
			By our construction of $\calA$, this history is unauthorised in $\calA$ as $\texttt{CNOT}\notin \Mq\bracks*{u,\braces*{X_1,X_3}}$ (see the authorisation rule in \Cref{def:subsys-control}).
			However, it is authorised in $\calA'$ because $\texttt{CNOT}\in \Mq'\bracks*{X_1}\cap \Mq'\bracks*{X_3}$ and $G'\bracks*{X_1}=G'\bracks*{X_3}$ (see the authorisation rule in \Cref{def:group-control}).
			Hence, we obtain a contradiction and the conclusion follows.
		\item
			Next we prove that $\SUBSYS \not\leq \ENT$.
			The proof idea is essentially using the difference between control of quantum operations and control of entanglement.
			In particular, a system in $\SUBSYS$ controls whether a quantum operation is authorised or unauthorised,
			and does not force a measurement to be applied before modifying $\Mq$.
			In contrast, a system in $\ENT$ controls whether entanglement is allowed to exist between quantum registers,
			and a measurement has to be applied if there is no promise of disentanglement,
			before we modify $\Me$ to disentangle two objects.

			Specifically, let us prove $\SUBSYS \not\leq \ENT^2$.
			Consider a system $\calA=\parens*{\Sub, \Obj, \Rt, \Atr, \Rule}\in \SUBSYS$, 
			where $\Rtc=\braces*{\texttt{read},\texttt{write}}$, $\Rtq=\braces*{\texttt{H},\texttt{CNOT},\texttt{measure}}$,
			$\Sub=\braces*{u,v}$, $\Objc=\braces*{\Mq}$, and $\Objq=\braces*{X_1,X_2,X_3}$.
			Here, $\Objc=\braces*{\Mq}$ implies that $\Mq$ can be dynamically modified.
			Attributes $\Mcc,\Mq\in \Atr$ are initialised as follows.
			For $s\in \Sub, o\in \Objc$:
			\begin{equation*}
				%\label{eq:sub-minus-ent-mc}
				\Mcc\bracks*{s,o} =
				\begin{cases}
					\braces*{\texttt{read},\texttt{write}}, & s=u,\\
					\emptyset, & o.w.
				\end{cases}
			\end{equation*}
			For $s\in \Sub, o\subseteq \Objq$:
			$\Mq\bracks*{s,o}=\braces*{\texttt{H},\texttt{CNOT}}$.
			%\begin{equation*}
			%	%\label{eq:sub-minus-ent-mq}
			%	\Mq\bracks*{s,o} = 
			%	\begin{cases}
			%		%\braces*{\texttt{H}}, & o = \braces*{X_1},\\
			%		\braces*{\texttt{CNOT}}, & o= \braces*{X_1,X_2}\vee o= \braces*{X_2,X_3},\\
			%		\emptyset, & o.w.
			%	\end{cases}
			%\end{equation*}

			Assume for contradiction that there exists another system $\calA'=\parens*{\Sub,\Obj',\Rt',\Atr',\Rule'}\in \ENT^2$
			with $\Obj'=\Objc'\cup \Objq$ and $\Mcc',\Mq',\Me',D'\in \Atr'$ such that $\calA'\simeq \calA$.
			Using similar reasoning to that in the proof of $\SUBSYS^{k}\not\leq\SUBSYS^{k-1}$ in \Cref{thm:flex-hierarchy},
			we can further restrict that $\Rtq'=\braces*{\texttt{H},\texttt{CNOT},\texttt{measure}}$.

			Consider an execution $\parens*{S,P}$.
			Here, the scheduler $S$ is defined by
			$S\parens*{\alpha\parens*{0},\ldots, \alpha\parens*{t-1}}=s\parens*{t}$,
			where $s\parens*{0}=s\parens*{1}=s\parens*{2}=u$ and $s\parens*{3}=s\parens*{4}=v$.
			The program $P$ is defined by
			\begin{equation*}
				P_u\equiv H\bracks*{X_1};\mathit{CNOT}\bracks*{X_1,X_2}; \mathit{forbid}\parens*{v,\braces*{X_1,X_2}}
			\end{equation*}
			and $P_v\equiv \bot$,
			where $\mathit{forbid}\parens*{v,\braces*{X_1,X_2}}$ 
			means to modify attributes such that future request $\parens*{v,\braces*{X_1,X_2},r}$
			will be unauthorised for any right $r$.
			By our construction of $\calA$,
			the history $\alpha$ generated by $\parens*{S,P}$
			in $\calA$ is 
			\begin{equation*}
                \begin{split}
                    &\parens*{u, \braces*{X_1}, \texttt{H}}, \parens*{u,\braces*{X_1,X_2},\texttt{CNOT}},\\
                    &\qquad\braces*{u,\Mq\bracks*{v,\braces*{X_1,X_2}},\texttt{read}}, \parens*{u, \Mq\bracks*{v,\braces*{X_1,X_2}},\texttt{write}}
                \end{split}
			\end{equation*}
			and authorised.
			%However, it is not authorised in $\calA'$, no matter how we choose $\Mq',\Me',E'$.
			
			On the other hand, since we assume $\calA'\simeq \calA$,
			the history $\alpha'$ generated by $\parens*{S, P}$ in $\calA'$ is also authorised.
			Note that according to the authorisation rule in~\Cref{def:2-ent-control},
			whether the future request $\parens*{v,\braces*{X_1,X_2},r}$ will be authorised in $\calA'$ is determined 
			by the attributes $\Me'\bracks*{X_1,X_2}$, $\Mq'\bracks*{v,X_1}$ and $\Mq'\bracks*{v,X_2}$.
			As $P_u$ contains $\mathit{forbid}\parens*{v,\braces*{X_1,X_2}}$,
			the above implies the following two cases:
			\begin{itemize}
				\item 
					Either there exists $t_3\in\N$ such that $\alpha'\parens*{t_3}=\parens*{u,\Me'\bracks*{X_1,X_2},r}$ 
					for some right $r$ that modifies (e.g., \texttt{write}) $\Me'\bracks*{X_1,X_2}$ to $\mathit{false}$.
					In this case, after $\mathit{CNOT}\bracks*{X_1,X_2}$ in $P_u$ is executed,
					the quantum state of $X_1,X_2$ becomes 
                    \begin{equation*}
                        \frac{1}{\sqrt{2}}\parens*{\ket{0}_{X_1}\ket{0}_{X_2}+\ket{1}_{X_1}\ket{1}_{X_2}},
                    \end{equation*}
					%Since $u$ performs $\mathit{CNOT}\bracks*{X_1,X_2}$ at the beginning,
					and we also have $D'\bracks*{X_1,X_2}=\mathit{false}$ at some time $t_1<t_3$
					(meaning that $X_1,X_2$ are not promised to be disentangled),
					due to the post-update rule in \Cref{def:2-ent-control}.
					Moreover, according to~\Cref{def:2-ent-control},
					this implies that there exists $t_2\in (t_1,t_3)$ with $\alpha'\parens*{t_2}=\parens*{s,X,\texttt{measure}}$
					for some $s\in \Sub$ and $X\in \braces*{X_1,X_2}$.
					However, such $\alpha'$ cannot be generated from program $P$,
					because $P$ does not contain measurement.
				\item
					Or there exists $t\in \N$ such that $\alpha'\parens*{t}=\parens*{u,\Mq'\bracks*{v,X},r}$
					for some right $r$ that modifies (e.g., \texttt{write}) $\Mq'\bracks*{v,X}$ and removes $\texttt{CNOT}$ from it,
					where $X\in \braces*{X_1,X_2}$.
					In this case, after the final request of $\alpha'$, we have $\texttt{CNOT}\notin\Mq'\bracks*{v,X}$.

					Now consider another execution $\parens*{S,P'}$ with $P_u'= P_u$
					and $P_v'\equiv \mathit{CNOT}\bracks*{X_2,X_3}; \mathit{CNOT}\bracks*{X_1,X_3}$.
					By the construction of the scheduler $S$,
					the history generated by $\parens*{S,P'}$ in $\calA$ is $\beta=\alpha,\parens*{v,\braces*{X_2,X_3},\texttt{CNOT}},\parens*{v,\braces*{X_1,X_3},\texttt{CNOT}}$,
					where $\alpha$ is the history generated by $\parens*{S,P}$ in $\calA$ previously.
					Since after $\alpha$, we still have $\texttt{CNOT}\in \Mq\bracks*{v,\braces*{X_1,X_3}}\cap \Mq\bracks*{v,\braces*{X_2,X_3}}$,
					the history $\beta$ is authorised in $\calA$ (see the authorisation rule in \Cref{def:subsys-control}).
					Similarly, the history generated by $\parens*{S,P'}$ in $\calA'$ is
                    \begin{equation*}
                        \beta'=\alpha',\parens*{v,\braces*{X_2,X_3},\texttt{CNOT}},\parens*{v,\braces*{X_1,X_3},\texttt{CNOT}},
                    \end{equation*}
					where $\alpha'$ is the history generated by $\parens*{S,P}$ in $\calA'$ previously.
					However, $\beta'$ is unauthorised in $\calA'$,
					because after $\alpha'$, we have $\texttt{CNOT}\notin\Mq'\bracks*{v,X}$ for some $X\in \braces*{X_1,X_2}$ 
					(see the authorisation rule in \Cref{def:2-ent-control}).
			\end{itemize}
			In either case, we obtain a contradiction and the conclusion follows.
	\end{itemize}
\end{proof}

Second, \Cref{thm:grp-not-leq-ent-subsys} in \Cref{thm:flex-compare} can be restated as the following lemma.

\begin{lemma}
	\label{lmm:grp-comp-ent-subsys}
	$\GRP \not\leq \ENT,\SUBSYS^{<N}$ and $\GRP\leq \SUBSYS$,
\end{lemma}

\begin{proof}
	\quad
	\begin{itemize}
		\item 
			We first prove that $\GRP\not\leq \SUBSYS^{< N}$.
			The proof idea is by noticing that a system in $\SUBSYS^{<N}$ only allows
			quantum operations on subsystem of size $<\abs*{\Objq}$,
			while a system in $\GRP$ can allow quantum operations on subsystem of size $\abs*{\Objq}$.

			Let us consider a system $\calA =\parens*{\Sub, \Obj, \Rt,\Atr, \Rule}\in \GRP$,
			where $\Rtc=\emptyset$, $\Rtq=\braces*{\texttt{QFT}_k}$,
			$\Sub=\braces*{u,v}$, $\Objc=\emptyset$, and $\Objq=\braces*{X_1,\ldots, X_k}$.
			Here, $\texttt{QFT}_k$ means the ability to perform a quantum Fourier transform circuit $\mathit{QFT}_k$ on $k$ qubits.
			Attributes $\Mcc,\Mq,G\in \Atr$ are initialised as follows.
			Since $\Objc=\emptyset$, we set $\Mcc=\emptyset$.
			For any $s\in \Sub, o\in \Objq$:
			\begin{equation*}
				%\label{eq:grp-minus-subsys-mq}
				\Mq\bracks*{s,o}=
				\begin{cases}
					\texttt{QFT}_k, & s=u,\\
					\emptyset, & o.w.
				\end{cases}
				%\label{eq:grp-minus-subsys-g}
			\end{equation*}
			and $G\bracks*{o}=1$.

			Consider another system $\calA'=\parens*{\Sub,\Obj,\Rt',\Atr',\Rule'}\in \SUBSYS^{<N}$
			with $\Mcc',\Mq'\in \Atr'$.
			Suppose that $\calA' \in \SUBSYS^{k'}$ for some $k'<\abs*{\Objq}=k$.

			Consider an execution $\parens*{S,P}$ with $P_u\equiv \mathit{QFT}_k$ and $P_v\equiv \bot$.
			By our construction of $\calA$,
			the history generated by $\parens*{S,P}$ in $\calA$ is 
			simply $\parens*{u,\braces*{X_1,\ldots, X_k},\texttt{QFT}_k}$ and authorised.
			Now consider the history $\alpha$ generated by $\parens*{S,P}$ in $\calA'$.
			There are two cases:
			either $\alpha$ contains multiple requests like in the proof of $\SUBSYS^k \not\leq \SUBSYS^{k-1}$ in~\Cref{thm:flex-hierarchy},
			and then we can derive $\calA\not\simeq \calA'$;
			or $\alpha=\parens*{u,\braces*{X_1,\ldots,X_k},\texttt{QFT}_k}$, which is unauthorised
			due to the authorisation rule in~\Cref{def:subsys-control} and $k=\abs*{\Objq}>k'$.
			In either case, $\calA\not\simeq \calA'$ and the conclusion follows.
		\item
			Next we prove $\GRP \not\leq \ENT$.
			Like in the proof of $\SUBSYS\not\leq \ENT$ in \Cref{lmm:subsys-comp-grp-ent},
			the idea is essentially using the difference between control of quantum operations and control of entanglement.
			Specifically, let us prove $\GRP \not\leq \ENT^2$.

			Consider a system $\calA=\parens*{\Sub, \Obj, \Rt, \Atr,\Rule}\in \GRP$,
			where $\Rtc=\braces*{\texttt{read},\texttt{write}}$, $\Rtq=\braces*{\texttt{H},\texttt{CNOT},\texttt{measure}}$,
			$\Sub=\braces*{u,v}$, $\Objc=\braces*{G}$, and $\Objq=\braces*{X_1,X_2,X_3,X_4}$.
			Note that $\Objc=\braces*{G}$ means $G$ can be dynamically modified. 
			Attributes $\Mcc,\Mq,G\in \Atr$ are initialised as follows.
			For $s\in \Sub, o\in \Objc$:
			\begin{equation}
				\Mcc\bracks*{s,o}= 
				\begin{cases}
					\braces*{\texttt{read},\texttt{write}}, & s=u,\\
					\emptyset,  & o.w.
				\end{cases}
			\end{equation}
			For $s\in \Sub,o\in \Objq$:
			$\Mq\bracks*{s,o}=\braces*{\texttt{H},\texttt{CNOT}}$.
			Let $G\bracks*{X_1}=G\bracks*{X_2}=G\bracks*{X_3}=1$ and $G\bracks*{X_4}=2$.

			Assume for contradiction that there exists another system $\calA'=\parens*{\Sub, \Obj', \Rt', \Atr', \Rule'}\in \ENT^2$
			with $\Obj'=\Objc'\cup \Objq$ and $\Mcc',\Mq',\Me',D'\in \Atr'$ such that $\calA'\simeq \calA$.
			Using similar reasoning to that in the proof of $\SUBSYS^{k}\not\leq\SUBSYS^{k-1}$ in \Cref{thm:flex-hierarchy},
			we can further restrict that $\Rtq'=\braces*{\texttt{H},\texttt{CNOT},\texttt{measure}}$.

			Consider an execution $\parens*{S,P}$.
			Here, the scheduler $S$ is defined by
			$S\parens*{\alpha\parens*{0},\ldots, \alpha\parens*{t-1}}=s\parens*{t}$,
			where $s\parens*{0}=s\parens*{1}=s\parens*{2}=u$ and $s\parens*{3}=s\parens*{4}=v$.
			The program $P$ is defined by
			\begin{equation*}
				P_u\equiv H\bracks*{X_1};\mathit{CNOT}\bracks*{X_1,X_2}; \mathit{newgrp}\parens*{X_1,X_4}
			\end{equation*}
			and $P_v\equiv \bot$, where $\mathit{newgrp}\parens*{X_1}$ means to modify attributes such that
			$X_1,X_4$ are put into a new group and any quantum operation on $X_1$ or $X_4$ should act within this group;
			i.e., for future request $\parens*{s,o,r}$, if $o\cap \braces*{X_1,X_4}\neq \emptyset$,
			then $o\subseteq \braces*{X_1,X_4}$.
			By our construction of $\calA$,
			the history generated by $\parens*{S,P}$ in $\calA$ is 
			\begin{equation*}
                \begin{split}
                    &\parens*{u,\braces*{X_1},\texttt{H}},\parens*{u,\braces*{X_1,X_2},\texttt{CNOT}},\parens*{u,G\bracks*{X_1},\texttt{read}},\\
                    &\qquad\parens*{u,G\bracks*{X_1},\texttt{write}},\parens*{u,G\bracks*{X_4},\texttt{read}},\parens*{u,G\bracks*{X_4},\texttt{write}}
                \end{split}
			\end{equation*}
			and authorised.

			The remaining reasoning is similar to the proof of $\SUBSYS\not\leq \ENT^2$ in \Cref{lmm:subsys-comp-grp-ent}.
			Since we assume $\calA\simeq \calA'$,
			the history $\alpha'$ generated by $\parens*{S,P}$ in $\calA'$ is also authorised.
			This implies initially $\Mq'\bracks*{X_1,X_2}=\mathit{true}$.
			According to the authorisation rule in \Cref{def:group-control},
			whether the future request $\parens*{s,o,r}$ with $o\cap \braces*{X_1,X_4}\neq \emptyset$
			will be authorised is determined by the attributes $\Me'\bracks*{X_1,X}$ for $X\neq X_1$,
			$\Me'\bracks*{X_4,X}$ for $X\neq X_4$,
			$\Mq'\bracks*{X_1}$ and $\Mq'\bracks*{X_4}$.
			As $P_u$ contains $\mathit{newgrp}\parens*{X_1,X_4}$, 
			which forbids future request like $\parens*{v,\braces*{X_1,X_2},r}$,
			the above implies the following two cases:
			\begin{itemize}
				\item 
					Either there exists $t\in \N$ such that $\alpha'\parens*{t}=\parens*{u,\Me'\bracks*{X_1,X_2},r}$ for some right $r$
					that modifies $\Me'\bracks*{X_1,X_2}$ to $\mathit{false}$.
					In this case, using the same reasoning as in the proof of $\SUBSYS\not\leq \ENT^2$,
					we can derive $\calA\not\simeq \calA'$.
				\item
					Or there exists $t\in \N$ such that $\alpha'\parens*{t}=\parens*{u, \Mq'\bracks*{v,X},r}$ for some right $r$
					that modifies (e.g., $\texttt{write}$) $\Mq'\bracks*{v,X}$ and removes $\texttt{CNOT}$ from it,
					where $X\in \braces*{X_1,X_2}$.
					In this case, after the final request of $\alpha'$, we have $\texttt{CNOT}\notin \Mq'\bracks*{v,X}$.

					Now consider another execution $\parens*{S,P'}$ with $P_u'=P_u$
					and $P_v'\equiv \mathit{CNOT}\bracks*{X_1, X_4}; \mathit{CNOT}\bracks*{X_2,X_3}$.
					By the construction of the scheduler $S$, 
					the history generated by $\parens*{S,P'}$ in $\calA$ is
                    \begin{equation*}
                        \beta=\alpha,\parens*{v,\braces*{X_1, X_4},\texttt{CNOT}},\parens*{v, \braces*{X_2,X_3},\texttt{CNOT}}.
                    \end{equation*}
					Since after $\alpha$, we still have $\texttt{CNOT}\in \Mq\bracks*{v,Y}$ for $Y\in \braces*{X_1,X_2,X_3,X_4}$,
					$G\bracks*{X_1}=G\bracks*{X_4}$ and $G\bracks*{X_2}=G\bracks*{X_3}$,
					the history $\beta$ is authorised in $\calA$ (see the authorisation rule in \Cref{def:group-control}).
					Similarly, the history $\beta'$ generated by $\parens*{S,P'}$ in $\calA'$ is
                    \begin{equation*}
                        \beta'=\alpha', \parens*{v, \braces*{X_1,X_4},\texttt{CNOT}},\parens*{v,\braces*{X_2,X_3},\texttt{CNOT}}.
                    \end{equation*}
					However, $\beta'$ is unauthorised in $\calA'$,
					because after $\alpha'$, we have $\texttt{CNOT}\notin \Mq'\bracks*{v,X}$ for some $X\in \braces*{X_1,X_2}$
					(see the authorisation rule in \Cref{def:2-ent-control}), .
			\end{itemize}
			In either case, we obtain a contradiction and the conclusion follows.
		\item
			Finally we prove $\GRP\leq \SUBSYS$.
			Consider a system $\calA=\parens*{\Sub, \Obj, \Rt,\Atr,\Rule}\in \GRP$ with $\Mcc,\Mq,G\in \Atr$.
			Then, we can define $\calA'=\parens*{\Sub, \Obj,\Rt,\Atr',\Rule'}\in \SUBSYS$,
			where $\Mcc',\Mq'\in \Atr'$ are defined by $\Mcc'=\Mcc$, and
			for $s\in \Sub,o\subseteq \Objq$:
			\begin{equation*}
				\Mq'\bracks*{s,o} = 
				\begin{cases}
					\bigcap_{X\in o}\Mq\bracks*{s,X}, & \forall X,Y\in o, G\bracks*{X}=G\bracks*{Y},\\
					\emptyset, & o.w.
				\end{cases}
			\end{equation*}
			It is easy to see that $\calA\simeq \calA'$ from this construction.
	\end{itemize}
\end{proof}

\section{Background on Probabilistic Graphical Models}
\label{sec:background_on_probabilistic_graphical_models}

In this section, we briefly introduce the notations and tools
in probabilistic graphical models, used in \Cref{sub:proof_of_classical-security}.
The readers are referred to the textbook~\cite{Pearl00} for a more thorough introduction.

\paragraph{Probabilities}

For a random variable $A$, we use $\Pr\bracks*{A=a}\in [0,1]$ to denote the probability of $A$ taking the value $a$.
It holds that $\sum_a \Pr\bracks*{A=a}=1$.
The joint probability 
\begin{equation*}
    \Pr\bracks*{A=a,B=b}=\Pr\bracks*{A=a\cap B=b}
\end{equation*}
denotes the probability of $A$ taking the value $a$
and another random variable $B$ taking the value $b$.
For simplicity, sometimes we simply write $A,B$ for the random variable $(A,B)$.

%If $\Pr\bracks*{B=b}>0$, the conditional probability $\Pr\bracks*{A=a\midv B=b}=\frac{\Pr\bracks*{A=a,B=b}}{\Pr\bracks*{B=b}}$
Let the conditional probability $\Pr\bracks*{A=a\mid B=b}$
be the probability of $A$ taking the value $a$ given that $B$ takes the value $b$. 
The joint probability $\Pr\bracks*{A=a,B=b}$ can be calculated by
\begin{equation}
    \label{eq:pr-product}
    \Pr\bracks*{A=a,B=b}=\Pr\bracks*{A=a\mid B=b}\cdot \Pr\bracks*{B=b}.
\end{equation}
By summing over $a$, we have the following decomposition of $\Pr\bracks*{A=a}$ by conditioning on different $B=b$:
\begin{equation}
    \Pr\bracks*{A=a}=\sum_b \Pr\bracks*{A=a\mid B=b}\cdot\Pr\bracks*{B=b}.
\end{equation}

A useful generalisation of \Cref{eq:pr-product} is the following \textit{chain rule}:
\begin{equation}
    \label{eq:chain-rule}
    \begin{split}
        &\Pr\bracks*{A_1=a_1,\ldots, A_n=a_n}\\
        =&\Pr\bracks*{A_n=a_n\mid A_{n-1}=a_{n-1},\ldots, A_1=a_1}\\
        &\qquad\qquad\ldots \Pr\bracks*{A_2=a_2\mid A_1=a_1}\Pr\bracks*{A_1=a_1}.
    \end{split}
\end{equation}

If $\Pr\bracks*{A=a\mid B=b}=\Pr\bracks*{B=b}$ whenever $\Pr\bracks*{B=b}>0$,
then $A$ and $B$ are said to be \textit{independent}, denoted by $A\CI B$.
More generally, if
\begin{equation*}
    \Pr\bracks*{A=a\mid B=b, C=c}=\Pr\bracks*{B=b\mid C=c}
\end{equation*}
whenever $\Pr\bracks*{B=b,C=c}>0$,
then $A$ and $B$ are said to be \textit{conditionally independent} given $C$, denoted by $A\CI B\mid C$.
Intuitively, it means that if we know the value of $C$, we cannot learn extra information about $A$ from learning the value of $B$.
We mention two useful properties of conditional independence:
\begin{itemize}
    \item 
        (Symmetry) $A\CI B\mid C$ implies $B\CI A \mid C$.
    \item 
        (Decomposition) $A\CI B,C \mid D$ implies $A\CI B\mid D$.
\end{itemize}
The concept of conditional independence plays an important role in probabilistic graphical models.
%$A$ and $B$ are said to be independent, if $\Pr\bracks*{A=a,B=b}=\Pr\bracks*{A=a}\cdot \Pr\bracks*{B=b}$
%for any $a$ and $b$.

As usual, we use $\E\bracks*{A}=\sum_a a\cdot \Pr\bracks*{A=a}$ to denote the expectation of $A$.

\paragraph{Probabilistic Graphical Models}

A graph can be used to represent the relations between multiple random variables.
Each vertex represents a random variable.
A \textit{directed} edge between random variables represents a \textit{causal relation};
i.e., $A\rightarrow B$ means that $A$ can influence $B$.
A \textit{bidirected} edge between random variables represent a \textit{mutual dependence},
often due to an unobserved common cause;
i.e., $A \leftrightarrow B$ means that $A$ and $B$ are dependent.
For example, in \Cref{fig:graphical},
the directed edge $A(t_1)\rightarrow B(t_v)$ comes from that 
$B(t_v)$ is written by user $v$, who reads the secret $A(t_1)$ written by user $u$;
the undirected edges between $X_1,\ldots, X_n$ comes from that
the values of these $X_j$ are randomly drawn by user $v$ from a distribution (see \Cref{fig:userv}).
Recall that in \Cref{fig:graphical}, vertices within a gray area are fully connected by bidirected edges. 

%If the graph is directed and acyclic, then
Based on the graph structure,
the joint probability of random variables corresponding to all vertices
can be decomposed into a product of conditional probabilities,
In particular, we can refine the chain rule in \Cref{eq:chain-rule} to 
\begin{equation}
    \begin{split}
        &\Pr\bracks*{A_1=a_1,\ldots, A_n=a_n}\\
        =&\prod_j \Pr\bracks*{A_j=a_j\midv \forall k, (A_k\rightarrow A_j\wedge k<j) \Rightarrow (A_k=a_k) },
    \end{split}
\end{equation}
where $A_k\rightarrow A_j$ denotes a directed edge.
%where $\operatorname{Parents}(X):=\braces*{Y:Y\rightarrow X}$ 
%stands for the set of parents of $X$.
We have used this decomposition rule to decompose \Cref{eq:joint-original} into \Cref{eq:joint-pr}
in \Cref{sub:proof_of_classical-security}.

Moreover, the graph structure allow us to
conveniently infer the conditional independence of random variables.
Let $X,Y,Z$ be sets of random variables. 
$X$ is said to be $d$-separated from $Y$ by $Z$ if,
for any (undirected) path $p$ from a node $A\in X$ to a node $B\in Y$,
one of the following conditions hold:
\begin{itemize}
    \item 
        $p$ contains $C\rightarrow D \rightarrow E$ or $C\leftarrow D\rightarrow E$ with $D\in Z$; or
    \item 
        $p$ contains $C\rightarrow D\leftarrow E$ such that for any $F$, if $D\rightarrow^* F$, then $F\notin Z$.
\end{itemize}
If $X$ is $d$-\textit{separated} from $Y$ by $Z$,
then we have $X\CI Y \mid Z$.
In \Cref{sub:proof_of_classical-security}, 
we have derived several conditional independence relations from \Cref{fig:graphical}
using this notion.

%\section{Further Discussion}
%\label{sec:further_discussion}
%
%How to generalise \texttt{write} on an integer register $X$?
%There are several choices.
%For example, we can generalise it to the set of rights: $\braces*{\texttt{reset}_x:x\in \mathbf{Int}}$,
%%\braces*{\sum_y \ket{x}\!\bra{y}\parens*{\cdot}\ket{y}\!\bra{x}:x\in \mathbf{Int}}.
%where each $\texttt{reset}_x$ means applying the quantum operation $\sum_y \ket{x}\!\bra{y}\parens*{\cdot}\ket{y}\!\bra{x}$ on $X$.
%Another choice is to generalise it to the set of rights:
%$\braces*{\texttt{reset}_{\psi}:\ket{\psi}\in \calH_{\mathbf{Int}}}$,
%%\braces*{\sum_{y} \ket{\psi}\!\bra{y}\parens*{\cdot}\ket{y}\!\bra{\psi}:\ket{\psi}\in \calH_{\mathbf{Int}}},
%where each $\texttt{reset}_{\psi}$ means applying the quantum operation $\sum_y \ket{\psi}\!\bra{y}\parens*{\cdot}\ket{y}\!\bra{\psi}$ on $X$.
%We can also generalise it to \texttt{swap},
%which means swapping the value in $X$ and the value in the local memory of the subject executing this operation.
%The above three generalisations are essentially different.
%There are definitely other reasonable choices.

\end{document}